\pdfoutput=1
\documentclass[11pt]{article}
\usepackage[margin=1in]{geometry}
\usepackage[T1]{fontenc}
\usepackage[utf8]{inputenc}
\usepackage{lmodern}
\usepackage{microtype}
\usepackage{amsmath,amssymb,amsthm,mathtools,mathrsfs}
\usepackage{booktabs,multirow,array,tabularx}
\usepackage{graphicx,subcaption,float}
\usepackage{placeins}
\usepackage{enumitem}
\usepackage[round]{natbib}
\usepackage{xcolor}
\usepackage{url}
\usepackage{hyperref}
\hypersetup{colorlinks=true,linkcolor=blue,citecolor=blue,urlcolor=blue}
\graphicspath{{numerics/}}

\newtheorem{theorem}{Theorem}[section]
\newtheorem{proposition}[theorem]{Proposition}
\newtheorem{lemma}[theorem]{Lemma}
\newtheorem{corollary}[theorem]{Corollary}
\theoremstyle{definition}

\newtheorem{example}[theorem]{Example}

\numberwithin{equation}{section}

\newcommand{\EE}{\mathbb E}
\newcommand{\Pbb}{\mathbb P}
\newcommand{\R}{\mathbb R}
\newcommand{\Z}{\mathbb Z}
\newcommand{\eps}{\varepsilon}
\newcommand{\Ga}{\Gamma}
\newcommand{\bK}{\boldsymbol{K}}
\newcommand{\bGa}{\boldsymbol{\Gamma}}
\newcommand{\blfootnote}[1]{%
  \begingroup
  \renewcommand{\thefootnote}{}%
  \begin{NoHyper}%
  \footnotetext{#1}%
  \end{NoHyper}%
  \endgroup
}

\title{Unifying Privacy Accounting:\\ Information Equivalence and Information Loss}

\author{%
  Buxin Su\textsuperscript{1,*}
  \and Qiaoshi Yang\textsuperscript{1,*}
  \and Yiding Su\textsuperscript{2,*}
  \and Chendi Wang\textsuperscript{3,\(\dagger\)}%
}
\date{\today}

\begin{document}
\maketitle
\blfootnote{\textsuperscript{1}\,University of Pennsylvania; Emails:
\texttt{subuxin@sas.upenn.edu} and \texttt{yangqs@sas.upenn.edu}.}
\blfootnote{\textsuperscript{2}\,University of Warwick; Email:
\texttt{yiding.su@warwick.ac.uk}.}
\blfootnote{\textsuperscript{3}\,Xiamen University; Email:
\texttt{chendi.wang@xmu.edu.cn}.}
\blfootnote{\textsuperscript{*}\,Equal contribution.}
\blfootnote{\textsuperscript{\(\dagger\)}\,Corresponding author.}

\begin{abstract}
Differential privacy (DP) admits several notions, but the choice among them may affect both privacy analysis and utility. In this paper, we consider four mainstream curve-based privacy notions within a unified information-theoretic framework. For a fixed ordered pair of output distributions, we establish information equivalence among the two directional privacy profiles of $(\varepsilon,\delta)$-DP, the pair of hypothesis-testing trade-off functions, and the extended privacy-loss distribution. The exact R\'enyi differential privacy (RDP) curve joins this equivalence class whenever it is finite at some order greater than one. Under this mild condition, choosing among these notions changes only their semantic interpretation and computational requirements.
In contrast, taking the maximum of the directional privacy profiles or compressing the RDP curve into a single zero-concentrated differential privacy (zCDP) parameter can lose information. We quantify the information loss between the exact RDP curve and its zCDP bound for standard noise mechanisms. This gap is zero for Gaussian noise but generally positive for Gaussian-mixture, Laplace, discrete Gaussian, and Poisson-subsampled Gaussian mechanisms.
Moreover, this gap grows linearly with the number of independently composed mechanisms. Our information-theoretic perspective has practical consequences. At the same certified privacy level, retaining the full RDP curve rather than using zCDP reduces the required noise variance by up to $45\%$ for Gaussian-mixture noise in workloads comparable in size to the American Community Survey. For DP-SGD on Fashion-MNIST under Poisson subsampling, an RDP-based privacy accountant improves test accuracy by up to $8.73$ percentage points compared to a zCDP-based accountant when both are calibrated to the same $(\varepsilon,\delta)$ guarantee. These results clarify when different privacy notions have equivalent information (when the gap is zero or nearly zero) and when retaining more detailed privacy information improves utility (when the gap is large).
\end{abstract}

\section{Introduction}
\label{sec:intro}

\begin{figure}[t]
\centering
\includegraphics[width=\linewidth]{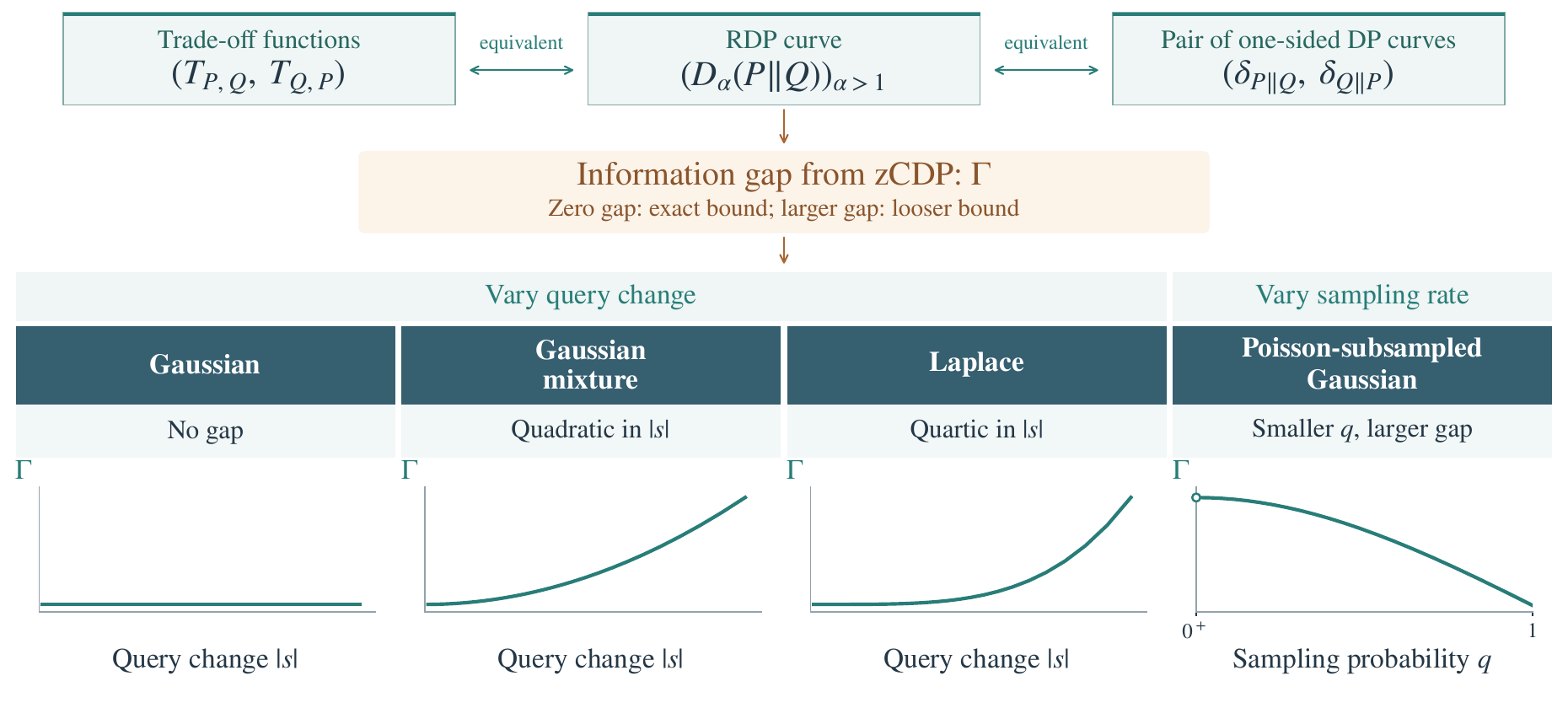}
\caption{Exact privacy curves and the zCDP--RDP gap.
Top: three of the four curve-based notions in Theorem~\ref{thm:equivalence}; the RDP curve joins them under a mild condition.
Bottom: how the information gap varies with the query change or the sampling probability for four mechanisms.
\textbf{Left three:} $s$ is the query change between neighboring datasets.
At fixed R\'enyi order $\alpha=2$, the Gaussian gap is zero, while the
symmetric Gaussian-mixture and Laplace gaps scale as $|s|^2$ and $|s|^4$,
respectively, as $|s|$ tends to $0$.
\textbf{Right:} $q$ is the sampling rate of a Poisson-subsampled Gaussian mechanism. At fixed $\alpha=2$, smaller $q>0$
yields a larger gap; $q=1$ corresponds to no subsampling, and the open
circle indicates the limit as $q\rightarrow0$.
All panels show trends in $\Ga(t)$, defined in~\eqref{eq:gap-definition}.
Detailed curves and parameter settings appear in
Figure~\ref{fig:common-mechanism-gaps}.}
\label{fig:privacy-overview}
\end{figure}

Differential privacy (DP) is a rigorous mathematical framework for
quantifying changes in the output distributions of a randomized
algorithm on two neighboring datasets that differ in at most one
element \citep{dwork2006differential,dwork2014algorithmic}.
For a fixed ordered pair of output distributions induced by a mechanism
on neighboring datasets, different privacy notions serve different purposes:
R\'enyi differential privacy (RDP) supports convenient composition
through moment accounting and is widely used for DP-SGD
\citep{mironov2017renyi,abadi2016deep,wang2019subsampled};
$f$-DP provides a hypothesis-testing trade-off function and a
tensor-product composition rule
\citep{Dong2022Gaussian,su2025statistical,su2025census};
privacy-loss distributions (PLDs) convolve under product composition
and support sharp numerical accounting
\citep{Koskela2020computing,gopi2021numerical,koskela2021tight};
and zero-concentrated differential privacy (zCDP) compresses the
RDP curve into a single parameter $\rho$, enabling simple scalar
composition \citep{bun2016concentrated}.
Meanwhile, $(\eps,\delta)$-DP remains the standard reporting format,
and its optimal composition admits an operational hypothesis-testing
interpretation \citep{kairouz2017composition}.

However, a single $(\varepsilon,\delta)$ pair gives only one operating
point of the full privacy profile curve. Such pairs also do not induce a
total ordering: for example, $(2,10^{-6})$ and $(1,10^{-5})$ are generally
incomparable. Recent work therefore studies mechanism comparison beyond
a single calibration point \citep{balle2020privacy,kifer2022bayesian,kaissis2024beyond,gomez2025reporting}.
Yet two mechanisms can share the same two-sided 
$(\varepsilon,\delta)$ profile, $\max(\delta_{P\|Q}(\varepsilon), \delta_{Q\|P}(\varepsilon))$, while having different exact RDP curves
(Section~\ref{sec:information-equivalence}).
Likewise, Gaussian and mixture noise can share the same optimal
zCDP parameter $\rho$ but have different exact RDP curves
(Section~\ref{sec:defect}).
This leads to two basic questions:
(1) Under what conditions do the four full privacy notions, with
directional information retained, contain the same privacy information
about a fixed pair of output distributions?
(2) How much information is lost when a full privacy curve is
compressed into a single zCDP parameter?

Answering these questions helps us decide which privacy notion to use.
The first question asks whether we can choose a privacy notion to suit
our needs without losing privacy information. 
Prior work establishes duality between complete $(\eps,\delta)$-curve
guarantees and symmetric $f$-DP \citep{Dong2022Gaussian}, as well as
optimal conversions from RDP bounds to $(\eps,\delta)$-DP and $f$-DP
\citep{asoodeh2021three,riess2026optimal}.
These conversions are optimal given the RDP bounds, but need not recover the exact privacy curve of a particular distribution pair.
In contrast, we ask whether the full exact RDP curve of a fixed ordered pair $(P,Q)$ uniquely determines its directional $(\eps,\delta)$ profiles, trade-off functions, and PLD.
The second question matters in practice: losing privacy information leads to more noise to satisfy the same privacy requirement and reduces accuracy. For example, an
$f$-DP re-accounting of the 2020 U.S.~Census \citep{abowd20222020}
showed that the discrete-Gaussian variance proxy $\sigma^2$ could be
reduced by $15.08\%$--$24.82\%$ while maintaining nearly the same certified
privacy guarantees \citep{su2025census,su2026sieve}.
In DP-SGD, using the full RDP curve rather than a zCDP summary
can also reduce the noise required during privacy calibration, as shown in 
Figure~\ref{fig:dpsgd-composition-sweep}.

\paragraph{Our Contributions.}
We answer the two questions above through the following contributions,
summarized in Figure~\ref{fig:privacy-overview}.
\begin{itemize}
\item \textbf{Information equivalence of curve-based privacy notions.}
We show that, for a fixed ordered pair $(P,Q)$, the pair of one-sided
$(\eps,\delta)$ profiles, the pair of trade-off functions, and
the extended PLD contain the same information, as does the
full exact RDP curve under mild conditions.
Our counterexamples show that taking the maximum of the one-sided profiles, $\max(\delta_{P\|Q}(\varepsilon), \delta_{Q\|P}(\varepsilon))$, can lose information.

\item \textbf{Information loss of zCDP for common mechanisms.}
We introduce the zCDP--RDP gap, the distance between the zCDP bound and the exact RDP curve, as a measure of the
information discarded when the curve is compressed into a single
parameter. The gap adds up under independent composition, and we expand
it for small query changes $s$ (Theorem~\ref{thm:sensitivity-gap}). It is
identically zero for Gaussian noise, of order $s^2$ for Gaussian-mixture
noise, and of order $s^4$ for Laplace noise. Under Poisson subsampling
with rate $q$ it grows as $q$ decreases and approaches the entire zCDP
bound as $q\to0$: zCDP discards the privacy amplification that
subsampling provides.

\item \textbf{Practical guidance for privacy accounting.}
Experiments on private data release and DP-SGD show that information
loss can require additional noise and reduce accuracy. Our results
support using the full RDP curve or a privacy notion with equivalent information
when tighter accounting is needed, and zCDP when the zCDP--RDP gap
is negligible.

\end{itemize}

\begin{figure}[!t]
\centering
\includegraphics[width=.82\linewidth]{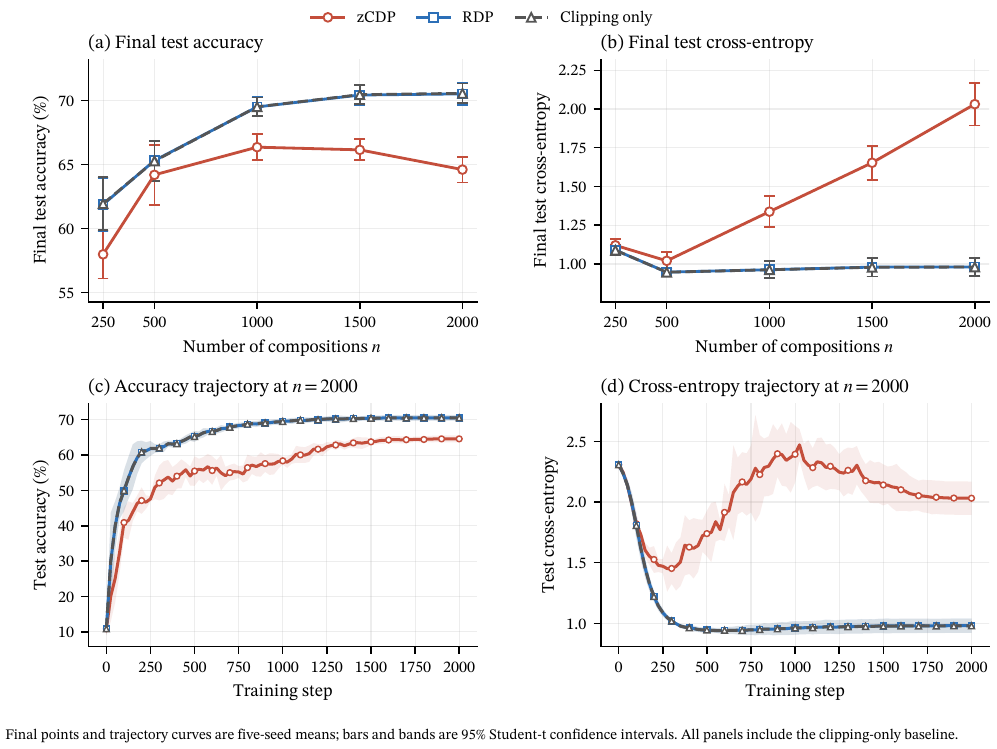}
\caption{DP-SGD on Fashion-MNIST with noise calibrated under RDP and under
zCDP to the same certified privacy $(\eps,\delta)=(3,10^{-5})$;
``clipping only'' adds no noise. \textup{(a)}--\textup{(b)}: final test
accuracy and cross-entropy as the number of training steps $n$ varies at
$q=0.01$, with the noise recalibrated for each $n$.
\textup{(c)}--\textup{(d)}: training histories for $n=2000$, using the noise
calibrated for $n=2000$ throughout. Curves are five-seed means; bars and
bands are 95\% confidence intervals. Setup in
Section~\ref{subsec:controlled-dpsgd}.}
\label{fig:dpsgd-composition-sweep}
\label{fig:dpsgd-composition-full}
\end{figure}

\paragraph{Related Work.}
Previous work connects the curve-based privacy notions through hypothesis
testing and transforms of privacy loss. The statistical view of differential
privacy uses hypothesis testing to compare privacy guarantees
\citep{su2025statistical}; convex duality connects one-sided
$(\eps,\delta)$-DP curves and trade-off functions
\citep{Dong2022Gaussian,balle2020hypothesis}.
Privacy-loss random variables give an equivalent description of a full
privacy curve \citep{gopi2021numerical}, and their characteristic functions
support conversions between one-sided curves, trade-off functions, and PLDs
\citep{zhu2022optimal}. The connection to RDP comes from exponential
moments: a finite privacy-loss CGF identifies the PLD
\citep{wang2019subsampled}, while Laplace-transform identities relate an
extended privacy curve and the RDP curve under absolute-continuity
conditions \citep{chourasia2025laplace}.

Theorem~\ref{thm:equivalence} brings these connections together for arbitrary
probability measures, including singular pairs. Finiteness at one real
order greater than one lets the full exact RDP curve recover the PLD,
the trade-off functions, and both one-sided curves. This recovery uses only
real-order RDP values and does not require mutual absolute continuity or
complex-order reports. The counterexamples in
Section~\ref{sec:information-equivalence} explain why the directional
information and the finite-order condition matter. This is a different
question from finding the best guarantee consistent with reported RDP
bounds \citep{asoodeh2021three,riess2026optimal}, as discussed above:
here the object to recover is the exact privacy of a fixed pair.
To our knowledge, prior work has not quantified RDP-to-zCDP information
loss through a gap at each order or studied its local expansion and
accumulation under independent composition.

\section{Information Equivalence between Exact Privacy Notions}
\label{sec:information-equivalence}

Fix two neighboring datasets and let $P$ and $Q$ be the output distributions
of the mechanism on them. For $\eps\geq0$, the \emph{one-sided
$(\eps,\delta)$-DP curve}
\[
 \delta_{P\|Q}(\eps):=\sup_A\{P(A)-e^\eps Q(A)\}
\]
is the smallest
$\delta$ such that $P(A)\leq e^\eps Q(A)+\delta$ for every event $A$;
exchanging $P$ and $Q$ gives $\delta_{Q\|P}$. The usual \emph{two-sided} curve
is $\delta(\eps)=\max\{\delta_{P\|Q}(\eps),\delta_{Q\|P}(\eps)\}$; we keep the
two one-sided curves separate.

For testing $P$ against $Q$, the trade-off function $T_{P,Q}(a)$ is the
smallest type II error among tests whose type I error is at most $a$:
\[
 T_{P,Q}(a):=\inf_{\phi}\{1-\EE_Q[\phi]:\EE_P[\phi]\leq a\},
 \qquad 0\leq a\leq1,
\]
where $\phi$ ranges over randomized tests with values in $[0,1]$.
A mechanism is $f$-DP if $T_{P,Q}\geq f$ for every ordered neighboring
pair \citep{Dong2022Gaussian}.
We write $L$ for the privacy loss, whose distribution under $P$ is the PLD.
For densities or probability mass functions $p$ and $q$,
$L=\log\{p(X)/q(X)\}$ with $X\sim P$ \citep{zhu2022optimal};
its value is $+\infty$ if an output is possible under $P$ but impossible
under $Q$. Appendix~\ref{app:proof-information-equivalence} gives the
definition for arbitrary probability measures.
We write $K(t):=\log\EE_P[e^{tL}]$ for its cumulant-generating function
(CGF). Throughout, $t=\alpha-1$. At order $\alpha=1+t>1$ the exact RDP
curve is $D_{1+t}(P\|Q)=K(t)/t$. These quantities may be infinite;
we set $K(0)=0$.

An $(\alpha,\eps_\alpha)$-RDP guarantee bounds $D_\alpha(P\|Q)$ by
$\eps_\alpha$ for every ordered neighboring pair \citep{mironov2017renyi}.
A mechanism is $\widetilde\rho$-zCDP if
$D_\alpha(P\|Q)\leq\widetilde\rho\alpha$ for every $\alpha>1$ and every
ordered neighboring pair \citep{bun2016concentrated}.
For a fixed pair, the optimal zCDP parameter is
\[
 \rho(P,Q):=\sup_{\alpha>1}\frac{D_\alpha(P\|Q)}{\alpha}.
\]
Thus zCDP replaces the RDP curve by a single number.

We say that a notion $\mathcal A$ \emph{determines} a notion $\mathcal B$ if
any two pairs $(P,Q)$ and $(P',Q')$ with the same $\mathcal A$ also have the
same $\mathcal B$; two notions that determine each other carry the same
information. For instance, the RDP curve determines $\rho$, but $\rho$ does not
determine the RDP curve (Section~\ref{sec:defect}).

\begin{theorem}[Information equivalence between curve-based notions]
\label{thm:equivalence}
Let $P$ and $Q$ be arbitrary probability measures on
$(\Omega,\mathcal F)$.  The following three curve-based notions
determine one another:
\begin{enumerate}[label=\textup{(\alph*)},leftmargin=2.6em,itemsep=0pt,topsep=1pt]
\item the pair of one-sided $(\eps,\delta)$-DP curves
$(\delta_{P\|Q},\delta_{Q\|P})$;
\item the pair of trade-off functions $(T_{P,Q},T_{Q,P})$;
\item the law of the extended privacy-loss random variable $L$ under $P$.
\end{enumerate}
Each notion determines the exact RDP curve
$\bigl(D_\alpha(P\|Q)\bigr)_{\alpha>1}$.
Conversely, if $D_{\alpha_0}(P\|Q)<\infty$ for some real
$\alpha_0>1$, the exact RDP curve determines all three curve-based notions above.
\end{theorem}

\textbf{Remark.}
The condition $D_{\alpha_0}(P\|Q)<\infty$ holds for every mechanism in
Sections~\ref{sec:practical-implications} and~\ref{sec:real-data-experiments}.
It fails only when $\EE_P[e^{tL}]=\infty$ for every $t>0$, for example when $P$
puts mass outside the support of $Q$ or the privacy loss is very heavy-tailed.
Although weak, the condition cannot be dropped, nor weakened to a finite mean
privacy loss.

\begin{example}[Finite mean privacy loss is insufficient]
\label{ex:finite-order-needed}
Write
$\zeta(w)=\sum_{k=1}^\infty k^{-w}$ and
$\operatorname{Li}_w(z)=\sum_{k=1}^\infty z^k/k^w$.  On
$\Omega=\{\partial,1,2,\ldots\}$, for $w>2$, define
\begin{align*}
 P_w\{k\}&=\frac{k^{-w}}{\zeta(w)},& P_w\{\partial\}&=0,\\
 Q_w\{k\}&=P_w\{k\}e^{-k},&
 Q_w\{\partial\}&=1-\frac{\operatorname{Li}_w(e^{-1})}{\zeta(w)}.
\end{align*}
Then $Q_w$ is a probability measure, $P_w\ll Q_w$, and the privacy loss
$L_w$ of $(P_w,Q_w)$ satisfies
\[
 \EE_{P_w}L_w=\frac{\zeta(w-1)}{\zeta(w)},\qquad
 \delta_{P_w\|Q_w}(1)
 =\frac{1}{\zeta(w)}\sum_{k=2}^\infty\frac{1-e^{1-k}}{k^w}.
\]
For $w=3$ and $w=4$, respectively,
\begin{align*}
 \EE_{P_3}L_3&=1.36843\ldots,
 &\delta_{P_3\|Q_3}(1)&=0.124864\ldots,\\
 \EE_{P_4}L_4&=1.11063\ldots,
 &\delta_{P_4\|Q_4}(1)&=0.053061\ldots.
\end{align*}
Both mean privacy losses are finite, and the forward curves
$\delta_{P_3\|Q_3}$ and $\delta_{P_4\|Q_4}$ differ. Nevertheless, $D_\alpha(P_w\|Q_w)=+\infty$ for every
$\alpha>1$ and either value of $w$. Thus finite mean privacy loss is
insufficient for recovery from the RDP curve.
\end{example}

\textbf{Both one-sided curves are needed.}
\phantomsection\label{subsec:asymmetric-noise}
Example~\ref{ex:one-side-counterexample} adds one of two integer noise
distributions to a count of sensitivity $1$. The two distributions have full
support, likelihood ratios at most $4$, the same right tail, and different left
tails; call the resulting pairs $(P_A,Q_A)$ and $(P_B,Q_B)$. For $\eps\geq0$,
the forward curve $\delta_{P\|Q}(\eps)$ involves only outputs where $P$ exceeds
$Q$, here the shared right tail, so both pairs have
$\delta_{P\|Q}(\eps)=\tfrac{2}{11}(4-e^\eps)_+$. The reverse curve involves the
left tail and differs: $\delta_{Q_A\|P_A}(\log3)=3/22$ but
$\delta_{Q_B\|P_B}(\log3)=0$. Because the reverse curve lies below the forward
curve for both pairs, the usual curve $\max\{\delta_{P\|Q},\delta_{Q\|P}\}$ is
the same for both, and the two mechanisms would be reported as equally
private. Yet their one-sided curves differ, so by
Theorem~\ref{thm:equivalence} their trade-off functions, PLDs, and RDP curves
differ too; indeed $\exp(D_2(P_A\|Q_A))=265/88$, whereas
$\exp(D_2(P_B\|Q_B))=3$. The equivalence in Theorem~\ref{thm:equivalence} is
therefore with the \emph{pair} of one-sided curves; their maximum, the usual
curve, carries less information.

\begin{example}[The same two-sided curve with different one-sided curves]
\label{ex:one-side-counterexample}
Let $c=6/11$. The two noise distributions described above have mass functions
\[
 q_A(k)=q_B(k)=c4^{-k}\quad(k\geq0),
\]
and, for $j\geq1$,
\[
 q_A(-j)=c2^{-(j+1)},\qquad q_B(-j)=c3^{-j}.
\]
For $i\in\{A,B\}$, let $Q_i=\operatorname{Law}(Z_i)$ and
$P_i=\operatorname{Law}(Z_i+1)$ when $Z_i$ has mass function $q_i$.
For $\eps\geq0$, their one-sided curves are
\begin{align*}
 \delta_{P_i\|Q_i}(\eps)
 &=\frac{4c}{3}\left(1-\frac{e^\eps}{4}\right)_+,\\
 \delta_{Q_A\|P_A}(\eps)
 &=c\left(1-\frac{e^\eps}{4}\right)_+
   +\frac c2\left(1-\frac{e^\eps}{2}\right)_+,\\
 \delta_{Q_B\|P_B}(\eps)
 &=\frac{3c}{2}\left(1-\frac{e^\eps}{3}\right)_+.
\end{align*}
The common forward curve dominates both reverse curves for every
$\eps\geq0$. Hence the two-sided curves, and the strongest symmetric
$f$-DP guarantees obtained from them, are the same for the two pairs.
\end{example}

\textbf{When a single curve suffices.}
Taking the maximum loses no information when the two one-sided curves coincide. This
happens for additive noise symmetric about zero, such as Gaussian and Laplace
noise: if $P$ and $Q$ are the distributions of $s+Z$ and $Z$ with
$Z\overset{d}{=}-Z$, the reflection $x\mapsto s-x$ exchanges $P$ and $Q$, so
$\delta_{P\|Q}=\delta_{Q\|P}$, $T_{P,Q}=T_{Q,P}$, and
$D_\alpha(P\|Q)=D_\alpha(Q\|P)$. The usual two-sided curve then carries all the
information (Corollary~\ref{cor:symmetric-additive} in
Appendix~\ref{app:symmetric-additive}).

\textbf{Scope of the equivalence.}
First, Theorem~\ref{thm:equivalence} concerns information, not computation:
converting one notion into another may require the entire curve and need not
have a closed form; practical conversion and composition methods are developed
in \citet{Koskela2020computing,koskela2021tight,gopi2021numerical,doroshenko2022connect,ghazi2022faster}.
Second, it concerns one fixed pair $(P,Q)$, not a whole mechanism. A
mechanism-level RDP guarantee takes, at each order $\alpha$, the maximum over
all neighboring pairs; the maximizing pair may change with $\alpha$, so the
resulting curve need not be the RDP curve of any single pair, and a conversion
from it must hold for every pair consistent with it
\citep{zhu2022optimal,asoodeh2021three,balle2020hypothesis}.
Proofs of the equivalence theorem and the two examples are given in
Appendix~\ref{app:proof-information-equivalence}.
By contrast, replacing the exact RDP curve by the zCDP parameter $\rho$ loses
information even for a fixed pair; Section~\ref{sec:defect} quantifies this
loss.

\section{The Information Gap}\label{sec:defect}

By Theorem~\ref{thm:equivalence}, the exact RDP curve, the one-sided curves,
the trade-off functions, and the PLD carry the same information. zCDP keeps
only $\rho$, and many curves share the same $\rho$: for a query change $s$,
Gaussian noise $N(0,1)$ and the mixture $\tfrac12N(-m,1)+\tfrac12N(m,1)$,
$m>0$, both have $\rho=s^2/2$ but different RDP curves
(Proposition~\ref{prop:common-mechanisms}). Throughout this section and
Section~\ref{sec:practical-implications}, we consider mutually absolutely
continuous pairs $(P,Q)$ with finite $\rho$.

The parameter $\rho$ gives the bound $D_\alpha(P\|Q)\leq\rho\alpha$ at every
order $\alpha>1$; we call the line $\alpha\mapsto\rho\alpha$ the \emph{zCDP
envelope}. What zCDP discards is how far the exact curve lies below this line,
which we measure by the \emph{information gap} defined as
\begin{equation}
\label{eq:gap-definition}
 \Ga(t)
 :=\rho t(t+1)-K(t)
 =t\left\{(t+1)\rho-D_{t+1}(P\|Q)\right\},
 \quad \hbox{for any }t>0.
\end{equation}
Here, $(t+1)\rho$ is the zCDP upper bound and
$D_{t+1}(P\|Q)$ is the exact RDP value at the same order $\alpha=t+1$,
so $\Ga(t)$ is their difference multiplied by $t$.
The gap has three useful properties (illustrated in Figure~\ref{fig:gap-metric}).
\emph{(1)} $\Ga(t)/t$ is the vertical distance at order $1+t$ between the
envelope and the exact RDP curve, so $\Ga\geq0$, with $\Ga\equiv0$ exactly when
the envelope reproduces the curve; the factor $t$ puts $\Ga$ on the scale of
$K$, which adds up under independent composition. \emph{(2)} For Gaussian
noise, $L\sim N(\rho,2\rho)$ under $P$, so $K(t)=\rho t(t+1)$ and $\Ga\equiv0$:
the Gaussian RDP curve \emph{is} its envelope. \emph{(3)} The gap controls the reported $\eps$ after order optimization
(Proposition~\ref{prop:optimized-gap}).

\begin{figure}[!htbp]
\centering
\includegraphics[width=.7\linewidth]{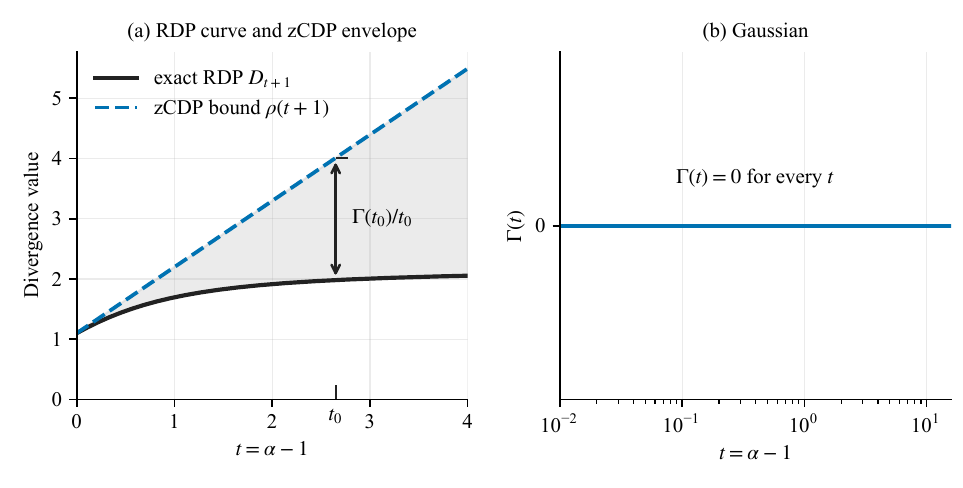}
\caption{The information gap. Left: $\Ga(t)/t$ is the vertical distance
between the exact RDP curve and its zCDP envelope; here $P=(9,3,3,1)/16$ and
$Q=(1,3,3,9)/16$ on four points, $\rho=\log3$. Right: for
Gaussian noise ($\rho=1/2$), $\Ga(t)=0$ for every $t>0$.}
\label{fig:gap-metric}
\end{figure}

\subsection{Local Expansion}
\label{subsec:local-expansion-theorem}

How large is the gap when the noise is large compared with the change one
record can cause? For the additive mechanism $\mathcal M(D)=\theta(D)+Z$, where
$Z$ has density or probability mass function $p$, let
$s:=\theta(D)-\theta(D')$ be the query change between neighboring datasets, so
$|s|$ is at most the global sensitivity of $\theta$. Subtracting $\theta(D')$
from both outputs gives the pair $(P_s,Q)$ with densities
\begin{equation}\label{eq:location-pair-convention}
 p_s(x)=p(x-s),\qquad q(x)=p(x).
\end{equation}
More generally, we consider any family $(P_s,Q_s)$ with $P_0=Q_0$, and study
small nonzero $s$ of either sign. For each $s$, let $L_s$, $K$, $\rho(s)$, and
$\Ga$ be the privacy loss, CGF, optimal zCDP parameter, and gap of
$(P_s,Q_s)$, and let $\kappa_j(L_s)$ be the $j$th cumulant of $L_s$ under
$P_s$ ($\kappa_1$ is the mean, $\kappa_2$ the variance; $\kappa_3$ and
$\kappa_4$ measure skewness and tail weight). Since $K$ and the supremum
defining $\rho$ rarely have closed forms, the next theorem reduces $\Ga$ for
small $s$ to the mean and the third and fourth cumulants of the privacy loss.

\begin{samepage}
\begin{theorem}[Local expansion of the zCDP--RDP gap]
\label{thm:sensitivity-gap}
Fix $T>0$. Suppose that, for all sufficiently small $s\neq0$, $P_s$ and $Q_s$
are mutually absolutely continuous, $\rho(s)<\infty$, and the CGF
$K(u)=\log\EE_{P_s}[e^{uL_s}]$ is five times continuously differentiable on
$[-1,T]$ with fifth derivative satisfying
$\sup_{-1\leq u\leq T}|K^{(5)}(u)|\leq M_T|s|^5$, where $M_T$ does not depend
on $s$. Then, for all such $s$ and all $0<t\leq T$, we have
\begin{align}
 \Ga(t)&=t(t+1)\bigl(\rho(s)-\EE_{P_s}L_s\bigr)
 -\frac{t^2(t+1)}6\kappa_3(L_s)
 +\frac{t^2(1-t^2)}{24}\kappa_4(L_s)+\mathcal R_s(t),
 \label{eq:unified-profile-cumulant-expansion}\\
 |\mathcal R_s(t)|&\leq\frac{M_T}{120}(t^2+t^5)|s|^5.
 \label{eq:one-pair-fifth-order-bound}
\end{align}
\end{theorem}
\end{samepage}

\textbf{Remark.}
The first term measures how far $\rho(s)$ exceeds the mean privacy loss
$\EE_{P_s}L_s$, which is the KL divergence of $P_s$ from $Q_s$; the other two
involve $\kappa_3$ and $\kappa_4$. For Gaussian noise all three vanish, since
$\rho(s)=\EE_{P_s}L_s$ and $\kappa_3=\kappa_4=0$. In general the three terms
are of order $s^2$, $s^3$, and $s^4$; Section~\ref{subsec:mechanism-verdicts}
finds the leading one for common mechanisms.

\subsection{From the Gap to Privacy Bounds}
\label{subsec:optimized-certificates}
We next relate the information gap to the privacy bounds obtained by
applying the same RDP-to-$(\eps,\delta)$ conversion to the exact RDP
curve and its zCDP envelope.
These comparisons keep the output pair,
including its noise scale, unchanged and do not require a small-shift
assumption. For a fixed pair $(P,Q)$ satisfying the standing assumptions of this
section, let $c_t:=t^t/(1+t)^{1+t}$. The conversion of
\citet{balle2020hypothesis} implies
\begin{equation}\label{eq:rdp-delta-bound}
 \delta_{P\|Q}(\eps)\leq c_t\exp\{K(t)-t\eps\},\qquad t>0.
\end{equation}
Its derivation is given in
Appendix~\ref{subsec:proof-pairwise-gap}.

\paragraph{The $\eps$ Bound at a Fixed $\delta$.}
Fix $\delta\in(0,1)$ and define
\begin{align*}
 F_{\rm R}(t)
 &:=\frac{K(t)+\log c_t-\log\delta}{t},
 \\
 F_{\rm Z}(t)
 &:=\rho(t+1)+\frac{\log c_t-\log\delta}{t}
 =F_{\rm R}(t)+\frac{\Ga(t)}{t}.
\end{align*}
Minimizing these objectives over $t>0$ gives the $\eps$ upper bounds
obtained from the exact RDP curve and its zCDP envelope, respectively.
The following proposition quantifies the difference after order
optimization.

\begin{proposition}[Compression loss after order optimization]
\label{prop:optimized-gap}
Suppose that
\[
 \eps_{\rm R}:=\inf_{t>0}F_{\rm R}(t),
 \qquad
 \eps_{\rm Z}:=\inf_{t>0}F_{\rm Z}(t)
\]
are attained at $t_{\rm R},t_{\rm Z}\in(0,\infty)$, respectively.
Then we have
\begin{equation}
 \frac{\Ga(t_{\rm Z})}{t_{\rm Z}}
 \leq \eps_{\rm Z}-\eps_{\rm R}
 \leq \frac{\Ga(t_{\rm R})}{t_{\rm R}}.
 \label{eq:optimized-gap-sandwich}
\end{equation}
Consequently, if an interval $I\subset(0,\infty)$ contains
$t_{\rm R}$ and
$\sup_{t\in I}\Ga(t)/t\leq\eta$, then
$0\leq\eps_{\rm Z}-\eps_{\rm R}\leq\eta$.
\end{proposition}

The relevant quantity is therefore $\Ga(t)/t$ near the optimizing
orders, rather than the gap at an arbitrary reference order.
A large value of $\Ga(t_{\rm Z})/t_{\rm Z}$ certifies a nonnegligible
compression cost, whereas a uniform upper bound over an interval
containing $t_{\rm R}$ certifies that the cost is small.
Locating such an interval and bounding the gap without computing the
full RDP curve requires mechanism-specific analysis.

These bounds compare privacy guarantees at a fixed noise scale.
Turning them into bounds on the required noise variance also requires
control of how the privacy guarantee changes with the noise scale.

\paragraph{Using the Local Expansion.}
Theorem~\ref{thm:sensitivity-gap} makes the upper bound explicit when
the optimizing order lies in the range of the local expansion.
Under its assumptions and those of Proposition~\ref{prop:optimized-gap},
if $0<t_{\rm R}\leq T$, then
\begin{equation}\label{eq:local-calibration-upper-bound}
\begin{aligned}
 0\leq\eps_{\rm Z}-\eps_{\rm R}
 \leq{}&(t_{\rm R}+1)\bigl(\rho(s)-\EE_{P_s}L_s\bigr)
 -\frac{t_{\rm R}(t_{\rm R}+1)}6\kappa_3(L_s)\\
 &+\frac{t_{\rm R}(1-t_{\rm R}^2)}{24}\kappa_4(L_s)
 +\frac{M_T}{120}(t_{\rm R}+t_{\rm R}^4)|s|^5.
\end{aligned}
\end{equation}
This follows by substituting
\eqref{eq:unified-profile-cumulant-expansion} and
\eqref{eq:one-pair-fifth-order-bound} into
\eqref{eq:optimized-gap-sandwich}. It separates the contribution of the
mean privacy loss from the third- and fourth-cumulant terms and the
remainder. The condition $t_{\rm R}\leq T$ is needed to apply the remainder
bound at the optimizing order.

\paragraph{The $\delta$ Bound at a Fixed $\eps$.}
If $\eps$ is fixed instead, optimizing over $t$ gives a $\delta$ upper
bound. For any valid zCDP upper bound $\widetilde\rho\geq\rho(P,Q)$, define
\begin{equation}\label{eq:zcdp-delta-bound}
 \delta_{\rm zCDP}(\eps;\widetilde\rho)
 :=\inf_{t>0}c_t\exp\{\widetilde\rho t(t+1)-t\eps\}.
\end{equation}
For $\widetilde\rho>0$, the minimizer $t_\eps(\widetilde\rho)$ is unique
and satisfies
\begin{equation}\label{eq:chernoff-stationarity}
 \log\frac{t_\eps(\widetilde\rho)}{1+t_\eps(\widetilde\rho)}
 +\widetilde\rho\{2t_\eps(\widetilde\rho)+1\}-\eps=0.
\end{equation}
When $\rho>0$, taking $\widetilde\rho=\rho$ gives
\begin{equation}\label{eq:pairwise-transfer}
 \delta_{P\|Q}(\eps)
 \leq e^{-\Ga(t_\eps)}\delta_{\rm zCDP}(\eps;\rho),
 \qquad t_\eps=t_\eps(\rho).
\end{equation}
Thus evaluating the exact RDP curve at the zCDP optimizing order
reduces this $\delta$ upper bound by a factor $e^{-\Ga(t_\eps)}$.
Appendix~\ref{subsec:proof-pairwise-gap} proves these statements.

\subsection{Information Loss under Independent Composition}
\label{subsec:composition-main}

We next show how the information gap accumulates under independent
composition. For component pairs $(P_i,Q_i)$, let $K_i$ be their
privacy-loss CGFs, $\rho_i<\infty$ their optimal zCDP parameters, and
$\Ga_i(t):=\rho_i t(t+1)-K_i(t)$ their information gaps. Define
\[
 \bK(t):=\sum_{i=1}^n K_i(t),
 \qquad
 \rho_{\mathrm{sum}}:=\sum_{i=1}^n\rho_i.
\]
The gap relative to the standard composed zCDP envelope is
\[
 \bGa(t):=\rho_{\mathrm{sum}}t(t+1)-\bK(t).
\]

\begin{theorem}[Independent composition]
\label{thm:composition-extension}
Let $(P_i,Q_i)$, $i=1,\ldots,n$, be independent component pairs with
optimal zCDP parameters $\rho_i<\infty$. Set
$P^{(n)}:=\bigotimes_{i=1}^n P_i$ and
$Q^{(n)}:=\bigotimes_{i=1}^n Q_i$.
Then, for every $t>0$, it holds that
$tD_{t+1}(P^{(n)}\|Q^{(n)})=\bK(t)$ and
\begin{equation}\label{eq:gap-additivity}
 \bGa(t)=\sum_{i=1}^n\Ga_i(t).
\end{equation}
In particular, for $n$ independent copies of a fixed pair $(P,Q)$
with optimal zCDP parameter $\rho$, the product pair has optimal
zCDP parameter $n\rho$ and information gap $n\Ga(t)$.
\end{theorem}

Thus, at a fixed order and with the component pair unchanged,
the gap grows linearly with the number of independent copies.
A component with $P_i=Q_i$ contributes zero.

\textbf{Standard versus optimal composition.}
For nonidentical components, the optimal zCDP parameter of the
product pair may be smaller than the sum of the component parameters:
\[
 \rho_\star:=
 \sup_{t>0}\frac{\bK(t)}{t(t+1)}
 \leq \rho_{\mathrm{sum}}.
\]
Consequently, we have
\begin{equation}\label{eq:slack-decomposition-main}
 \bGa(t)=(\rho_{\mathrm{sum}}-\rho_\star)\,t(t+1)
 +\bigl\{\rho_\star t(t+1)-\bK(t)\bigr\}.
\end{equation}
Both terms are nonnegative. The first records the difference between
adding the component parameters and optimizing the parameter of the
product pair; it vanishes for identical components. The second is the
gap between the product pair's exact RDP curve and its optimal zCDP
envelope.

Under the assumptions of Theorem~\ref{thm:sensitivity-gap} for each
component family, the local expansions also add componentwise.

\section{Information Gaps for Common Mechanisms}
\label{sec:practical-implications}

\subsection{Continuous Additive Noise Mechanisms}
\label{subsec:mechanism-verdicts}

We evaluate $\Ga$ for three continuous additive noise distributions,
Gaussian, Gaussian-mixture, and Laplace, in the setting
of~\eqref{eq:location-pair-convention}. All three are symmetric about zero, so
one direction suffices (Section~\ref{sec:information-equivalence}). By
Theorem~\ref{thm:sensitivity-gap}, the leading power of $s$ in $\Ga(t)$
measures the loss: $\Ga\equiv0$ means none, and the lower the power, the
larger the loss for small $s$.

\begin{proposition}[Gaps for common additive noise]
\label{prop:common-mechanisms}
Let $(P_s,Q)$ be as in~\eqref{eq:location-pair-convention} and fix $T>0$. The
remainders below hold as $s\to0$, uniformly over $0<t\leq T$.
\begin{enumerate}[label=(\alph*),leftmargin=2.2em,itemsep=1pt,topsep=2pt]
\item \emph{Gaussian.} If $Z\sim N(0,\sigma^2)$, then
$\rho(s)=s^2/(2\sigma^2)$ and $\Ga(t)=0$ for every $s$ and $t>0$.
\item \emph{Gaussian mixture.} If $Z\sim\tfrac12N(-m,1)+\tfrac12N(m,1)$ with
$m\geq0$, then we have $\rho(s)=s^2/2$ for every $s$, and
\begin{equation}\label{eq:gaussian-mixture-leading}
 \Ga(t)
 =\frac{m^2}{2}\EE[\operatorname{sech}^2(mZ)]\,t(t+1)s^2+O(s^4).
\end{equation}
The coefficient is positive when $m>0$; for $m=0$ this is case~(a) with
$\sigma=1$.
\item \emph{Laplace.} If $Z$ has density $\tfrac12e^{-|x|}$, then we have
$\rho(s)=|s|+e^{-|s|}-1$ for every $s$ \citep{harrison2025exact}, and
\[
 \Ga(t)=\frac{[t(t+1)]^2}{12}s^4+O(|s|^5).
\]
\end{enumerate}

\end{proposition}

\begin{figure}[t]
\centering
\includegraphics[width=\linewidth]{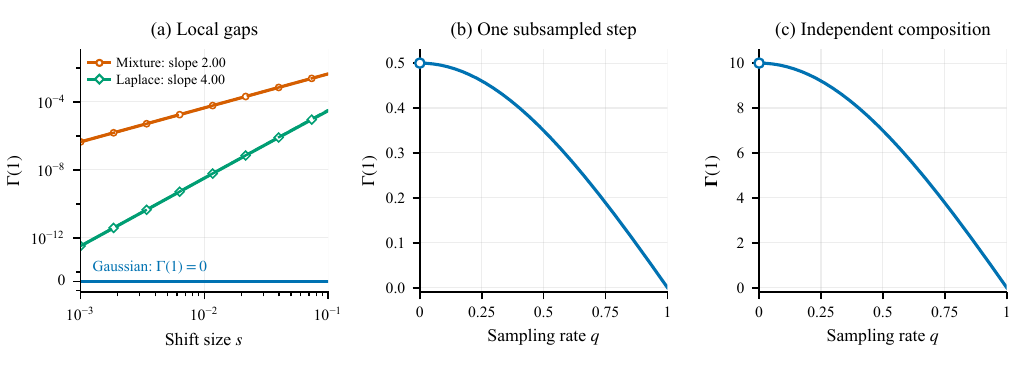}
\caption{Information gaps at $t=1$. (a) Gaussian noise has zero gap; the
Gaussian-mixture ($m=1$) and Laplace gaps are quadratic and quartic in $s$
(solid: numerical; dashed: leading terms of
Proposition~\ref{prop:common-mechanisms}; slopes fitted over $s\leq0.01$).
(b)--(c) Single-step and composed gaps of the subsampled Gaussian decrease in
$q$ ($n=20$, $\rho_{\mathrm{sum}}=5$, $z=\sqrt2$; open circles are limits as
$q\downarrow0$). Appendix~\ref{app:alternative-order-robustness} repeats these
comparisons at $t=1.5$.}
\label{fig:common-mechanism-gaps}
\end{figure}

In short, zCDP loses no information for Gaussian noise, loses at order $s^2$ for the
Gaussian mixture, and only at order $s^4$ for Laplace noise;
Figure~\ref{fig:common-mechanism-gaps}(a) confirms these powers, with fitted
log--log slopes $2.00$ and $4.00$. Discrete noise is considered in
Section~\ref{subsec:discrete-noise}.
By~\eqref{eq:gap-additivity}, a workload of $d$
queries has $d$ times the one-query gap; Section~\ref{subsec:acs-composition}
shows what this costs in noise variance.

\subsection{Poisson Subsampling}
\label{subsec:poisson-subsampling}

At every step, DP-SGD clips each per-example gradient to norm at most $C$,
sums them over a Poisson subsample that includes each record independently
with probability $q$, and adds Gaussian noise $N(0,z^2C^2I)$ with \emph{noise
multiplier} $z$ \citep{abadi2016deep}. Subsampling amplifies privacy.
Finite-order RDP captures this amplification, whereas standard zCDP, whose
guarantee must hold uniformly over all R\'enyi orders, does not retain it for
the sampled Gaussian mechanism \citep{zhu2019poission,bun2018composable}.
Under add-or-remove neighboring datasets, one step reduces to the sampled
Gaussian mechanism \citep{mironov2019sampled}:
\[
 Q=N(0,z^2C^2I),\qquad
 P_q=(1-q)Q+qN(\mathbf v,z^2C^2I),
\]
where $\mathbf v$ is the clipped gradient
of the extra record, with $\|\mathbf v\|=C$ in the worst case.

\begin{proposition}[Poisson subsampling at a fixed zCDP budget]
\label{prop:subsampled-gaussian}
Fix $n\geq1$, $C>0$, and $\rho_{\mathrm{sum}}>0$, set
$z_n^2:=n/(2\rho_{\mathrm{sum}})$, and let $(P_q,Q)$ be the one-step pair above
with $z=z_n$ and $\|\mathbf v\|=C$.
\begin{enumerate}[label=\textup{(\alph*)},leftmargin=2.2em,itemsep=1pt,topsep=2pt]
\item For $0<q\leq1$, we have
\begin{align}
 K(t)
 &=\log\EE_{G\sim N(0,1)}\!\left[
 \left\{1-q+q\exp\left(\frac{G}{z_n}-\frac{1}{2z_n^2}\right)\right\}^{t+1}
 \right],\notag\\
 \rho(P_q,Q)&=\frac{\rho_{\mathrm{sum}}}{n},
 \qquad \hbox{and} \qquad
 \Ga(t)
 =\frac{\rho_{\mathrm{sum}}}{n}t(t+1)-K(t).
 \label{eq:subsampled-one-step-gap}
\end{align}
For every fixed $t>0$, $\Ga(t)$ is strictly decreasing in $q$ on $(0,1]$,
equals $0$ at $q=1$, and tends to $\frac{\rho_{\mathrm{sum}}}{n}t(t+1)$ as
$q$ tends to $0$. At $t=1$, we have
\[
 \Ga(1)
 =\frac{2\rho_{\mathrm{sum}}}{n}
 -\log\left\{1+\left(e^{2\rho_{\mathrm{sum}}/n}-1\right)q^2\right\}.
\]
\item For $n$ independent copies of this pair, the product pair has optimal
zCDP parameter $\rho_{\mathrm{sum}}$, and by~\eqref{eq:gap-additivity} its gap is
\[
 \bGa(t)=\rho_{\mathrm{sum}}t(t+1)-\bK(t)=n\Ga(t),
\]
which has the same monotonicity in $q$, vanishes at $q=1$, and tends to
$\rho_{\mathrm{sum}}t(t+1)$ as $q$ tends to $0$.
\end{enumerate}
\end{proposition}

At $q=1$ there is no subsampling, each step is Gaussian, and zCDP loses
no information. As $q$ tends to $0$, the exact RDP curve tends to zero while the zCDP
parameter stays at $\rho_{\mathrm{sum}}$, so the gap approaches its largest
possible value $\rho_{\mathrm{sum}}t(t+1)$: zCDP charges for privacy loss that
does not occur (Figure~\ref{fig:common-mechanism-gaps}(b)--(c)). This is one source of
the utility differences in Section~\ref{subsec:controlled-dpsgd}.

\FloatBarrier

\subsection{Allocation under a Fixed zCDP Composition Parameter}
\label{subsec:privacy-allocation}

Independent composition adds the gaps, but fixing the total zCDP
parameter does not fix the total gap. Its value can depend on how that
parameter is divided among the releases. We study this dependence for
$n$ independent applications of the same one-dimensional additive-noise
mechanism.  Let $Z_1,\ldots,Z_n$ be independent copies of a
real-valued random variable $Z$, and, for shift magnitudes $r_i\geq0$, set
$Q_i=\operatorname{Law}(Z_i)$ and
$P_{i,r_i}=\operatorname{Law}(Z_i+r_i)$ for $i=1,\ldots,n$.
The noise draws have the same distribution, but the shift magnitudes
may differ.

Fix $\rho_{\mathrm{sum}}>0$. To express each shift in terms of its
zCDP parameter, assume
that $\rho(0)=0$ and that $r\mapsto\rho(r)$ is continuous and
strictly increasing on a range whose image contains $[0,\rho_{\mathrm{sum}}]$.  Its
inverse on this interval is denoted by
$r=\psi(x)$, so $\rho\bigl(\psi(x)\bigr)=x$ for
$0\leq x\leq\rho_{\mathrm{sum}}$.
At $r=0$, we use $K(t)=\Ga(t)=0$.  Set $\rho_i=\rho(r_i)$ and
$r_i=\psi(\rho_i)$.  We study allocations
satisfying
\begin{equation}\label{eq:fixed-zcdp-composition-budget}
 \rho_i\geq0,
 \qquad
 \sum_{i=1}^n\rho_i=\rho_{\mathrm{sum}}.
\end{equation}
Each $\rho_i$ is the optimal zCDP parameter of its component pair.
For fixed $t$, write the information gap as a function of this parameter:
\[
 \Phi_t(x):=\left.\Ga(t)\right|_{r=\psi(x)},
 \qquad 0\leq x\leq\rho_{\mathrm{sum}}.
\]
Equivalently, $\Phi_t(x)=xt(t+1)-K(t)$, where $K(t)$ is evaluated for the
single pair with $r=\psi(x)$.  Substituting $r_i=\psi(\rho_i)$ into
\eqref{eq:gap-additivity} gives
\[
 \bGa(t)=\rho_{\mathrm{sum}}t(t+1)-\bK(t)
 =\sum_{i=1}^n\Phi_t(\rho_i).
\]
The same function $\Phi_t$ applies to every component because the
noise distribution is fixed. Its convexity determines whether spreading
the zCDP parameter evenly increases or decreases the total gap.

\begin{proposition}[Allocation of a fixed total zCDP parameter]
\label{prop:zcdp-budget-allocations}
Fix $t>0$ and consider all allocations satisfying
\eqref{eq:fixed-zcdp-composition-budget}.
\begin{enumerate}[label=\textup{(\alph*)}]
\item If $\Phi_t$ is convex on $[0,\rho_{\mathrm{sum}}]$, then every feasible
allocation satisfies
\begin{equation}\label{eq:convex-zcdp-allocation-bounds}
 n\Phi_t\!\left(\frac{\rho_{\mathrm{sum}}}{n}\right)
 \leq \bGa(t)
 \leq \Phi_t(\rho_{\mathrm{sum}}).
\end{equation}
Hence the uniform zCDP allocation
\[
 \rho_i=\frac{\rho_{\mathrm{sum}}}{n},
 \qquad
 r_i=\psi\left(\frac{\rho_{\mathrm{sum}}}{n}\right),
 \qquad i=1,\ldots,n,
\]
minimizes the gap.  Every permutation of the concentrated allocation
\[
 (\rho_1,\ldots,\rho_n)=(\rho_{\mathrm{sum}},0,\ldots,0),
 \qquad
 (r_1,\ldots,r_n)=(\psi(\rho_{\mathrm{sum}}),0,\ldots,0),
\]
maximizes it.
\item If $\Phi_t$ is concave on
$[0,\rho_{\mathrm{sum}}]$, then
\[
 \Phi_t(\rho_{\mathrm{sum}})
 \leq \bGa(t)
 \leq n\Phi_t\!\left(\frac{\rho_{\mathrm{sum}}}{n}\right),
\]
so concentration minimizes the gap and uniform allocation maximizes it.
\end{enumerate}
If the relevant function is strictly convex or strictly concave, the
uniform extremum is unique, and the concentrated extrema are exactly the
$n$ permutations displayed above.
\end{proposition}

For a convex $\Phi_t$, spreading the zCDP parameter evenly gives the
smallest gap; concentrating it in one release gives the largest. Concavity
reverses these choices. This comparison concerns one fixed RDP order,
and the best allocation need not be the same at every order.

Figure~\ref{fig:zcdp-budget-allocation} compares Laplace and Gaussian-mixture
noise as the allocation moves from uniform to concentrated. Both panels
use $t=1$, $n=8$, $\rho_{\mathrm{sum}}=1/2$, and the path
\begin{equation}\label{eq:zcdp-allocation-path}
 \rho_1(\omega)
 =\frac{\rho_{\mathrm{sum}}}{n}\{1+(n-1)\omega\},\qquad
 \rho_i(\omega)
 =\frac{\rho_{\mathrm{sum}}}{n}(1-\omega),\quad i=2,\ldots,n,
\end{equation}
with $r_i(\omega)=\psi\bigl(\rho_i(\omega)\bigr)$ and $0\leq\omega\leq1$.
Thus $\sum_i\rho_i(\omega)=\rho_{\mathrm{sum}}$ for every $\omega$,
$\omega=0$ is the uniform zCDP allocation, and $\omega=1$ is the
concentrated allocation, with shifts
$\{\psi(\rho_{\mathrm{sum}}),0,\ldots,0\}$.

\begin{figure}[!htbp]
\centering
\includegraphics[width=.85\linewidth]{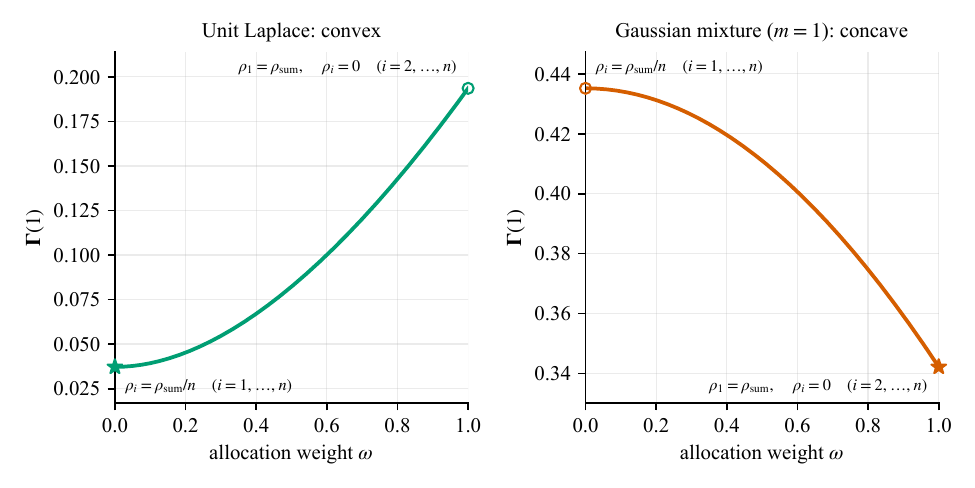}
\caption{Information gaps under independent composition as a fixed
total zCDP parameter is reallocated. Both panels use
$t=1$, $n=8$, $\rho_{\mathrm{sum}}=1/2$, and the path
\eqref{eq:zcdp-allocation-path}, from uniform ($\omega=0$) to concentrated
($\omega=1$). Left: for unit-scale Laplace noise, strict convexity of
$\Phi_1$ proves that the uniform allocation is the unique minimum and
concentrated allocations are the maxima. Right: for Gaussian-mixture
noise with $m=1$, the computed gap decreases along the path; centered
second differences of $\Phi_1$ are negative on the displayed grid at both
quadrature resolutions. Numerical settings are in Appendix~\ref{app:numerics};
Appendix~\ref{app:alternative-order-robustness} repeats the comparison
at $t=1.5$.}
\label{fig:zcdp-budget-allocation}
\end{figure}

For unit-scale Laplace noise,
$\rho(r)=r+e^{-r}-1$, and the proof verifies that
$\Phi_1$ is strictly convex.  The composed gap therefore
increases from $0.037202015$ at the uniform allocation to $0.193535993$ at
the concentrated allocation.  For the symmetric Gaussian mixture with
$m=1$ in Proposition~\ref{prop:common-mechanisms},
$\rho(r)=r^2/2$, so $\psi(x)=\sqrt{2x}$.
At two Gauss--Hermite resolutions, centered second differences of
$\Phi_1$, computed from the exact CGF, are negative on the displayed
grid in $[0,1/2]$, providing numerical evidence for concavity. Along this path, the gap decreases from
$0.435179700$ at the uniform allocation to $0.342190807$ at a
concentrated allocation. Unlike the Laplace convexity result, this
Gaussian-mixture comparison is a numerical observation on the displayed
grid.

\FloatBarrier

\subsection{Tail Parameters and the Information Gap}
\label{subsec:tail-weight}

We return to additive noise and ask how the gap changes as the tail
parameter varies within a noise family. We compare one ordered location
pair at a time, so composition does not affect this comparison.
For Student-$t_\nu$ noise, $\rho_\nu(s)$ equals the mean privacy loss
for sufficiently small shifts. The local expansion therefore gives a
quartic gap, as for Laplace noise, but its coefficient need not decrease
as the tails become lighter.

\begin{proposition}[Student-$t_\nu$ location pair]
\label{prop:student-general}
Let
\[
 f_\nu(x)=
 \frac{\Gamma((\nu+1)/2)}
 {\sqrt{\nu\pi}\,\Gamma(\nu/2)}
 \left(1+\frac{x^2}{\nu}\right)^{-(\nu+1)/2},
 \qquad \nu\geq1,
\]
and let $Z$ have density $f_\nu$, $P_s=\operatorname{Law}(Z+s)$, and
$Q=\operatorname{Law}(Z)$.  Write $\rho_\nu(s):=\rho(P_s,Q)$.  Then
$\rho_\nu(s)<\infty$ for every $s\in\R$.  For every fixed $\nu\geq1$,
one has $\rho_\nu(s)=\EE_{P_s}L_s$ for all sufficiently
small nonzero $s$.  Moreover, for every fixed $T>0$, uniformly for
$0<t\leq T$ and all sufficiently small nonzero $s$,
\begin{equation}\label{eq:student-general-gap}
 \Ga(t)=c_{4,\nu}[t(t+1)]^2s^4+O(|s|^6),
 \qquad
 c_{4,\nu}:=
 \frac{3(\nu+1)^2(\nu^2+4\nu-1)}
 {4\nu(\nu+3)^2(\nu+5)(\nu+7)}.
\end{equation}
The coefficient has a unique maximum at
$\nu=6.013052641\ldots$ and satisfies
$c_{4,\nu}=3/(4\nu)+O(\nu^{-2})\to0$ as $\nu\to\infty$.
\end{proposition}

The leading coefficient first increases and then decreases as $\nu$
grows, even though the tails become lighter throughout. It reaches its
maximum near $\nu=6.013$ and tends to zero in the Gaussian limit.
Thus the tail parameter alone does not give a monotone ordering of the gap.

The generalized Gaussian family provides a second comparison: $\beta=1$
gives Laplace noise and $\beta=2$ gives Gaussian noise. Here the local
expansion can also contain a contribution from
$\rho_\beta(s)-\EE_{P_s}L_s$, so the quartic term alone need not give
the leading gap.

\begin{proposition}[Generalized Gaussian location pair]
\label{prop:gg-single}
Let $f_\beta(x)\propto\exp(-|x|^\beta/\beta)$ for $1\leq\beta\leq2$,
and let $Z$ have density $f_\beta$, $P_s=\operatorname{Law}(Z+s)$, and
$Q=\operatorname{Law}(Z)$.  Write $\rho_\beta(s):=\rho(P_s,Q)$ and define
\[
 \chi_\beta:=
 \frac{\beta^{4-4/\beta}}{24\Gamma(1/\beta)^2}
 \left\{
  3\Gamma(2-1/\beta)^2
  -\Gamma(4-3/\beta)\Gamma(1/\beta)
 \right\}\geq0.
\]
The optimal pairwise zCDP parameter $\rho_\beta(s)$ is finite for
every $s\in\R$.  For every fixed $\beta\in[1,2]$ and $T>0$, uniformly for
$0<t\leq T$ and all sufficiently small nonzero $s$,
\begin{equation}\label{eq:gg-single-gap}
 \Ga(t)
 =t(t+1)\{\rho_\beta(s)-\EE_{P_s}L_s\}
 +\chi_\beta[t(t+1)]^2s^4+o(|s|^4).
\end{equation}
At $\beta=1$, $\rho_1(s)=\EE_{P_s}L_s=|s|+e^{-|s|}-1$ and
$\chi_1=1/12$, recovering Laplace noise. At $\beta=2$,
$\rho_2(s)=\EE_{P_s}L_s=s^2/2$ and $\chi_2=0$, and $\Ga(t)=0$
identically for this Gaussian endpoint.
\end{proposition}

For $1<\beta<2$, we do not have a closed form for
$\rho_\beta(s)-\EE_{P_s}L_s$ at a finite shift.
Since $\rho_\beta(s)\geq\EE_{P_s}L_s$, for every fixed
$T>0$, uniformly for $0<t\leq T$ and all sufficiently small nonzero $s$,
\begin{equation}\label{eq:gg-quartic-lower-bound}
 \Ga(t)
 \geq\chi_\beta[t(t+1)]^2s^4+o(|s|^4).
\end{equation}

\begin{figure}[!htbp]
\centering
\includegraphics[width=.85\linewidth]{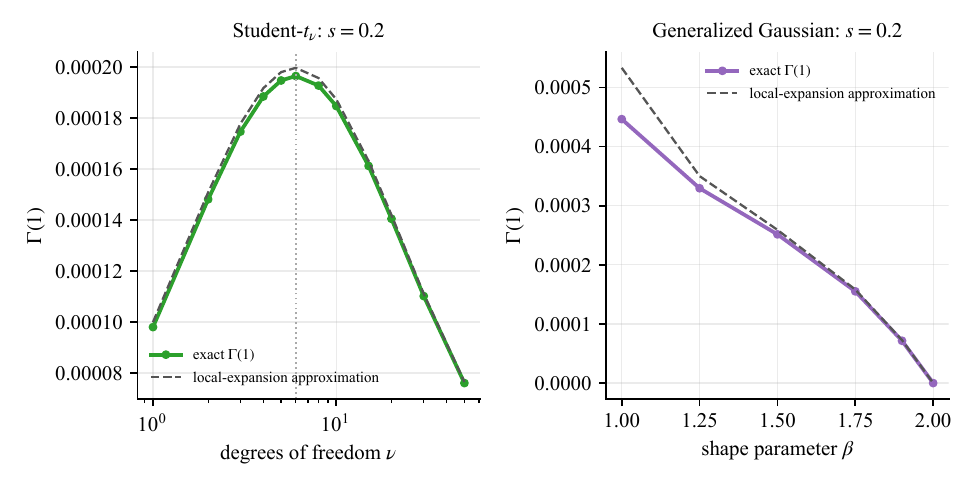}
\caption{Information gaps at $t=1$ and $s=0.2$ as the tail parameters
vary. Left: Student-$t_\nu$ noise, with quartic approximation
$4c_{4,\nu}s^4$ from \eqref{eq:student-general-gap}; the dotted line marks
the maximum of $c_{4,\nu}$. Right: generalized Gaussian noise, with
$\rho_\beta(s)$ obtained by numerically maximizing the normalized CGF and
quartic approximation $4\chi_\beta s^4$ from
\eqref{eq:gg-quartic-lower-bound}. The factor $4$ is $[t(t+1)]^2$ at
$t=1$. Solid curves show the computed gaps; dashed curves show the
small-shift approximations. The Student-$t_\nu$ gap is nonmonotone,
while the generalized Gaussian gap decreases to zero at $\beta=2$.
Appendix~\ref{app:alternative-order-robustness} repeats the comparison
at $t=1.5$.}
\label{fig:tail-weight}
\end{figure}

Figure~\ref{fig:tail-weight} compares the gaps at $t=1$ and $s=0.2$.
The Student-$t_\nu$ gap is nonmonotone in $\nu$, while the generalized
Gaussian gap decreases to zero at $\beta=2$ on the displayed grid.
For each generalized Gaussian pair, we obtain $\rho_\beta(s)$ by
numerically maximizing $K(u)/[u(u+1)]$ over $u>0$, including the
$u\downarrow0$ endpoint. The difference
$\rho_\beta(s)-\EE_{P_s}L_s$ is zero to numerical precision at every
displayed $\beta$. The plotted quartic term is therefore close to the
computed gap, although the full expansion in \eqref{eq:gg-single-gap}
also contains $2\{\rho_\beta(s)-\EE_{P_s}L_s\}$ at $t=1$.
The dashed curves are local approximations and are not asserted to bound
the gaps at this fixed shift.

\FloatBarrier

\subsection{Discrete and Lattice Mechanisms}
\label{subsec:discrete-noise}

For integer-valued queries, the RDP order reveals another difference
between noise laws. Continuous Gaussian noise has zero gap at every order;
the discrete Gaussian zCDP envelope is exact only at orders where the
shifted lattice realigns with $\Z$. We compare it with discrete Laplace noise, whose
privacy loss takes finitely many values and hence gives a finite-sum CGF.

\begin{proposition}[Exact gaps for discrete Laplace and discrete Gaussian noise]
\label{prop:discrete-examples}
Let $\theta$ be integer-valued and fix neighboring datasets $D,D'$ for which
$s:=\theta(D)-\theta(D')\ne0$ and
$r:=|s|\in\{1,2,\ldots\}$.  Because both noise laws below are centrally
symmetric, Corollary~\ref{cor:symmetric-additive} lets us use the positive
representative of each ordered location pair.  As in the continuous
families above, let $L_r$ denote the privacy loss of each displayed pair.
The following formulas hold for every $t>0$ and
extend continuously to $t=0$, where the gap is zero.
\begin{enumerate}[label=(\alph*)]
\item Fix $a>0$, let $Z$ have discrete Laplace mass proportional to
$e^{-a|z|}$,
put $Q=\operatorname{Law}(Z)$,
$P_r=\operatorname{Law}(Z+r)$, and $\lambda=e^{-a}$.  Then
\begin{align}
 K(t)
 &=\log\!\left[
 \frac{1-\lambda}{1+\lambda}\left\{
 \frac{e^{tar}+\lambda^r e^{-tar}}{1-\lambda}
 +\sum_{j=1}^{r-1}\lambda^j e^{ta(r-2j)}\right\}\right],\notag\\
 \rho(r)
 &=\EE_{P_r}L_r
 =a\left\{r+(\lambda^r-1)\operatorname{csch}(a)\right\},\notag\\
 \Ga(t)
 &=a\left\{r+(\lambda^r-1)\operatorname{csch}(a)\right\}t(t+1)-K(t).
 \label{eq:discrete-laplace-transform}
\end{align}
The finite sum is empty when $r=1$.
\item For discrete Gaussian noise with scale $\sigma>0$ \citep{Canonne2020discrete}, let
$Z$
have mass function
\[
 \Pbb(Z=k)
 =\frac{e^{-k^2/(2\sigma^2)}}{\Theta_\sigma(0)},
 \qquad
 \Theta_\sigma(u):=\sum_{k\in\mathbb Z}
 e^{-(k-u)^2/(2\sigma^2)}.
\]
Put $Q=\operatorname{Law}(Z)$ and $P_r=\operatorname{Law}(Z+r)$.
Then $\rho(r)=\EE_{P_r}L_r=r^2/(2\sigma^2)$, and
\begin{align}
 K(t)
 &=\frac{r^2}{2\sigma^2}t(t+1)
 +\log\frac{\Theta_\sigma((1+t)r)}{\Theta_\sigma(0)},\notag\\
 \Ga(t)
 &=\log\frac{\Theta_\sigma(0)}
 {\Theta_\sigma((1+t)r)}=4e^{-2\pi^2\sigma^2}\sin^2(\pi(1+t)r)
 +O(e^{-4\pi^2\sigma^2}),
 \label{eq:discrete-gaussian-gap}
\end{align}
uniformly in $t$ as $\sigma\to\infty$.  The gap is $1/r$-periodic in $t$
and vanishes exactly when $(1+t)r\in\mathbb Z$; it is positive otherwise.
\end{enumerate}
\end{proposition}

\begin{figure}[!htp]
\centering
\includegraphics[width=.85\linewidth]{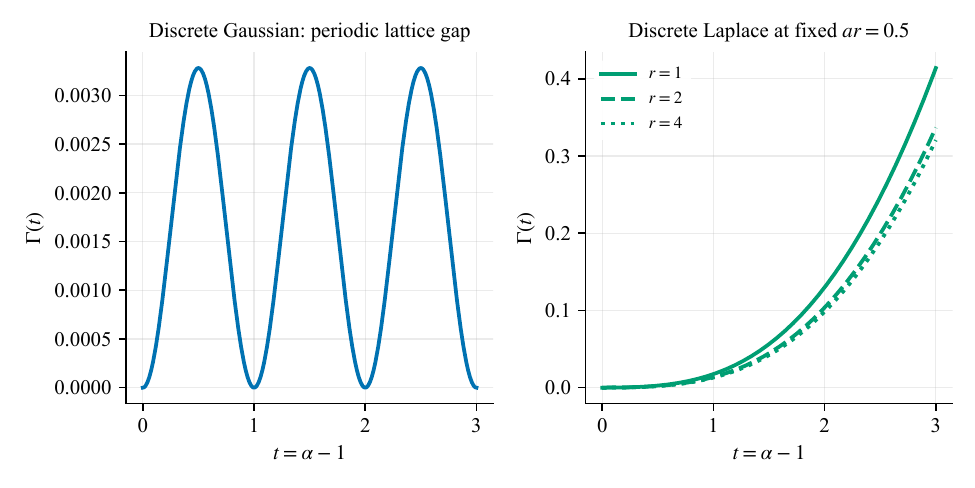}
\caption{The exact zCDP--RDP gaps $\Ga(t)$ of two lattice mechanisms are
compared over the same range of the order offset.  The periodic exact zeros
for discrete Gaussian noise and the positive discrete Laplace gaps for $t>0$
show how the accuracy of zCDP can vary with the RDP order.
Left: discrete Gaussian noise uses $r=1$ and
$\sigma=0.60$, with exact zeros whenever $(1+t)r\in\mathbb Z$.
Right: discrete Laplace noise uses the labeled shifts $r\in\{1,2,4\}$ at
fixed $ar=0.5$, so $a=0.5/r$; the curves are computed from
Proposition~\ref{prop:discrete-examples}.  Both panels use $0\leq t\leq3$ but
retain separate vertical scales.}
\label{fig:mechanisms}
\end{figure}

The discrete Gaussian gap is periodic in $t$, with period $1/r$.
Thus the same zCDP parameter is exact at some RDP orders and loses
information at others. The expansion in \eqref{eq:discrete-gaussian-gap}
also shows that this loss is exponentially small as $\sigma$ grows.
Figure~\ref{fig:mechanisms} fixes $\sigma=0.60$ so that the periodic
correction and its zeros remain visible, and compares it with the
discrete Laplace gaps computed from \eqref{eq:discrete-laplace-transform}.

\FloatBarrier

\section{Experiments}
\label{sec:real-data-experiments}

We test whether the predictions of Section~\ref{sec:practical-implications}
persist when both notions are calibrated to the same privacy level
$(\eps_0,\delta_0)=(3,10^{-5})$, and whether they affect utility, on an ACS
table release and on DP-SGD. Both accountants convert an RDP curve
$U=(U_\alpha)_{\alpha>1}$ to $(\eps,\delta)$-DP with
\citet{balle2020hypothesis}, optimized over the order:
\begin{equation}\label{eq:common-empirical-conversion}
 \mathcal C_\delta[U]
 :=\inf_{\alpha>1}\left\{
 U_\alpha
 -\frac{\log\delta+\log\alpha}{\alpha-1}
 +\log\frac{\alpha-1}{\alpha}
 \right\}.
\end{equation}
The RDP accountant applies $\mathcal C_\delta$ to the exact curve, and the
zCDP accountant to the zCDP envelope. Both apply the same conversion rule, isolating the effect of compression on noise calibration.

\subsection{ACS-sized Workload}
\label{subsec:acs-composition}

For each of the 20 states with the most counties, the workload consists of
five county-level median earnings by educational attainment from the ACS
\citep{Census2020ACSST5Y2020.S1501}, giving $d_i=5m_i$ queries for a state
with $m_i$ complete counties ($330\leq d_i\leq1060$), so
by~\eqref{eq:gap-additivity} the total gap is $d_i$ times the one-query gap.
This is a workload illustration built on published ACS
aggregates, not a Census person-level DP mechanism
(Appendix~\ref{app:acs-details}). For each state and noise family (Gaussian,
Gaussian-mixture, Laplace), we find the smallest per-query noise variance that
certifies $(\eps_0,\delta_0)$ under \eqref{eq:common-empirical-conversion},
$\operatorname{Var}(Z_{\rm RDP})$ and $\operatorname{Var}(Z_{\rm zCDP})$. The
noise is centered and independent across queries, so the per-query MSE equals
this variance; we report
\[
 \Delta\mathrm{MSE}:=\operatorname{Var}(Z_{\rm zCDP})
 -\operatorname{Var}(Z_{\rm RDP}).
\]

\begin{figure}[t]
\centering
\includegraphics[width=.74\linewidth]{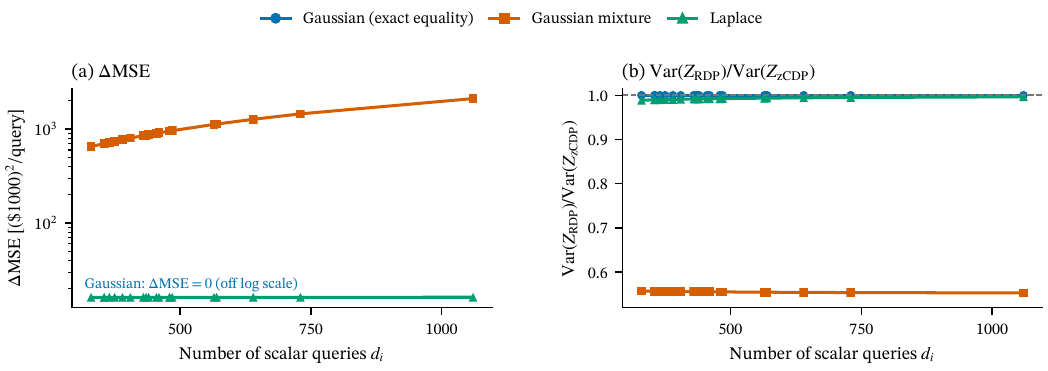}
\caption{Noise variance required at $(\eps_0,\delta_0)=(3,10^{-5})$ on
ACS-sized workloads. Left: $\Delta\mathrm{MSE}$ against the number of queries
$d_i$ (log scale; Gaussian is exactly $0$ and omitted). Right:
$\operatorname{Var}(Z_{\rm RDP})/\operatorname{Var}(Z_{\rm zCDP})$; smaller
means that retaining the RDP curve saves more noise. Release-level
distributions are in Figure~\ref{fig:acs-state-violins}.}
\label{fig:acs-accounting-utility}
\end{figure}

\textbf{Result.}
For Gaussian noise the two calibrations coincide ($\Delta\mathrm{MSE}=0$),
since the Gaussian RDP curve equals its envelope. For the Gaussian mixture,
retaining the RDP curve needs only $0.552$--$0.557$ of the zCDP variance,
about $45\%$ lower noise variance, with $\Delta\mathrm{MSE}$ from $652.3$ to $2115.7$ in
units of $(\$1000)^2$ per query. For Laplace noise the ratio is
$0.989$--$0.997$ and $\Delta\mathrm{MSE}$ is $16.170$--$16.205$. This ordering
matches the zero, quadratic, and quartic gaps of
Proposition~\ref{prop:common-mechanisms}, although the ratios also depend on
the expansion coefficients and the calibration. Monte Carlo releases with
$20{,}000$ paired draws per state and family agree with the analytic
$\Delta\mathrm{MSE}$ within $0.158\%$. Results at $(\eps,\delta)=(5,10^{-5})$
are in Appendix~\ref{app:acs-alternative-privacy}.

\begin{figure}[!htbp]
\centering
\includegraphics[width=.98\linewidth]{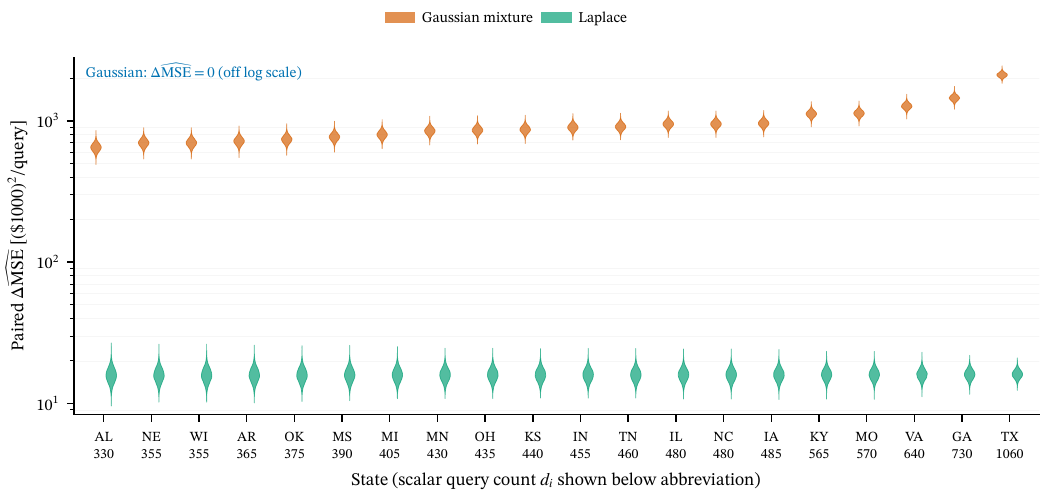}
\caption{The paired $\Delta\widehat{\mathrm{MSE}}$ distributions of the
two non-Gaussian families are compared across the 20 ACS
workloads.  Their separation reproduces the ordering in
Figure~\ref{fig:acs-accounting-utility}: Gaussian-mixture values exceed
Laplace values.  The horizontal axis gives states ordered by $d_i$; the
logarithmic vertical axis gives $\Delta\widehat{\mathrm{MSE}}$ in units of
$(\$1000)^2$ per query.  At each state, each family violin contains $20{,}000$
paired draws.  Gaussian is identically zero and is omitted from the
logarithmic axis.  The violin widths show the distribution of single-release
outcomes, not confidence intervals.}
\label{fig:acs-state-violins}
\end{figure}
\FloatBarrier

\subsection{DP-SGD}
\label{subsec:controlled-dpsgd}

DP-SGD is where zCDP's simple bookkeeping is most attractive and, by
Proposition~\ref{prop:subsampled-gaussian}, where it loses the most. We train
a two-convolutional-layer CNN on Fashion-MNIST \citep{xiao2017fashion} in
Opacus \citep{yousefpour2021opacus}, with Poisson sampling, per-example
clipping at $C=1$, and the sampled-Gaussian RDP bound
\citep{abadi2016deep,mironov2019sampled}. For each sampling rate $q$ and number
of steps $n$, we compute the smallest noise multiplier that certifies
$(\eps_0,\delta_0)$ under each notion and train with it; a third,
``clipping-only'' run adds no noise. One sweep fixes $n=500$ and varies
$q\in\{0.005,0.0075,0.01,0.02,0.05\}$; the other fixes $q=0.01$ and varies
$n\in\{250,500,1000,1500,2000\}$. The privacy guarantees rest on per-step RDP
bounds with adaptive composition; Proposition~\ref{prop:subsampled-gaussian}
only explains the trend. Each setting uses five seeds; within a seed, all runs
share initialization and minibatches, and the RDP and zCDP runs share noise
draws, so we report Student-$t$
intervals over paired differences. Optimizer settings are in
Appendix~\ref{app:reproducibility-details}.

\begin{figure}[t]
\centering
\includegraphics[width=.82\linewidth]{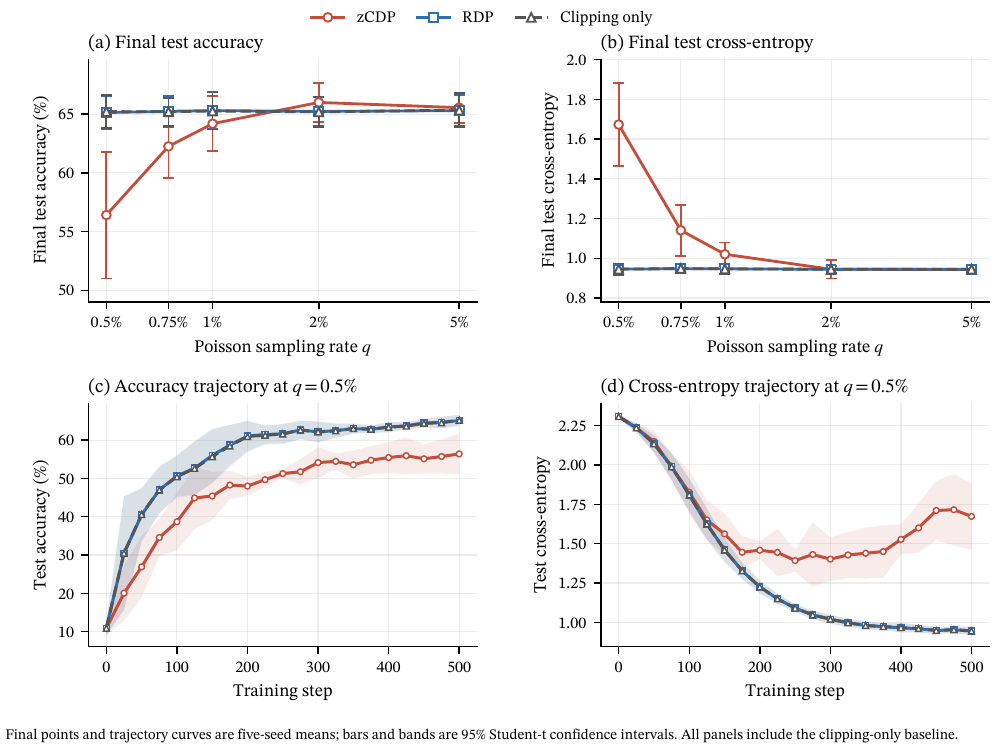}
\caption{The sampling-rate sweep at $n=500$ and $(\eps,\delta)=(3,10^{-5})$.
\textup{(a)}--\textup{(b)}: final test accuracy and cross-entropy of the
RDP-calibrated, zCDP-calibrated, and clipping-only runs as $q$ varies, with
the noise recalibrated for each $q$. \textup{(c)}--\textup{(d)}: the
corresponding training histories at $(q,n)=(0.005,500)$. Curves are
five-seed means; bars and bands are 95\% confidence intervals.}
\label{fig:dpsgd-q-main}
\label{fig:dpsgd-q-sweep}
\end{figure}

\paragraph{Sampling Rate.}
With RDP calibration, test accuracy stays near $65.2\%$ and cross-entropy
near $0.95$ across the $q$ sweep (Figure~\ref{fig:dpsgd-q-main}(a)--(b)), close to
the clipping-only run. With zCDP calibration, accuracy rises
from $56.408\%$ at $q=0.005$ to about $66\%$ for $q\geq0.02$ and cross-entropy falls
from $1.673$ to $0.943$; for $q\geq0.02$ the three methods agree within
seed-to-seed variability. At $q=0.005$, RDP calibration reaches $65.138\%$, a
paired gain of $8.73\pm6.10$ percentage points. The zCDP penalty shrinks as
$q$ grows, as
Proposition~\ref{prop:subsampled-gaussian} predicts. The seedwise log-$q$
trend is significant for test loss ($p=7.06\times10^{-4}$) and accuracy
($p=0.0108$). Since changing $q$ at fixed $n$
also changes the batch size and the number of epochs ($nq$ from $2.5$ to
$25$), this trend is consistent with, but not a direct consequence of,
Proposition~\ref{prop:subsampled-gaussian}.

\paragraph{Composition Count.}
In Figure~\ref{fig:dpsgd-composition-sweep}(a)--(b), RDP accuracy rises from
$61.866\%$ at $n=250$ to $70.522\%$ at $n=2000$; across all nine settings of
both sweeps, RDP accuracy is within $0.08$ percentage points of the
clipping-only run. zCDP accuracy peaks at
$66.354\%$ at $n=1000$ and falls to $64.600\%$ at $n=2000$, while its
cross-entropy falls from $1.120$ to $1.020$ at $n=500$ and then rises to
$2.031$. The paired test-loss difference grows with $n$ (seedwise log-$n$
trend $p=2.9\times10^{-5}$). The paired accuracy difference is not monotone,
and neither its trend ($p=0.161$) nor the endpoint contrast ($p=0.146$) is
significant: each horizon is recalibrated, and optimization turns noise into
accuracy nonlinearly, so cross-entropy is the cleaner measure here
(training histories in Figure~\ref{fig:dpsgd-composition-sweep}(c)--(d)). Full results
and sweeps at
$(\eps,\delta)=(5,10^{-5})$ are in Appendix~\ref{app:controlled-dpsgd}.

\section{Discussion and Limitations}
\label{sec:discussion}

Our results give a simple rule for choosing a privacy notion. The two
one-sided $(\eps,\delta)$-DP curves, the trade-off functions, the PLD, and,
when it is finite at one order, the exact RDP curve carry the same information
about a pair (Theorem~\ref{thm:equivalence}); choosing among them affects
interpretation and computation, not information. The usual two-sided
$(\eps,\delta)$ curve discards directional information, whereas zCDP is lossy even for
a fixed pair.
Its loss, measured by $\Ga$, is zero for Gaussian noise and grows for
mixture-type noise and under Poisson subsampling at small $q$
(Figure~\ref{fig:privacy-overview}). zCDP is therefore lossless for Gaussian noise
without subsampling; otherwise, retaining one of the equivalent curves can
reduce the noise needed for the same certified privacy.

For future work, a natural next step is to extend the analysis to
mechanism-level guarantees and adaptive composition.
Theorem~\ref{thm:equivalence} concerns one fixed pair: a mechanism-level RDP
bound takes the worst pair at each order and need not be the curve of any
single pair, so the theorem does not apply to it directly. The composition
results assume independent components and add their zCDP parameters, while
the local expansions hold for small shifts and bounded orders.
The gap isolates the RDP-to-zCDP step; studying the additional loss in
conversion to $(\eps,\delta)$-DP would help explain why the final bound can
be loose even when $\Ga=0$.
The experiments cover three noise families, ACS-sized workloads, and one
Fashion-MNIST model with five seeds. Broader experiments could examine
whether the observed utility patterns persist across other datasets,
optimizers, and deployments.
Another direction is rigorous numerical evaluation of $\rho$ when the
maximizing order has no closed form.

\bibliographystyle{abbrvnat}
\bibliography{ref}

@article{su2026sieve,
  title={A Sieve-Accelerated Quadrature Method for Exact Privacy Accounting in the 2020 {US Decennial Census}},
  author={Su, Buxin and Su, Weijie and Wang, Chendi},
  journal={arXiv preprint arXiv:2606.29835},
  year={2026}
}

@inproceedings{bun2016concentrated,
  title={Concentrated Differential Privacy: Simplifications, Extensions, and Lower Bounds},
  author={Bun, Mark and Steinke, Thomas},
  booktitle={Theory of Cryptography: 14th International Conference, {TCC} 2016-B, Part I},
  series={Lecture Notes in Computer Science},
  volume={9985},
  pages={635--658},
  year={2016},
  publisher={Springer},
  doi={10.1007/978-3-662-53641-4_24}
}

@inproceedings{Canonne2020discrete,
 author = {Canonne, Cl\'{e}ment L. and Kamath, Gautam and Steinke, Thomas},
 booktitle = {Advances in Neural Information Processing Systems},
 pages = {15676--15688},
 publisher = {Curran Associates, Inc.},
 title = {The Discrete {G}aussian for Differential Privacy},
 volume = {33},
 year = {2020},
}

@article{Dong2022Gaussian,
  title={Gaussian Differential Privacy},
  author={Dong, Jinshuo and Roth, Aaron and Su, Weijie J.},
  journal={Journal of the Royal Statistical Society: Series B (Statistical Methodology)},
  volume={84},
  number={1},
  pages={3--37},
  year={2022},
  publisher={Wiley},
  doi={10.1111/rssb.12454}
}

@article{dwork2014algorithmic,
  title={The algorithmic foundations of differential privacy},
  author={Dwork, Cynthia and Roth, Aaron},
  journal={Foundations and Trends in Theoretical Computer Science},
  volume={9},
  number={3--4},
  pages={211--407},
  year={2014},
  doi={10.1561/0400000042}
}

@inproceedings{dwork2006differential,
  author       = {Cynthia Dwork},
  title        = {Differential Privacy},
  booktitle    = {Automata, Languages and Programming, 33rd International Colloquium,
                  {ICALP} 2006, Venice, Italy, July 10-14, 2006, Proceedings, Part {II}},
  series       = {Lecture Notes in Computer Science},
  volume       = {4052},
  pages        = {1--12},
  publisher    = {Springer},
  year         = {2006},
  doi          = {10.1007/11787006_1}
}

@article{abowd20222020,
  title={The 2020 {C}ensus {D}isclosure {A}voidance {S}ystem {T}op{D}own {A}lgorithm},
  author={Abowd, John M and Ashmead, Robert and Cumings-Menon, Ryan and Garfinkel, Simson and Heineck, Micah and Heiss, Christine and Johns, Robert and Kifer, Daniel and Leclerc, Philip and Machanavajjhala, Ashwin and Moran, Brett and Sexton, William and Spence, Matthew and Zhuravlev, Pavel},
  journal={Harvard Data Science Review},
  note={Special Issue 2},
  year={2022},
  doi={10.1162/99608f92.529e3cb9}
}

@article{kifer2022bayesian,
  title={Bayesian and frequentist semantics for common variations of differential privacy: Applications to the 2020 {C}ensus},
  author       = {Daniel Kifer and
                  John M. Abowd and
                  Robert Ashmead and
                  Ryan Cumings{-}Menon and
                  Philip Leclerc and
                  Ashwin Machanavajjhala and
                  William Sexton and
                  Pavel Zhuravlev},
  journal={arXiv preprint arXiv:2209.03310},
  year={2022}
}

@article{kairouz2017composition,
  author       = {Peter Kairouz and
                  Sewoong Oh and
                  Pramod Viswanath},
  title        = {The Composition Theorem for Differential Privacy},
  journal      = {{IEEE} Trans. Inf. Theory},
  volume       = {63},
  number       = {6},
  pages        = {4037--4049},
  year         = {2017}
}

@article{balle2020privacy,
  author       = {Borja Balle and
                  Gilles Barthe and
                  Marco Gaboardi},
  title        = {Privacy Profiles and Amplification by Subsampling},
  journal      = {Journal of Privacy and Confidentiality},
  volume       = {10},
  number       = {1},
  year         = {2020},
  doi          = {10.29012/jpc.726}
}

@inproceedings{bun2018composable,
  title={Composable and Versatile Privacy via Truncated {CDP}},
  author={Bun, Mark and Dwork, Cynthia and Rothblum, Guy N and Steinke, Thomas},
  booktitle={Proceedings of the 50th Annual ACM SIGACT Symposium on Theory of Computing},
  pages={74--86},
  year={2018},
  publisher={Association for Computing Machinery},
  doi={10.1145/3188745.3188946}
}

@inproceedings{gopi2021numerical,
  author       = {Gopi, Sivakanth and
                  Lee, Yin Tat and
                  Wutschitz, Lukas},
  title        = {Numerical Composition of Differential Privacy},
  booktitle    = {Advances in Neural Information Processing Systems},
  volume       = {34},
  pages        = {11631--11642},
  publisher    = {Curran Associates, Inc.},
  year         = {2021}
}

@inproceedings{Koskela2020computing,
  author       = {Koskela, Antti and
                  J{\"{a}}lk{\"{o}}, Joonas and
                  Honkela, Antti},
  title        = {Computing Tight Differential Privacy Guarantees Using {FFT}},
  booktitle    = {Proceedings of the Twenty Third International Conference on Artificial Intelligence and Statistics},
  series       = {Proceedings of Machine Learning Research},
  volume       = {108},
  pages        = {2560--2569},
  publisher    = {{PMLR}},
  year         = {2020},
}

@article{doroshenko2022connect,
  title={Connect the Dots: Tighter Discrete Approximations of Privacy Loss Distributions},
  author={Doroshenko, Vadym and Ghazi, Badih and Kamath, Pritish and Kumar, Ravi and Manurangsi, Pasin},
  journal={Proceedings on Privacy Enhancing Technologies},
  volume={2022},
  number={4},
  pages={552--570},
  year={2022},
  doi={10.56553/popets-2022-0122},
}

@inproceedings{ghazi2022faster,
  title={Faster Privacy Accounting via Evolving Discretization},
  author={Ghazi, Badih and Kamath, Pritish and Kumar, Ravi and Manurangsi, Pasin},
  booktitle={Proceedings of the 39th International Conference on Machine Learning},
  series={Proceedings of Machine Learning Research},
  volume={162},
  pages={7470--7483},
  publisher={PMLR},
  year={2022},
}

@inproceedings{balle2020hypothesis,
  author       = {Borja Balle and
                  Gilles Barthe and
                  Marco Gaboardi and
                  Justin Hsu and
                  Tetsuya Sato},
  editor       = {Silvia Chiappa and
                  Roberto Calandra},
  title        = {Hypothesis Testing Interpretations and {R{\'e}nyi} Differential Privacy},
  booktitle    = {Proceedings of the Twenty Third International Conference on Artificial Intelligence and Statistics},
  series       = {Proceedings of Machine Learning Research},
  volume       = {108},
  pages        = {2496--2506},
  publisher    = {{PMLR}},
  year         = {2020},
}

@misc{Census2020ACSST5Y2020.S1501,
  author       = {{U.S. Census Bureau}},
  title        = {{S1501: Educational Attainment (American Community Survey
                  5-Year, 2020)}},
  year         = {2020},
  publisher    = {U.S. Census Bureau},
  howpublished = {\url{https://data.census.gov/table/ACSST5Y2020.S1501}},
  urldate      = {2025-10-17},
  note         = {Table ID: S1501; Version: 2016--2020 5-year estimates}
}

@inproceedings{mironov2017renyi,
  title={{R{\'e}nyi} Differential Privacy},
  author={Mironov, Ilya},
  booktitle={2017 IEEE 30th Computer Security Foundations Symposium (CSF)},
  pages={263--275},
  year={2017},
  publisher={IEEE},
  doi={10.1109/CSF.2017.11}
}

@article{su2025statistical,
  title={A {S}tatistical {V}iewpoint on {D}ifferential {P}rivacy: {H}ypothesis {T}esting, {R}epresentation, and {B}lackwell's Theorem},
  author={Su, Weijie J},
  journal={Annual Review of Statistics and Its Application},
  volume={12},
  number={1},
  pages={157--175},
  year={2025},
  publisher={Annual Reviews}
}

@inproceedings{zhu2022optimal,
  title={Optimal accounting of differential privacy via characteristic function},
  author={Zhu, Yuqing and Dong, Jinshuo and Wang, Yu-Xiang},
  booktitle={Proceedings of The 25th International Conference on Artificial Intelligence and Statistics},
  series={Proceedings of Machine Learning Research},
  volume={151},
  pages={4782--4817},
  year={2022},
  publisher={PMLR}
}

@article{su2025census,
author = {Buxin Su  and Weijie J. Su  and Chendi Wang },
title = {The 2020 {US} {D}ecennial {C}ensus is more private than you (might) think},
journal = {Proceedings of the National Academy of Sciences},
volume = {122},
number = {45},
pages = {e2500337122},
year = {2025},
doi = {10.1073/pnas.2500337122},
eprint = {https://www.pnas.org/doi/pdf/10.1073/pnas.2500337122}}

@inproceedings{wang2019subsampled,
  title={Subsampled {R}{\'e}nyi Differential Privacy and Analytical Moments Accountant},
  author={Wang, Yu-Xiang and Balle, Borja and Kasiviswanathan, Shiva Prasad},
  booktitle={Proceedings of the Twenty-Second International Conference on Artificial Intelligence and Statistics},
  pages={1226--1235},
  year={2019},
  editor={Chaudhuri, Kamalika and Sugiyama, Masashi},
  volume={89},
  series={Proceedings of Machine Learning Research},
  publisher={PMLR},
}

@inproceedings{zhu2019poission,
  title={Poission Subsampled {R}{\'e}nyi Differential Privacy},
  author={Zhu, Yuqing and Wang, Yu-Xiang},
  booktitle={Proceedings of the 36th International Conference on Machine Learning},
  pages={7634--7642},
  year={2019},
  editor={Chaudhuri, Kamalika and Salakhutdinov, Ruslan},
  volume={97},
  series={Proceedings of Machine Learning Research},
  publisher={PMLR},
}

@inproceedings{abadi2016deep,
  title={Deep learning with differential privacy},
  author={Abadi, Mart{\'i}n and Chu, Andy and Goodfellow, Ian and McMahan, H Brendan and Mironov, Ilya and Talwar, Kunal and Zhang, Li},
  booktitle={Proceedings of the 2016 ACM SIGSAC conference on computer and communications security},
  pages={308--318},
  year={2016}
}

@inproceedings{harrison2025exact,
  title={Exact {zCDP} Characterizations for Fundamental Differentially Private Mechanisms},
  author={Harrison, Charlie and Manurangsi, Pasin},
  booktitle={7th Symposium on Foundations of Responsible Computing (FORC 2026)},
  series={Leibniz International Proceedings in Informatics (LIPIcs)},
  volume={368},
  pages={3:1--3:18},
  publisher={Schloss Dagstuhl -- Leibniz-Zentrum f{\"u}r Informatik},
  year={2026},
  doi={10.4230/LIPIcs.FORC.2026.3},
}

@inproceedings{chourasia2025laplace,
  title={Laplace Transform Interpretation of Differential Privacy},
  author={Chourasia, Rishav and Javaid, Uzair and Sikdar, Biplab},
  booktitle={6th Symposium on Foundations of Responsible Computing (FORC 2025)},
  series={Leibniz International Proceedings in Informatics (LIPIcs)},
  volume={329},
  pages={11:1--11:23},
  publisher={Schloss Dagstuhl -- Leibniz-Zentrum f{\"u}r Informatik},
  year={2025},
  doi={10.4230/LIPIcs.FORC.2025.11},
  note={Extended version: arXiv:2411.09142},
}

@article{asoodeh2021three,
  title={Three Variants of Differential Privacy: Lossless Conversion and Applications},
  author={Asoodeh, Shahab and Liao, Jiachun and Calmon, Flavio P. and Kosut, Oliver and Sankar, Lalitha},
  journal={IEEE Journal on Selected Areas in Information Theory},
  volume={2},
  number={1},
  pages={208--222},
  year={2021},
  doi={10.1109/JSAIT.2021.3054692}
}

@inproceedings{riess2026optimal,
  title={Optimal conversion from {R{\'e}nyi} Differential Privacy to {$f$}-Differential Privacy},
  author={Riess, Anneliese and Gomez, Juan Felipe and {du Pin Calmon}, Flavio and Schnabel, Julia Anne and Kaissis, Georgios},
  booktitle={Proceedings of the 43rd International Conference on Machine Learning},
  year={2026},
  note={Accepted at ICML 2026; arXiv:2602.04562},
  url={https://arxiv.org/abs/2602.04562},
}

@article{xiao2017fashion,
  title={Fashion-{MNIST}: A Novel Image Dataset for Benchmarking Machine Learning Algorithms},
  author={Xiao, Han and Rasul, Kashif and Vollgraf, Roland},
  journal={arXiv preprint arXiv:1708.07747},
  year={2017},
  doi={10.48550/arXiv.1708.07747}
}

@article{yousefpour2021opacus,
  title={Opacus: User-Friendly Differential Privacy Library in {PyTorch}},
  author={Yousefpour, Ashkan and Shilov, Igor and Sablayrolles, Alexandre and Testuggine, Davide and Prasad, Karthik and Malek, Mani and Nguyen, John and Ghosh, Sayan and Bharadwaj, Akash and Zhao, Jessica and Cormode, Graham and Mironov, Ilya},
  journal={arXiv preprint arXiv:2109.12298},
  year={2021},
  doi={10.48550/arXiv.2109.12298}
}

@article{mironov2019sampled,
  title={{R{\'e}nyi} Differential Privacy of the Sampled {G}aussian Mechanism},
  author={Mironov, Ilya and Talwar, Kunal and Zhang, Li},
  journal={arXiv preprint arXiv:1908.10530},
  year={2019},
  doi={10.48550/arXiv.1908.10530}
}

@inproceedings{kaissis2024beyond,
  author       = {Georgios Kaissis and
                  Stefan Kolek and
                  Borja Balle and
                  Jamie Hayes and
                  Daniel Rueckert},
  title        = {Beyond the Calibration Point: Mechanism Comparison in Differential
                  Privacy},
  booktitle    = {Forty-first International Conference on Machine Learning, {ICML} 2024,
                  Vienna, Austria, July 21-27, 2024},
  series       = {Proceedings of Machine Learning Research},
  volume       = {235},
  pages        = {22840--22860},
  publisher    = {{PMLR} / OpenReview.net},
  year         = {2024},
}

@article{gomez2025reporting,
  title={Gaussian DP for reporting differential privacy guarantees in machine learning},
  author={Gomez, Juan Felipe and Kulynych, Bogdan and Kaissis, Georgios and Calmon, Flavio P and Hayes, Jamie and Balle, Borja and Honkela, Antti},
  journal={arXiv preprint arXiv:2503.10945},
  year={2025}
}

@inproceedings{koskela2021tight,
  author       = {Antti Koskela and
                  Joonas J{\"{a}}lk{\"{o}} and
                  Lukas Prediger and
                  Antti Honkela},
  editor       = {Arindam Banerjee and
                  Kenji Fukumizu},
  title        = {Tight Differential Privacy for Discrete-Valued Mechanisms and for
                  the Subsampled Gaussian Mechanism Using {FFT}},
  booktitle    = {The 24th International Conference on Artificial Intelligence and Statistics,
                  {AISTATS} 2021, April 13-15, 2021, Virtual Event},
  series       = {Proceedings of Machine Learning Research},
  volume       = {130},
  pages        = {3358--3366},
  publisher    = {{PMLR}},
  year         = {2021},
  url          = {http://proceedings.mlr.press/v130/koskela21a.html},
  bibsource    = {dblp computer science bibliography, https://dblp.org}
}

\clearpage
\appendix

\pdfbookmark[0]{Appendix contents}{appendix.contents}
\section*{Appendix Contents}
\begingroup
\small
\setlength{\parskip}{0pt}
\makeatletter
\renewcommand{\@pnumwidth}{2em}
\renewcommand{\@tocrmarg}{3em}
\newcommand{\appcontentssection}[2]{%
  \addvspace{0.65\baselineskip}%
  \@dottedtocline{1}{0pt}{1.8em}%
    {\bfseries\hyperref[#1]{\numberline{\ref*{#1}}#2}}%
    {\bfseries\pageref{#1}}%
}
\newcommand{\appcontentssubsection}[2]{%
  \addvspace{0.15\baselineskip}%
  \@dottedtocline{2}{1.8em}{2.6em}%
    {\hyperref[#1]{\numberline{\ref*{#1}}#2}}%
    {\pageref{#1}}%
}
\appcontentssection{app:proof-information-equivalence}{Proofs for Section~\ref*{sec:information-equivalence}}
\appcontentssubsection{app:proof-equivalence}{Proof of Theorem~\ref*{thm:equivalence}}
\appcontentssubsection{app:symmetric-additive}{Proof of Corollary~\ref*{cor:symmetric-additive}}
\appcontentssubsection{app:proof-finite-order-needed}{Proof of Example~\ref*{ex:finite-order-needed}}
\appcontentssubsection{app:proof-one-side-counterexample}{Proof of Example~\ref*{ex:one-side-counterexample}}
\appcontentssection{app:proof-gap-composition}{Proofs for Section~\ref*{sec:defect}}
\appcontentssubsection{app:proof-cher}{Proof of Lemma~\ref*{lem:cher}}
\appcontentssubsection{subsec:proof-pairwise-gap}{Proof of the conversion bounds in Section~\ref*{subsec:optimized-certificates}}
\appcontentssubsection{app:proof-optimized-gap}{Proof of Proposition~\ref*{prop:optimized-gap}}
\appcontentssubsection{app:proof-plrv-quartic}{Proof of Lemma~\ref*{lem:plrv-quartic}}
\appcontentssubsection{app:proof-local-expansion}{Proof of Theorem~\ref*{thm:sensitivity-gap}}
\appcontentssubsection{app:proof-composition}{Proof of Theorem~\ref*{thm:composition-extension}}
\appcontentssection{app:deferred-section4}{Proofs for Section~\ref*{sec:practical-implications}}
\appcontentssubsection{app:proof-localization}{Proof of Lemma~\ref*{lem:localization}}
\appcontentssubsection{app:proof-smooth-cumulants}{Proof of~(\ref*{eq:smooth-third-cumulant}) and~(\ref*{eq:smooth-fourth-cumulant})}
\appcontentssubsection{app:proof-common-mechanisms}{Proof of Proposition~\ref*{prop:common-mechanisms}}
\appcontentssubsection{app:proof-subsampled-gaussian}{Proof of Proposition~\ref*{prop:subsampled-gaussian}}
\appcontentssubsection{app:proof-zcdp-budget-allocations}{Proof of Proposition~\ref*{prop:zcdp-budget-allocations}}
\appcontentssubsection{app:proof-student-general}{Proof of Proposition~\ref*{prop:student-general}}
\appcontentssubsection{app:proof-gg-single}{Proof of Proposition~\ref*{prop:gg-single}}
\appcontentssubsection{app:proof-discrete-examples}{Proof of Proposition~\ref*{prop:discrete-examples}}
\appcontentssection{app:section4-additional-results}{Supplementary Figures and Numerical Details}
\appcontentssubsection{app:numerics}{Computational setup}
\appcontentssubsection{app:alternative-order-robustness}{Results at an alternative reporting order}
\appcontentssubsection{app:acs-details}{ACS workloads and additional results}
\appcontentssubsection{app:controlled-dpsgd}{DP-SGD settings and additional results}
\makeatother
\endgroup
\clearpage

\section{Proofs for Section~\ref{sec:information-equivalence}}
\label{app:proof-information-equivalence}

\subsection{Proof of Theorem~\ref{thm:equivalence}}
\label{app:proof-equivalence}

\begin{proof}[Proof of Theorem~\ref{thm:equivalence}]
For arbitrary $P$ and $Q$, take $\mu=P+Q$ and define
\[
 \ell(x)=\log\frac{dP/d\mu(x)}{dQ/d\mu(x)}.
\]
Set $\ell=+\infty$ where only $dP/d\mu$ is positive and $\ell=-\infty$
where only $dQ/d\mu$ is positive. For $X\sim P$, the privacy loss
$L=\ell(X)$ has no mass at $-\infty$; its mass at $+\infty$ is the mass of
the part of $P$ singular to $Q$. We use $e^{-\infty}=0$.
For $\gamma\geq1$, write
$H_\gamma(P\|Q)=\sup_A\{P(A)-\gamma Q(A)\}$, so that
$\delta_{P\|Q}(\eps)=H_{e^\eps}(P\|Q)$ for $\eps\geq0$.

Let $\mathcal E_{P,Q}$ be the closed convex testing-error region consisting of
pairs
\[
 \bigl(\EE_P[\phi],1-\EE_Q[\phi]\bigr),
 \qquad 0\leq\phi\leq1.
\]
The trade-off function $T_{P,Q}$ is the lower boundary of
$\mathcal E_{P,Q}$.  Since $(a,b)\mapsto(b,a)$ maps $\mathcal E_{P,Q}$
onto $\mathcal E_{Q,P}$, $T_{Q,P}$ is the generalized inverse of
$T_{P,Q}$, so $T_{P,Q}$ determines the pair $(T_{P,Q},T_{Q,P})$.
For $\gamma\geq1$, the two
hockey-stick divergences are equivalently the support values
\[
 H_\gamma(P\|Q)=1-\inf_{(a,b)\in\mathcal E_{P,Q}}(a+\gamma b),
 \qquad
 H_\gamma(Q\|P)=1-\inf_{(a,b)\in\mathcal E_{P,Q}}(\gamma a+b).
\]
As $\gamma$ ranges over $[1,\infty)$, the two families cover all positive supporting slopes. Convex duality therefore shows that the pair of one-sided $(\eps,\delta)$-DP curves and the pair of trade-off functions determine one another.

We next connect the two one-sided curves to privacy loss.  The hockey-stick
identity gives, for $x\geq0$,
\[
 \delta_{P\|Q}(x)
 =\EE_P\!\left[(1-e^{x-L})_+\right],
 \qquad
 \delta'_{P\|Q,+}(x)
 =-e^x\EE_P\!\left[e^{-L}\mathbf 1_{\{L>x\}}\right].
\]
Consequently,
\[
 P\{L>x\}=\delta_{P\|Q}(x)-\delta'_{P\|Q,+}(x),
 \qquad
 P\{L<-x\}=-e^{-x}\delta'_{Q\|P,+}(x).
\]
These strict-tail identities, together with normalization and
$P\{L=+\infty\}=\lim_{x\to\infty}\delta_{P\|Q}(x)$, recover the law of
$L$ under $P$.  Conversely, that law determines both one-sided curves via
\[
 \delta_{P\|Q}(\eps)=\EE_P[(1-e^{\eps-L})_+],
 \qquad
 \delta_{Q\|P}(\eps)
 =1-\EE_P[e^{-L}]
  +\EE_P[(e^{-L}-e^\eps)_+].
\]
The identity $D_{1+t}(P\|Q)=K(t)/t$ then shows that the privacy-loss law
determines the exact RDP curve.

For the converse, suppose $D_{\alpha_0}(P\|Q)<\infty$ for some
$\alpha_0>1$, and put $c:=\alpha_0-1$.  Since $K$ is convex with $K(0)=0$
and $K(c)<\infty$, $K$ is finite on $[0,c]$, and the RDP curve determines
$K(t)=tD_{1+t}(P\|Q)$ for $0<t<c$.  In particular $P\{L=+\infty\}=0$, so
$L$ is real-valued $P$-a.s.  Fix $t_0\in(0,c)$ and exponentially tilt the law of $L$ by the
density $e^{t_0L-K(t_0)}$.  The tilted law has moment-generating function
\[
 \EE[e^{uL_{t_0}}]
 =\exp\{K(t_0+u)-K(t_0)\},
 \qquad -t_0<u<c-t_0.
\]
This function is known from RDP and is finite on a neighborhood of zero, so
moment-generating-function uniqueness determines the tilted law.  Untilting
by
\[
 P\{L\in B\}
 =e^{K(t_0)}\EE\!\left[e^{-t_0L_{t_0}}
   \mathbf 1_{\{L_{t_0}\in B\}}\right]
\]
recovers the original privacy-loss law.  The preceding formulas recover both
one-sided curves, and convex duality for $\mathcal E_{P,Q}$ then recovers the
pair of trade-off functions.
\end{proof}

\subsection{Proof of Corollary~\ref{cor:symmetric-additive}}
\label{app:symmetric-additive}
\begin{corollary}[Centrally symmetric additive noise]
\label{cor:symmetric-additive}
Let $\theta$ be a real-valued query and let
$\mathcal M(D)=\theta(D)+Z$.  Fix ordered neighboring datasets $D,D'$, put
$s=\theta(D)-\theta(D')\in\R$, and translate both outputs by
$-\theta(D')$.  Suppose $Z\overset{d}{=}-Z$, and set
$P=\operatorname{Law}(s+Z)$, $Q=\operatorname{Law}(Z)$. The reflection $R_s(x)=s-x$ exchanges
$P$ and $Q$, so
\[
 \delta_{P\|Q}(\eps)=\delta_{Q\|P}(\eps),\qquad
 D_\alpha(P\|Q)=D_\alpha(Q\|P),\qquad
 T_{P,Q}=T_{Q,P}.
\]
The common one-sided $(\eps,\delta)$-DP curve, the exact $f$-DP curve, and
the law of $L$ determine one another and determine every exact RDP value
through $D_{1+t}(P\|Q)=K(t)/t$.  Conversely, if
$D_{\alpha_0}(P\|Q)<\infty$ for some real $\alpha_0>1$, the exact RDP curve
recovers the other three notions.
\end{corollary}

\begin{proof}[Proof of Corollary~\ref{cor:symmetric-additive}]
Set $R_s(x)=s-x$.  If $X=s+Z\sim P$, then
$R_s(X)=-Z\sim Q$ by central symmetry.  If $X=Z\sim Q$, then
$R_s(X)=s-Z\sim P$.  Thus $R_s$ exchanges $P$ and $Q$,
$R_s\circ R_s$ is the identity, and $R_s$ preserves $\mu=P+Q$.  By
uniqueness of Radon--Nikodym derivatives, versions may be chosen so that
$\left(dP/d\mu\right)\circ R_s=dQ/d\mu$ and
$\left(dQ/d\mu\right)\circ R_s=dP/d\mu$ $\mu$-almost everywhere.
Consequently,
$\ell\circ R_s=-\ell$, including when $\ell$ takes an extended value.
Transforming measurable sets by $R_s$ now gives
$\delta_{P\|Q}=\delta_{Q\|P}$.  Invariance of R\'enyi divergence under the
bimeasurable transformation $R_s$ similarly gives
$D_\alpha(P\|Q)=D_\alpha(Q\|P)$.
If a randomized test $\phi$ of $P$ against $Q$ has type I and type II errors
$(a,b)$, then $\phi\circ R_s$, viewed as a test of $Q$ against $P$, also has
error pair $(a,b)$.  Applying the reflection argument with $(Q,P)$ in place
of $(P,Q)$ gives
$T_{P,Q}=T_{Q,P}$.

Because the two one-sided $(\eps,\delta)$-DP curves coincide, the common
curve supplies both inputs required by Theorem~\ref{thm:equivalence}.
Theorem~\ref{thm:equivalence} gives the unconditional forward recoveries and,
under the finite-order condition stated in
Corollary~\ref{cor:symmetric-additive}, the converse recovery from the RDP
curve.
\end{proof}

\subsection{Proof of Example~\ref{ex:finite-order-needed}}
\label{app:proof-finite-order-needed}

\begin{proof}[Proof of Example~\ref{ex:finite-order-needed}]
The masses of $Q_w$ sum to one by the definition of
$\operatorname{Li}_w(e^{-1})$. At each $k\geq1$, the privacy loss is
$L_w=k$, which gives the mean and the hockey-stick divergence stated in
the example. For every $t>0$,
\[
 \EE_{X\sim P_w}\!\left[e^{tL_w}\right]
 =\frac{1}{\zeta(w)}\sum_{k=1}^\infty\frac{e^{tk}}{k^w}=+\infty.
\]
Hence the RDP curve is everywhere infinite for every $w>2$, in particular
for $w=3$ and $w=4$, whereas the forward curves differ because
$\delta_{P_3\|Q_3}(1)\neq\delta_{P_4\|Q_4}(1)$.
\end{proof}

\subsection{Proof of Example~\ref{ex:one-side-counterexample}}
\label{app:proof-one-side-counterexample}

\begin{proof}[Proof of Example~\ref{ex:one-side-counterexample}]
The masses on the nonnegative and negative support points sum to
$4c/3=8/11$ and $c/2=3/11$, respectively.  Thus both are full-support
probability mass functions.  If $W_i(k)=P_i\{k\}/Q_i\{k\}$, direct division
gives
\[
\begin{array}{cccc}
 &k\geq1&k=0&k\leq-1\\
 W_A(k)&4&1/4&1/2\\
 W_B(k)&4&1/3&1/3.
\end{array}
\]
For $\gamma=e^\eps\geq1$, recall that
$H_\gamma(P_i\|Q_i)=\delta_{P_i\|Q_i}(\eps)$. Since
$H_\gamma(P_i\|Q_i)=\sum_k q_i(k)(W_i(k)-\gamma)_+$, only $k\geq1$
contributes to the forward divergence when $\gamma\geq1$.  This gives the
common formula in the example.  Similarly,
$H_\gamma(Q_i\|P_i)=\sum_k q_i(k)(1-\gamma W_i(k))_+$ gives the two reverse
formulas.

Writing $F$ for the common forward one-sided $(\eps,\delta)$-DP curve and $R_A,R_B$ for
the reverse one-sided $(\eps,\delta)$-DP curves, direct subtraction gives
\[
 F-R_A=
 \begin{cases}
  c(\gamma-1)/6,&1\leq\gamma\leq2,\\
  c(4-\gamma)/12,&2\leq\gamma\leq4,
 \end{cases}
\]
and
\[
 F-R_B=
 \begin{cases}
  c(\gamma-1)/6,&1\leq\gamma\leq3,\\
  c(4-\gamma)/3,&3\leq\gamma\leq4.
 \end{cases}
\]
All curves vanish after their stated cutoffs, so the pointwise maxima
agree.  The strongest symmetric $f$-DP guarantee depends on the pair only
through this maximum \citep[Propositions~5--6]{Dong2022Gaussian}, so it also
agrees.  Finally,
\[
 \exp(D_2(P_i\|Q_i))=\sum_k q_i(k)W_i(k)^2.
\]
Substituting the table and the three region masses gives $265/88$ for pair
$A$ and $3$ for pair $B$.  The table also shows that both $W_i$ and
$W_i^{-1}$ are bounded by $4$.
\end{proof}

\FloatBarrier
\section{Proofs for Section~\ref{sec:defect}}
\label{app:proof-gap-composition}

\subsection{Proof of Lemma~\ref{lem:cher}}
\label{app:proof-cher}

\begin{lemma}[Hockey-stick Chernoff bound]\label{lem:cher}
For $x\in\R$ and $t>0$,
\[
 (1-e^x)_+\leq\frac{t^t}{(1+t)^{1+t}}e^{-tx}.
\]
\end{lemma}

\begin{proof}[Proof of Lemma~\ref{lem:cher}]
If $x\geq0$, the left side is zero.  Now suppose $x<0$ and set
$y=e^x\in(0,1)$.  The function $y^t(1-y)$ is maximized at
$y=t/(1+t)$, where $y^t(1-y)$ equals $t^t/(1+t)^{1+t}$.  Therefore
\[
 (1-e^x)e^{tx}\leq\frac{t^t}{(1+t)^{1+t}},
\]
which is the stated inequality.
\end{proof}

\subsection{Proof of the conversion bounds in Section~\ref{subsec:optimized-certificates}}
\label{subsec:proof-pairwise-gap}

We prove the conversion bounds stated in
Section~\ref{subsec:optimized-certificates}.

Applying Lemma~\ref{lem:cher} with $x=\eps-L$ and taking expectation under $P$ gives
\[
 \delta_{P\|Q}(\eps)
 \leq \frac{t^t}{(1+t)^{1+t}}\exp\{K(t)-t\eps\}.
\]
Since $K(t)=tD_{t+1}(P\|Q)$, an order-$1+t$ bound
$D_{t+1}(P\|Q)\leq\bar D$ therefore implies
\[
 \delta_{P\|Q}(\eps)
 \leq \frac{t^t}{(1+t)^{1+t}}\exp\{t(\bar D-\eps)\}.
\]
This is the refined fixed-order conversion of
\citet{balle2020hypothesis}; dropping the prefactor recovers the original
conversion of \citet{mironov2017renyi}.  If
$\bar D=\widetilde\rho(1+t)$, optimization over $t>0$ gives
\eqref{eq:zcdp-delta-bound}.

Let $b(t)$ denote the logarithm of the objective in
\eqref{eq:zcdp-delta-bound}. Its first two derivatives are
\[
 b'(t)=\log\frac{t}{1+t}+\widetilde\rho(2t+1)-\eps,
 \qquad
 b''(t)=\frac{1}{t(1+t)}+2\widetilde\rho>0.
\]
For $\widetilde\rho>0$, the first derivative tends to $-\infty$ as
$t\downarrow0$ and to $+\infty$ as $t\to\infty$. Hence the minimizer is
unique and satisfies \eqref{eq:chernoff-stationarity}.
Finally, when $\rho>0$, substitute
$K(t)=\rho t(t+1)-\Ga(t)$ into \eqref{eq:rdp-delta-bound} at
$t=t_\eps(\rho)$ to obtain \eqref{eq:pairwise-transfer}.

\subsection{Proof of Proposition~\ref{prop:optimized-gap}}
\label{app:proof-optimized-gap}

\begin{proof}[Proof of Proposition~\ref{prop:optimized-gap}]
Since $F_{\rm R}(t_{\rm Z})\geq\eps_{\rm R}$,
\[
 \eps_{\rm Z}
 =F_{\rm R}(t_{\rm Z})+\frac{\Ga(t_{\rm Z})}{t_{\rm Z}}
 \geq\eps_{\rm R}+\frac{\Ga(t_{\rm Z})}{t_{\rm Z}}.
\]
Similarly, evaluating $F_{\rm Z}$ at $t_{\rm R}$ gives
\[
 \eps_{\rm Z}
 \leq F_{\rm Z}(t_{\rm R})
 =\eps_{\rm R}+\frac{\Ga(t_{\rm R})}{t_{\rm R}}.
\]
The final assertion follows from $\Ga\geq0$ and
$\Ga(t_{\rm R})/t_{\rm R}\leq\sup_{t\in I}\Ga(t)/t\leq\eta$, since
$t_{\rm R}\in I$.
\end{proof}

\subsection{Proof of Lemma~\ref{lem:plrv-quartic}}
\label{app:proof-plrv-quartic}

Lemma~\ref{lem:plrv-quartic} isolates the quantities needed for a
quantitative fourth-order expansion: the excess of $\rho(s)$ over
$\EE_{P_s}L_s$ and the third and fourth cumulants of privacy loss.  The lemma applies equally to
continuous and discrete mechanisms and retains an explicit fifth-order
remainder.

\begin{lemma}[Quantitative fourth-order identity for the CGF gap]\label{lem:plrv-quartic}
Fix $T>0$.  Let $L_s$ be the privacy-loss variables of mutually absolutely continuous
pairs, with CGF $K$, so that $K(-1)=0$.  As usual, the dependence
of $K$ on $s$ is implicit.  Suppose, for all sufficiently small
$s\neq0$, that $\rho(s)<\infty$,
$K\in C^5([-1,T])$, and
there is an $M_T<\infty$, independent of $s$, such that
\begin{equation}\label{eq:plrv-fourth-remainder}
 \sup_{-1\leq u\leq T}|K^{(5)}(u)|\leq M_T|s|^5.
\end{equation}
Then, for $0<t\leq T$,
\begin{align*}
 \rho(s)t(t+1)-K(t)
 &=(\rho(s)-\kappa_1(L_s))t(t+1)
   -\frac{\kappa_3(L_s)}6t^2(t+1)\\
 &\quad+\frac{\kappa_4(L_s)}{24}t^2(1-t^2)+E_s(t).
\end{align*}
Moreover,
\begin{equation}\label{eq:plrv-fifth-error}
 |E_s(t)|\leq\frac{M_T|s|^5}{120}(t^2+t^5),
 \qquad 0<t\leq T.
\end{equation}
\end{lemma}

\begin{proof}[Proof of Lemma~\ref{lem:plrv-quartic}]
Taylor's theorem and \eqref{eq:plrv-fourth-remainder} give
\[
 K(u)=\sum_{j=1}^4\frac{\kappa_j(L_s)}{j!}u^j+\Delta_s(u),
 \qquad
 |\Delta_s(u)|\leq\frac{M_T|s|^5}{120}|u|^5
\]
for $-1\leq u\leq T$.  Applying this identity at $u=-1$ gives
\[
 \kappa_1(L_s)=\frac{\kappa_2(L_s)}{2}
 -\frac{\kappa_3(L_s)}{6}
 +\frac{\kappa_4(L_s)}{24}+\Delta_s(-1).
\]
Substituting this identity into the cumulant expansion of $K(t)$ gives
\[
 \kappa_1(L_s)t(t+1)-K(t)
 =-\frac{\kappa_3(L_s)}6t^2(t+1)
  +\frac{\kappa_4(L_s)}{24}t^2(1-t^2)
  +t^2\Delta_s(-1)-\Delta_s(t).
\]
Thus $E_s(t)=t^2\Delta_s(-1)-\Delta_s(t)$.  The two Taylor-remainder bounds give
\eqref{eq:plrv-fifth-error}; adding
$(\rho(s)-\kappa_1(L_s))t(t+1)$ proves the identity.
\end{proof}

\subsection{Proof of Theorem~\ref{thm:sensitivity-gap}}
\label{app:proof-local-expansion}

\begin{proof}[Proof of the local expansion theorem, Theorem~\ref{thm:sensitivity-gap}]
Apply Lemma~\ref{lem:plrv-quartic} with
$\kappa_1(L_s)=\EE_{P_s}L_s$. Its identity is
\eqref{eq:unified-profile-cumulant-expansion}, and
\eqref{eq:plrv-fifth-error} gives
\eqref{eq:one-pair-fifth-order-bound}.
\end{proof}

\subsection{Proof of Theorem~\ref{thm:composition-extension}}
\label{app:proof-composition}

\begin{proof}[Proof of Theorem~\ref{thm:composition-extension}]
Independence gives
\[
 tD_{t+1}(P^{(n)}\|Q^{(n)})
 =\sum_{i=1}^n K_i(t)=\bK(t).
\]
Subtracting this identity from
$\rho_{\mathrm{sum}}t(t+1)=\sum_i\rho_i t(t+1)$ proves
\eqref{eq:gap-additivity}.
For $n$ copies of one pair, $\bK=nK$, so
$\sup_{t>0}\bK(t)/\{t(t+1)\}=n\rho$ and $\bGa=n\Ga$.
\end{proof}

\FloatBarrier
\section{Proofs for Section~\ref{sec:practical-implications}}
\label{app:deferred-section4}

For the mechanism calculations, it is useful to divide out powers of the
shift. For $s\ne0$, define
\begin{equation}\label{eq:profile-defect}
 U_s:=\frac{L_s-\EE_{P_s}L_s}{s},
 \qquad
 A_s:=\frac{2\{\rho(s)-\EE_{P_s}L_s\}}{s^2}\geq0.
\end{equation}
The nonnegativity follows from
$\lim_{t\downarrow0}K(t)/\{t(t+1)\}=K'(0)=\EE_{P_s}L_s$
and $\rho(s)=\sup_{t>0}K(t)/\{t(t+1)\}$.
These definitions give $\rho(s)-\EE_{P_s}L_s=s^2A_s/2$,
$\kappa_3(L_s)=s^3\kappa_3(U_s)$, and
$\kappa_4(L_s)=s^4\kappa_4(U_s)$.

\subsection{Proof of Lemma~\ref{lem:localization}}
\label{app:proof-localization}

Throughout the
additive-noise calculations, $s$ is a signed output shift.  For even $f$,
$K$ depends on $s$ only through $r=|s|$, and we use $r$ when rescaling the
order.

For the following endpoint-localization lemma, specialize the location
convention~\eqref{eq:location-pair-convention} to a noise variable $Z$ with
density $f$.  It gives sufficient conditions for $\rho(s)=K'(0)=\EE_{P_s}L_s$;
Proposition~\ref{prop:student-general} verifies them for Student-$t$
noise.  The privacy-loss and CGF definitions give
\[
 L_s=\log f(Z)-\log f(Z+s),\qquad
 K(t)=\log\EE_{Z\sim f}[e^{tL_s}],
\]
and define the location score and its CGF
\[
 S:=-(\log f)'(Z),\qquad
 \Lambda_S(\tau):=\log\EE_{Z\sim f}[e^{\tau S}].
\]
These equalities hold wherever the expectations are finite.  Since
$\EE_{Z\sim f}[S]=0$, the fourth cumulant used
below is
\[
 \kappa_4(S)
 =\EE_{Z\sim f}[S^4]
 -3\bigl(\EE_{Z\sim f}[S^2]\bigr)^2.
\]

\begin{lemma}[Uniform control over large RDP orders]\label{lem:localization}
Assume that $f$ is even. For each nonzero shift $s$, put $r=|s|$.
Symmetry makes the CGFs at shifts $s$ and $r$ identical.
For the current shift, define
\[
 \mathcal H_s(\tau)
 =\frac{K(\tau/r)}{\tau(\tau+r)},\qquad \tau>0.
\]
Assume that $\mathcal H_s$ extends continuously to $\tau=0$ with
$\mathcal H_s(0)=K'(0)/r^2$, and that
\begin{equation}\label{eq:lan-assumptions}
 \frac{K'(0)}{s^2}
 \longrightarrow\frac12\EE_{Z\sim f}[S^2],\hspace{2em}
 K(t)=\frac12\EE_{Z\sim f}[S^2]s^2t(t+1)+o(|s|^2)
\end{equation}
uniformly on every compact $J\subset(0,\infty)$.

For endpoint localization, assume that differentiation under the defining
integral is valid through the required order and that the following
uniform bound holds for every sufficiently small fixed
$\eta>0$:
\begin{equation}\label{eq:tail-localization}
 \limsup_{s\to0}\sup_{\tau\geq\eta}\mathcal H_s(\tau)
 \leq\sup_{\tau\geq\eta}\frac{\Lambda_S(\tau)}{\tau^2}.
\end{equation}
Near $(r,\tau)=(0,0)$ suppose the Taylor remainder is uniform and satisfies
\begin{equation}
 \mathcal H_s(\tau)-\mathcal H_s(0)
 =\frac{\kappa_4(S)}{24}\tau(\tau+r)
   +O\!\left(\tau(\tau+r)^2\right).
\label{eq:H4}
\end{equation}
Assume
$\kappa_4(S)<0$ and, for some $a_4>0$,
\begin{equation}\label{eq:quant-subg-4}
 \frac12\EE_{Z\sim f}[S^2]-\frac{\Lambda_S(\tau)}{\tau^2}
 \geq a_4\min(\tau^2,1),\qquad \tau>0.
\end{equation}
Then
\[
 \rho(s)=K'(0)=\EE_{P_s}L_s
\]
for all sufficiently small nonzero $s$.  Hence $A_s=0$ exactly.
\end{lemma}

\begin{proof}[Proof of Lemma~\ref{lem:localization}]
We begin with the factorization used for endpoint localization.  The
likelihood identity gives $K(0)=K(-1)=0$.  Therefore, the definition of
$\mathcal H_s$ gives the exact factorization
\[
 K(t)=r^2t(t+1)\mathcal H_s(rt).
\]
Moreover, $K'(0)=\EE_{P_s}L_s$.  Dividing the numerator and denominator
in the definition of $\mathcal H_s(\tau)$ by $\tau$ and letting $\tau\to0$
gives $\mathcal H_s(0)=K'(0)/r^2$, as stated.

Choose $\eta>0$ small enough that, for all
$|s|$ sufficiently small, the absolute value of the remainder in
\eqref{eq:H4} is bounded by
\[
 -\frac{\kappa_4(S)}{48}\tau(\tau+r)
\]
whenever $0<\tau\leq\eta$.  It follows that
\[
 \mathcal H_s(\tau)-\mathcal H_s(0)
 \leq \frac{\kappa_4(S)}{48}\tau(\tau+r)<0,
 \qquad 0<\tau\leq\eta.
\]
For $\tau\geq\eta$, \eqref{eq:tail-localization} and
\eqref{eq:quant-subg-4} put the supremum a fixed distance below
$\tfrac12\EE_{Z\sim f}[S^2]$.  Because
$\mathcal H_s(0)=K'(0)/r^2$ tends to
$\tfrac12\EE_{Z\sim f}[S^2]$, the region $\tau\geq\eta$ also lies
below $\mathcal H_s(0)$ for all sufficiently small nonzero $s$.  Hence
$r^2\sup_{\tau>0}\mathcal H_s(\tau)=K'(0)$, which gives
$\rho(s)=K'(0)=\EE_{P_s}L_s$.

This completes the proof.
\end{proof}

The hypotheses of Lemma~\ref{lem:localization} can be checked from $K$ and
$\Lambda_S$ without
first evaluating $\rho$.  For a $C^5$ log-density, a common integrable
bound on the required derivatives justifies the Taylor expansions in
\eqref{eq:H4}.  The uniform bound in
\eqref{eq:tail-localization} must be
verified separately for each mechanism, as for the Student-$t_\nu$ family in
Proposition~\ref{prop:student-general}.  Nondifferentiable log-densities
require piecewise calculations.  The sign and uniform-control conditions exclude a
zero fourth cumulant and rule out a competing second-order maximum away
from the origin.

\subsection{Derivation sketch for~(\ref{eq:smooth-third-cumulant}) and~(\ref{eq:smooth-fourth-cumulant})}
\label{app:proof-smooth-cumulants}

For a smooth location family, continue to use the location-score notation
$S$ introduced above.  Suppose $\log f$ is
five times continuously differentiable, differentiation under the relevant
expectations is dominated through fifth order, the required score moments are
finite, and the integration-by-parts boundary terms vanish.  Taylor expansion
of $L_s=\log f(Z)-\log f(Z+s)$ and integration by parts then give
\begin{align}
 \kappa_3(L_s)
 &=\kappa_3(S)s^3+\frac12\kappa_4(S)s^4+O(|s|^5),
 \label{eq:smooth-third-cumulant}\\
 \kappa_4(L_s)
 &=\kappa_4(S)s^4+O(|s|^5).
 \label{eq:smooth-fourth-cumulant}
\end{align}
Indeed, the integration-by-parts identities
\begin{align*}
 \EE[(\log f)''(Z)]
 &=-\EE[((\log f)'(Z))^2],\\
 \EE[((\log f)'(Z))^2(\log f)''(Z)]
 &=-\frac13\EE[((\log f)'(Z))^4]
\end{align*}
give the
fourth-order coefficient in \eqref{eq:smooth-third-cumulant}.  These
identities also show why a kink in the log-density can change the
fourth-order coefficient.

\subsection{Proof of Proposition~\ref{prop:common-mechanisms}}
\label{app:proof-common-mechanisms}
\paragraph{Exact CGFs.}
In the setting of Proposition~\ref{prop:common-mechanisms}, the following hold
for every $s\in\R$ and $t>0$. For Gaussian noise,
\begin{equation}\label{eq:gaussian-exact-cgf}
 K(t)=\frac{t(t+1)}{2\sigma^2}s^2.
\end{equation}
For the Gaussian mixture, write $Z=m\xi+G$ with $\xi$ uniform on $\{-1,1\}$
and $G\sim N(0,1)$ independent; then
\begin{equation}\label{eq:mixture-exact-cgf}
 K(t)=\log\EE\!\left[
  \exp\{t(sZ+s^2/2)\}
  \left\{\frac{\cosh(mZ)}{\cosh(m(Z+s))}\right\}^{t}
 \right].
\end{equation}
For unit-scale Laplace noise,
\begin{equation}\label{eq:laplace-exact-cgf}
 K(t)=\log\!\left(
 \frac{(1+t)e^{t|s|}+te^{-(1+t)|s|}}{1+2t}\right).
\end{equation}

\begin{proof}[Proof of Proposition~\ref{prop:common-mechanisms}]
For part \textup{(a)}, completing the square in the Gaussian likelihood
ratio gives~\eqref{eq:gaussian-exact-cgf} and the stated optimal parameter.  Under
$P_s$,
\[
 L_s=\frac{sZ}{\sigma^2}+\frac{s^2}{2\sigma^2},
 \qquad Z\sim N(0,\sigma^2),
\]
so $A_s=0$, $U_s=Z/\sigma^2$, and
$\kappa_3(U_s)=\kappa_4(U_s)=0$.  The CGF is quadratic, so the remainder in
the local expansion theorem also vanishes.  The theorem therefore gives
$\Ga(t)=0$ exactly.

For part \textup{(b)}, the mixture density is proportional to
$e^{-x^2/2}\cosh(mx)$, so
\[
 (\log f_m)'(x)=-x+m\tanh(mx),\qquad
 (\log f_m)''(x)=-1+m^2\operatorname{sech}^2(mx).
\]
Evaluating $\log\{f_m(Z)/f_m(Z+s)\}$ gives~\eqref{eq:mixture-exact-cgf}.
Integration by parts gives
\[
 \EE[((\log f_m)'(Z))^2]
 =1-m^2\EE[\operatorname{sech}^2(mZ)].
\]
Generating $Z+s$
by first shifting the unit-variance Gaussian component $G$ by $s$ and then
adding the independent term $m\xi$ shows by data processing that $\rho(s)\leq s^2/2$.
The density is even, so both the CGF and $\rho(s)$ are even in $s$.
For $z\geq0$,
\[
 \log\frac{f_m(z)}{f_m(z+|s|)}
 \geq |s|z+\frac{s^2}{2}-m|s|,
 \qquad f_m(z)\geq\frac12\phi(z-m).
\]
Consequently,
\[
 e^{K(t)}
 \geq\frac12e^{s^2t(t+1)/2}
 \Pbb_{Y\sim N(m+t|s|,1)}[Y\geq0].
\]
After division by $t(t+1)$ and passage to $t\to\infty$, this proves the
matching lower bound $\rho(s)\geq s^2/2$.

Put
\[
 S_m:=-(\log f_m)'(Z)=Z-m\tanh(mZ),
 \qquad I_m:=\EE[S_m^2].
\]
The integration-by-parts identity above gives
$I_m=1-m^2\EE[\operatorname{sech}^2(mZ)]$.  Since $f_m$ is even,
$\EE_{P_s}L_s$ is even in $s$, and its second-order expansion is
\[
 \EE_{P_s}L_s=\frac{I_m}{2}s^2+O(s^4).
\]
Together with $\rho(s)=s^2/2$, this yields
\[
 A_s=m^2\EE[\operatorname{sech}^2(mZ)]+O(s^2).
\]
Moreover, $S_m$ is symmetric.  The smooth-location cumulant expansions
\eqref{eq:smooth-third-cumulant}--\eqref{eq:smooth-fourth-cumulant} give
\[
 \kappa_3(U_s)=\tfrac12\kappa_4(S_m)s+O(s^2)=O(s),
 \qquad
 \kappa_4(U_s)=\kappa_4(S_m)+O(|s|)=O(1).
\]
Substitution in the local expansion theorem,
Theorem~\ref{thm:sensitivity-gap}, proves
\eqref{eq:gaussian-mixture-leading}.

For part \textup{(c)}, symmetry makes the CGF depend on the shift only through
$|s|$.  Splitting at the two kinks gives~\eqref{eq:laplace-exact-cgf}.
The exact value $\rho(s)=|s|+e^{-|s|}-1$ is the continuous-Laplace
formula established by \citet{harrison2025exact}; differentiating the exact
CGF at zero also gives $\EE_{P_s}L_s=\rho(s)$ and hence $A_s=0$.
For $s>0$, differentiation of the same transform gives
\[
 \kappa_3(U_s)=-s+O(s^2),
 \qquad
 \kappa_4(U_s)=-2+O(s).
\]
Symmetry gives $U_{-s}\overset d=-U_s$, so the signed-$s$ form of the first
display is $\kappa_3(U_s)=-s+O(s^2)$, while the fourth-cumulant remainder is
$O(|s|)$.  The local expansion theorem therefore gives
\begin{align*}
 \Ga(t)
 &=\frac{s^4}{6}t^2(t+1)
   -\frac{s^4}{12}t^2(1-t^2)+O(|s|^5)\\
 &=\frac{[t(t+1)]^2}{12}s^4+O(|s|^5),
\end{align*}
which proves item \emph{(c)}.

It remains to verify the common remainder hypothesis used above.  It is
immediate for Gaussian noise.  For the Gaussian mixture and $|s|\leq s_0$,
\[
 |L_s|\leq |s|\bigl(|Z|+m+s_0/2\bigr),
\]
and the mixture has uniformly bounded exponentially tilted moments of every
fixed order.  For Laplace noise,
\[
 |L_s|\leq|s|.
\]
Thus the fifth absolute centered moments under every tilt with
$-1\leq u\leq T$ are $O(|s|^5)$. Since $K^{(5)}(u)$ is the fifth
cumulant under the tilt, it is also $O(|s|^5)$ uniformly on this interval.
The apparent singularities in the displayed Laplace formulas
are removable because these bounded privacy-loss variables have positive
smooth moment-generating functions on the real line.  Hence all three
families satisfy the hypotheses of the local expansion theorem.

\end{proof}

\subsection{Proof of Proposition~\ref{prop:subsampled-gaussian}}
\label{app:proof-subsampled-gaussian}

\begin{proof}[Proof of Proposition~\ref{prop:subsampled-gaussian}]
Fix $z>0$ and $\|\mathbf v\|=C$, let $Q=N(0,z^2C^2I)$, and, for $0\leq q\leq1$, set
$P_q=(1-q)Q+qN(\mathbf v,z^2C^2I)$.  Let $W$ be the likelihood ratio
of $N(\mathbf v,z^2C^2I)$ with respect to $Q$.  Under $Q$,
\[
 \log W=\frac{G}{z}-\frac{1}{2z^2},
 \qquad G\sim N(0,1),
\]
and hence, for $\alpha\geq1$,
\begin{equation}\label{eq:gaussian-lr-moment}
 \EE_Q[W^\alpha]
 =\exp\left\{\frac{\alpha(\alpha-1)}{2z^2}\right\}.
\end{equation}
Put $W_q=1-q+qW=dP_q/dQ$.  Convexity of
$x\mapsto x^\alpha$ gives
\[
 \EE_Q[W_q^\alpha]
 \leq1-q+q\EE_Q[W^\alpha]
 \leq \exp\left\{\frac{\alpha(\alpha-1)}{2z^2}\right\}.
\]
Thus the one-step optimal zCDP parameter is at most $1/(2z^2)$.  Since
$W_q\geq qW$,
\[
 K(t)\geq(t+1)\log q+\frac{t(t+1)}{2z^2}.
\]
Dividing by $t(t+1)$ and letting $t\to\infty$ proves
$\rho(P_q,Q)=1/(2z^2)$ for every $q>0$.  With $z=z_n$, $1/(2z_n^2)=\rho_{\mathrm{sum}}/n$,
which proves the one-step identities in \eqref{eq:subsampled-one-step-gap}.

At $t=1$,
\[
 e^{K(1)}=\EE_Q[W_q^2]
 =1+q^2\EE_Q[(W-1)^2]
 =1+\left(e^{1/z^2}-1\right)q^2,
\]
which, with $e^{1/z_n^2}=e^{2\rho_{\mathrm{sum}}/n}$, proves the formula for $\Ga(1)$.

It remains to establish monotonicity.  For fixed $z$ and $t>0$, write
\[
 M_t(q):=\EE_Q\!\left[(1+q(W-1))^{t+1}\right].
\]
Then $M_t'(0)=(t+1)\EE(W-1)=0$, whereas, for $0<q<1$,
\[
 M_t''(q)
 =t(t+1)\EE_Q\!\left[
 (W-1)^2\{1+q(W-1)\}^{t-1}\right]>0.
\]
Hence $M_t(q)$, and therefore $K(t)$ viewed as a
function of $q$, is strictly increasing.  The one-step gap is consequently
strictly decreasing.  At $q=1$,
the pair is an ordinary Gaussian shift and its gap is zero.  Continuity of
the expectation at $q=0$ gives the displayed one-step right-hand limit.

For part~\textup{(b)}, set $z=z_n$, so that
$1/(2z_n^2)=\rho_{\mathrm{sum}}/n$.  Exact product additivity gives
$\bK(t)=nK(t)$, and its normalized supremum is
$n/(2z_n^2)=\rho_{\mathrm{sum}}$.  Thus the standard-composition parameter is optimal
for the product pair, and item \emph{(b)} follows.
The monotonicity and endpoint statements follow immediately from
part~\textup{(a)}.
\end{proof}

\subsection{Proof of Proposition~\ref{prop:zcdp-budget-allocations}}
\label{app:proof-zcdp-budget-allocations}

\begin{proof}[Proof of Proposition~\ref{prop:zcdp-budget-allocations}]
Independence makes the product likelihood ratio the product of the component
likelihood ratios, so their log moment-generating functions add.  Hence
$\bGa(t)=\sum_i\Phi_t(\rho_i)$ exactly.  If $\Phi_t$ is convex, Jensen's inequality and
$\sum_i\rho_i=\rho_{\mathrm{sum}}$ give
\[
 \bGa(t)
 \geq
 n\Phi_t\!\left(\frac{\rho_{\mathrm{sum}}}{n}\right).
\]
Because $\psi(0)=0$ and the single-pair gap is zero there, the chord inequality between
$0$ and $\rho_{\mathrm{sum}}$ gives
\[
 \Phi_t(\rho_i)
 \leq\frac{\rho_i}{\rho_{\mathrm{sum}}}\Phi_t(\rho_{\mathrm{sum}}).
\]
Summing this inequality proves the upper bound in
\eqref{eq:convex-zcdp-allocation-bounds}.  Concavity reverses Jensen's
inequality and the chord inequality, and therefore reverses both bounds.
The equality statements follow from the strict forms of the same two
inequalities.

For the unit-scale Laplace illustration, parameterize the selected inverse
branch by the magnitude $r\geq0$.  Here
$x=\rho(r)=r+e^{-r}-1$, equivalently $r=\psi(x)$.  Direct differentiation of the exact
CGF~\eqref{eq:laplace-exact-cgf} gives, for $r>0$,
\[
 \left.
 \frac{d^2}{dx^2}\Phi_1(x)
 \right|_{x=\rho(r)}
 =\frac{2e^{2r}\{2e^{3r}+6e^{2r}+3e^r+1\}}
 {(2e^{3r}+1)^2}>0.
\]
The right limit at $r=0$ is $8/3$, so
$\Phi_1$ is strictly convex throughout the interval
used in Figure~\ref{fig:zcdp-budget-allocation}.  For the Gaussian-mixture
panel, $\psi(x)=\sqrt{2x}$, and strict concavity of $\Phi_1$ on the
displayed interval is supported numerically, but not proved, by negative
centered finite differences on the plotted grid.
\end{proof}

\subsection{Proof of Proposition~\ref{prop:student-general}}
\label{app:proof-student-general}

\begin{proof}[Proof of Proposition~\ref{prop:student-general}]
Fix $\nu$ and retain the notation $Z,P_s,Q$ and $\rho_\nu(s)$ from the
proposition.  Write
$\ell_\nu=\log f_\nu$.  The location score is
\[
 S_\nu:=-\ell_\nu'(Z)=\frac{(\nu+1)Z}{\nu+Z^2},
 \qquad
 |S_\nu|\leq B_\nu:=\frac{\nu+1}{2\sqrt\nu}.
\]
The density is even, so the CGF and $\rho_\nu(s)$ are even in $s$.  Put
$r=|s|$.  The mean-value theorem gives
$|L_s|\leq B_\nu r$.  Consequently,
\[
 0\leq\frac{K(t)}{t(t+1)}\leq\frac{B_\nu r}{t+1},
\]
which tends to zero as $t\to\infty$.  The normalized CGF is continuous on
compact $t$-intervals and has the finite endpoint
$K'(0)=\EE_{P_s}L_s$, so $\rho_\nu(s)<\infty$ for every $s\in\R$.

To localize the envelope for all sufficiently small nonzero $s$, put
$Y_\nu=Z^2/(\nu+Z^2)$.  The beta representation of a Student random variable
gives $Y_\nu\sim\operatorname{Beta}(1/2,\nu/2)$ and hence, for $j\geq1$,
\begin{equation}\label{eq:student-score-moments}
 \mu_{2j,\nu}:=\EE[S_\nu^{2j}]
 =\frac{(\nu+1)^{2j}}{\nu^j}
  \frac{(1/2)_j(\nu/2)_j}{((\nu+1)/2)_{2j}}.
\end{equation}
In particular,
\begin{equation}\label{eq:student-score-cumulants}
 I_\nu:=\mu_{2,\nu}=\frac{\nu+1}{\nu+3},
 \qquad
 \kappa_4(S_\nu)
 =-\frac{18(\nu+1)^2(\nu^2+4\nu-1)}
 {\nu(\nu+3)^2(\nu+5)(\nu+7)}<0.
\end{equation}
Comparison with the even moments of $N(0,I_\nu)$ gives
\[
 \frac{\mu_{2j+2,\nu}/\mu_{2j,\nu}}{(2j+1)I_\nu}
 =\frac{(\nu+1)(\nu+3)(2j+\nu)}
 {\nu(4j+\nu+1)(4j+\nu+3)}<1,
 \qquad j\geq1,
\]
where the difference between denominator and numerator is
$2j(8\nu j+3\nu^2+4\nu-3)>0$ for $\nu\geq1$.
Because $S_\nu$ is symmetric and bounded, termwise comparison of the entire
moment series yields
\[
 \EE[e^{\tau S_\nu}]<e^{I_\nu\tau^2/2},
 \qquad \tau\ne0.
\]
Equation~\eqref{eq:student-score-cumulants} supplies a quadratic separation
near $\tau=0$, continuity supplies uniform separation on compact intervals
away from zero, and
$\log\EE[e^{\tau S_\nu}]\leq B_\nu\tau$ controls large $\tau$.
Thus the quantitative sub-Gaussian condition
\eqref{eq:quant-subg-4} holds.

All derivatives of positive order of $\ell_\nu$ are bounded.  Uniformly in
$Z$,
\[
 L_r=-\sum_{j=1}^5\frac{r^j}{j!}\ell_\nu^{(j)}(Z)+O(r^6),
\]
so $L_r/r\to S_\nu$ uniformly.  Expanding the log moment-generating function
and using the usual integration-by-parts identities gives
\eqref{eq:lan-assumptions} and \eqref{eq:H4}, uniformly on the stated compact
sets.  Symmetry makes the CGFs at shifts $s$ and $r=|s|$ equal, and
$K(\tau/r)\to\log\EE[e^{\tau S_\nu}]$ locally uniformly in $\tau$.  The bound
\[
 0\leq\frac{K(\tau/r)}{\tau(\tau+r)}
 \leq\frac{B_\nu}{\tau+r}
\]
controls $\tau\geq R$ uniformly once $R$ is large, while local uniform
convergence controls each interval $[\eta,R]$.  This proves
\eqref{eq:tail-localization}.  Lemma~\ref{lem:localization}
therefore implies
\[
 \rho_\nu(s)=\EE_{P_s}L_s
\]
for every sufficiently small nonzero $s$.

Thus $A_s=0$.  Symmetry gives
$\kappa_3(S_\nu)=0$, and
\eqref{eq:smooth-third-cumulant}--\eqref{eq:smooth-fourth-cumulant}, followed
by the local expansion theorem, Theorem~\ref{thm:sensitivity-gap}, gives
\[
 \Ga(t)
 =-\frac{\kappa_4(S_\nu)}{24}[t(t+1)]^2s^4+O(|s|^5).
\]
The coefficient obtained from \eqref{eq:student-score-cumulants} is exactly
$c_{4,\nu}$ in \eqref{eq:student-general-gap}.  Bounded derivatives justify
six uniform differentiations in $s$.  Since both $K(t)$ and
$\EE_{P_s}L_s=\rho_\nu(s)$ are even functions of $s$, the gap is even and its
fifth-order Taylor coefficient vanishes.  The remainder is therefore
$O(|s|^6)$.

Finally,
\[
 c_{4,\nu}'
 =-\frac{3(\nu+1)p(\nu)}
 {4\nu^2(\nu+3)^3(\nu+5)^2(\nu+7)^2},
\]
where
\[
 p(\nu)=\nu^6+8\nu^5-19\nu^4-288\nu^3
 -613\nu^2-72\nu-105.
\]
The coefficients of $p$ have one sign change, so Descartes' rule shows that
$p$ has a unique positive root.  Since $p(1)<0$ and $p(\nu)\to\infty$,
this root, numerically $6.013052641\ldots$, is the unique maximizer on
$[1,\infty)$.  Direct expansion gives
$c_{4,\nu}=3/(4\nu)-9/\nu^2+O(\nu^{-3})$.
\end{proof}

\subsection{Proof of Proposition~\ref{prop:gg-single}}
\label{app:proof-gg-single}

\begin{proof}[Proof of Proposition~\ref{prop:gg-single}]
The density $f_\beta$ is even, so the CGF and
$\rho_\beta(s)$ are even functions of the shift, and
$U_{-s}\overset d=-U_s$.  Put $r=|s|$ when only the shift
magnitude is relevant.
Under $P_s$, write
\[
 L_s=\frac{|Z+s|^\beta-|Z|^\beta}{\beta}.
\]
For $1<\beta<2$, the mean-value inequality and Young's inequality give,
for $t\geq1$,
\[
 tL_s
 \leq tr(|Z|+r)^{\beta-1}
 \leq\frac{|Z|^\beta}{2\beta}+C_{\beta,r}(1+t^\beta).
\]
Integration against $f_\beta$ therefore yields
$K(t)\leq C'_{\beta,r}(1+t^\beta)$, so
$K(t)/\{t(t+1)\}\to0$ as $t\to\infty$.  Continuity on compact
order intervals and the finite endpoint $K'(0)=\EE_{P_s}L_s$ prove
$\rho_\beta(s)<\infty$.  For $\beta=1$, $|L_s|\leq r$ gives
the same conclusion.  For $\beta=2$, completing the square gives
$K(t)=s^2t(t+1)/2$.

The absolute moments of $Z$ are
\[
 \EE|Z|^p
 =\beta^{p/\beta}
   \frac{\Gamma((p+1)/\beta)}{\Gamma(1/\beta)},
 \qquad p>-1.
\]
For this proof, define
\[
 S_\beta(z):=\operatorname{sgn}(z)|z|^{\beta-1},\qquad
 I_\beta:=\EE[S_\beta(Z)^2],\qquad
 J_\beta:=\EE[S_\beta(Z)^4].
\]
Below, an unqualified $S_\beta$ or $S_\beta'$ inside an expectation is
evaluated at $Z$.
The moment formula gives
\[
 I_\beta=\beta^{2-2/\beta}
 \frac{\Gamma(2-1/\beta)}{\Gamma(1/\beta)},\qquad
 J_\beta=\beta^{4-4/\beta}
 \frac{\Gamma(4-3/\beta)}{\Gamma(1/\beta)}.
\]
In particular, $J_\beta-3I_\beta^2=-24\chi_\beta$.  To prove the sign of
$\chi_\beta$, put
$x=1/\beta\in[1/2,1]$ and
\[
 \mathcal C(x):=\frac{J_\beta}{I_\beta^2}
 =\frac{\Gamma(4-3x)\Gamma(x)}{\Gamma(2-x)^2}.
\]
The integral representation of the digamma function gives
\[
 \frac{d}{dx}\log \mathcal C(x)
 =\int_0^\infty
 \frac{3e^{-(4-3x)u}-e^{-xu}-2e^{-(2-x)u}}
      {1-e^{-u}}\,du.
\]
For $x<1$, the numerator equals
\[
 e^{-xu}(y^2-1)(3y^2+1)<0,
 \qquad y=e^{-(1-x)u}\in(0,1).
\]
Thus $\mathcal C$ decreases from $\mathcal C(1/2)=3$ to
$\mathcal C(1)=1$, proving
$J_\beta\leq3I_\beta^2$ and hence $\chi_\beta\geq0$.

Define the difference quotient
\[
 V_s(z)
 :=\frac{|z+s|^\beta-|z|^\beta}{\beta s}.
\]
Then $L_s=sV_s(Z)$ and, by
\eqref{eq:profile-defect},
\[
 U_s=V_s(Z)-\EE[V_s(Z)],
\]
and
$V_s(Z)\to S_\beta(Z)$ in every fixed $L^p$.
Consequently,
\[
 \kappa_4(U_s)=J_\beta-3I_\beta^2+o(1).
\]
For $1<\beta<2$, write
$S_\beta'(z)=(\beta-1)|z|^{\beta-2}$ for $z\ne0$.  For $k=1,2,3$, split the
integral defining $\EE[V_s^k]$ into $|z|\leq2|s|$ and
$|z|>2|s|$.  On the first region, both $V_s$ and
$S_\beta$ are $O(|s|^{\beta-1})$, so the absolute value of its contribution to
$s^{-1}\{\EE[V_s^k]-\EE[S_\beta^k]\}$ is
$O(|s|^{k(\beta-1)})=o(1)$.  On the second region, Taylor's formula gives
\[
 V_s(z)=S_\beta(z)+\frac{s}{2}S_\beta'(z)
 +o\!\left(|s|\,|z|^{\beta-2}\right).
\]
The required dominating powers are integrable because, for $k\leq3$,
$S_\beta^{k-1}S_\beta'$ behaves at zero as
$|z|^{(k-1)(\beta-1)+\beta-2}$, whose exponent is greater than $-1$.
Hence
\[
 \left.\frac{d}{ds}\EE[V_s^k]\right|_{s=0}
 =\frac{k}{2}\EE[S_\beta^{k-1}S_\beta'],\qquad k=1,2,3.
\]
Since $\EE[S_\beta]=\EE[S_\beta^3]=0$, differentiating
$\kappa_3(X)=\EE X^3-3(\EE X)(\EE X^2)+2(\EE X)^3$ gives
\[
 \lim_{s\to0}\frac{\kappa_3(U_s)}{s}
 =\frac32\{\EE[S_\beta^2S_\beta']
               -I_\beta\EE[S_\beta']\}.
\]
Integration by parts yields
$\EE[S_\beta']=I_\beta$ and
$\EE[S_\beta^2S_\beta']=J_\beta/3$.  Hence
\begin{equation}\label{eq:gg-cumulant-limits}
 \kappa_3(U_s)
 =-12\chi_\beta s+o(|s|),
 \qquad
 \kappa_4(U_s)=-24\chi_\beta+o(1).
\end{equation}
For $\beta=1$, differentiation of the exact Laplace transform gives
$\kappa_3(U_s)=-s+O(s^2)$ and
$\kappa_4(U_s)=-2+O(|s|)$, the same identity.  At $\beta=2$ both
cumulants vanish.

It remains to verify the fifth-derivative hypothesis. For the pair
$(P_s,Q)$, let $P_{s,u}$ have density
$dP_{s,u}/dP_s=\exp\{uL_s-K(u)\}$. Then $K^{(5)}(u)$ is the fifth
cumulant of $L_s$ under $P_{s,u}$. For $0<|s|\leq1$,
\[
 |L_s|\leq |s|G_\beta(Z),\qquad
 G_\beta(z):=(|z|+1)^{\beta-1}.
\]
Let $C_T=\max(1,T)$.  Since $\beta-1<\beta$,
\[
 \int (1+G_\beta(z)^5)e^{C_T G_\beta(z)}f_\beta(z)\,dz<\infty.
\]
Moreover $\sup_{-1\leq u\leq T}|K(u)|\to0$, so the normalizing
constants of these tilted laws are bounded away from zero for all
sufficiently small nonzero $s$.
Consequently,
\[
 \sup_{-1\leq u\leq T}
 \EE_{P_{s,u}}|L_s-\EE_{P_{s,u}}L_s|^5
 \leq C_{\beta,T}|s|^5.
\]
The fifth cumulant equals the fifth centered moment minus ten times the
product of the third and second centered moments; Lyapunov's inequality
bounds that product by the fifth absolute centered moment.  Therefore
$\sup_{-1\leq u\leq T}|K^{(5)}(u)|\leq C'_{\beta,T}|s|^5$.
For $\beta=1$, the same conclusion follows from $|L_s|\leq |s|$; for
$\beta=2$ it is immediate from the quadratic CGF.

Insert \eqref{eq:gg-cumulant-limits} in the local expansion theorem,
Theorem~\ref{thm:sensitivity-gap}.  The
cubic and quartic terms combine to
$\chi_\beta[t(t+1)]^2s^4+o(|s|^4)$.  Since
\[
 \frac{t(t+1)}2
 \frac{2\{\rho_\beta(s)-\EE_{P_s}L_s\}}{s^2}s^2
 =t(t+1)\{\rho_\beta(s)-\EE_{P_s}L_s\},
\]
this proves \eqref{eq:gg-single-gap}.

The endpoint $t\downarrow0$ in the definition of
$\rho_\beta(s)$ gives
$\rho_\beta(s)\geq\EE_{P_s}L_s$.  Replacing
$\rho_\beta(s)$ by this smaller endpoint value shows that the first
term on the right-hand side of \eqref{eq:gg-single-gap} is nonnegative.
This proves \eqref{eq:gg-quartic-lower-bound}.

At $\beta=1$, the exact Laplace CGF~\eqref{eq:laplace-exact-cgf} gives
$K'(0)=|s|+e^{-|s|}-1=\rho_1(s)$, and
$I_1=J_1=1$ gives $\chi_1=1/12$.  At $\beta=2$, the quadratic Gaussian
CGF gives $\rho_2(s)=\EE_{P_s}L_s=s^2/2$, $\chi_2=0$, and
$\Ga(t)=0$ for every $t>0$.
\end{proof}

\subsection{Proof of Proposition~\ref{prop:discrete-examples}}
\label{app:proof-discrete-examples}

\begin{proof}[Proof of Proposition~\ref{prop:discrete-examples}]
For part \textup{(a)}, the discrete Laplace privacy loss is $ar$ for
$z\geq0$, $-ar$ for $z\leq-r$, and $a(r-2j)$ at $z=-j$ for
$1\leq j<r$.  Therefore
\[
 e^{K(t)}=\frac{1-\lambda}{1+\lambda}\left\{
 e^{tar}\sum_{z=0}^\infty \lambda^z
 +e^{-tar}\sum_{j=r}^\infty \lambda^j
 +\sum_{j=1}^{r-1}\lambda^j e^{ta(r-2j)}\right\}.
\]
Summing the two geometric series and taking the logarithm gives $K(t)$ in
\eqref{eq:discrete-laplace-transform}.  The exact discrete Laplace formula
established by \citet{harrison2025exact} gives
$\rho(r)=\EE_{P_r}L_r$ as displayed.  Substitution into the gap
definition~\eqref{eq:gap-definition} proves the general-$t$ formula.

For part \textup{(b)}, completing the square in the discrete Gaussian sum
gives
\begin{align*}
 e^{K(t)}
 &=e^{r^2t(t+1)/(2\sigma^2)}
   \frac{\Theta_\sigma(tr)}{\Theta_\sigma(0)}\\
 &=e^{r^2t(t+1)/(2\sigma^2)}
   \frac{\Theta_\sigma((1+t)r)}{\Theta_\sigma(0)},
\end{align*}
where the second equality uses integer periodicity.  Thus
\[
 K(t)=\frac{r^2}{2\sigma^2}t(t+1)
 +\log\frac{\Theta_\sigma((1+t)r)}{\Theta_\sigma(0)}.
\]
The theta sum is maximized at integers, with equality only at integers.  The
logarithmic term is therefore nonpositive, while $t=1/r$ makes
$(1+t)r=r+1$ an integer and attains the quadratic envelope.  This proves
$\rho(r)=r^2/(2\sigma^2)$; differentiating at $t=0$
also gives
$\EE_{P_r}L_r
=r^2/(2\sigma^2)$.  Periodicity of the theta sum gives the stated
period and zero set.  Poisson summation gives, uniformly in $u$,
\[
 \Theta_\sigma(u)
 =\sqrt{2\pi}\sigma\{1+2e^{-2\pi^2\sigma^2}
 \cos(2\pi u)+O(e^{-8\pi^2\sigma^2})\}.
\]
Expanding the logarithm proves
\eqref{eq:discrete-gaussian-gap}.
\end{proof}

\FloatBarrier
\section{Supplementary Figures and Numerical Details}
\label{app:section4-additional-results}
\label{app:section5-additional-results}

The supplementary figures check whether the comparisons in the main text
persist under different settings: the analytical plots use $t=1.5$ in
place of $t=1$, and the experiments use $(\eps,\delta)=(5,10^{-5})$ in
place of $(3,10^{-5})$. Each figure is linked to the corresponding
main-text figure so that the effect of the changed setting can be compared
directly. The numerical checks assess quadrature accuracy, order
truncation, and noise calibration; the experimental details and tables
record the workloads, training settings, and complete results behind
Section~\ref{sec:real-data-experiments}. These checks support the reported
patterns within the settings examined.

\subsection{Computational Setup}\label{app:numerics}

The analytical figures in Sections~\ref{sec:defect}
and~\ref{sec:practical-implications} are deterministic.
Gaussian, Laplace, and lattice curves use the closed forms
\eqref{eq:gaussian-exact-cgf}, \eqref{eq:laplace-exact-cgf}, and
Proposition~\ref{prop:discrete-examples}.
Gaussian-mixture expectations use 192-node Gauss--Hermite quadrature;
Student-$t_\nu$ expectations use 2048-node Gauss--Legendre quadrature after
$x=\sqrt\nu\tan\theta$; and generalized-Gaussian transforms use adaptive
quadrature split at $-s$ and $0$.  For the last family, we maximize
$K(u)/[u(u+1)]$ over $0<u\leq32$, including the $u\downarrow0$ endpoint.
The large-order truncation is checked separately from the asymptotic
statement in Proposition~\ref{prop:gg-single}.  Specifically, for
$1<\beta<2$, $r=|s|$, and $0<\zeta<1$, the bound
\[
 \frac{K(u)}{u(u+1)}
 \leq B_{\beta,r,\zeta}(u)
 :=\frac{ur^\beta+
 (ur)^\beta(\beta-1)^{\beta-1}/(\beta\zeta^{\beta-1})
 -\beta^{-1}\log(1-\zeta)}{u(u+1)}
\]
follows from
$L_s\leq r|Z|^{\beta-1}+r^\beta$, Young's inequality, and
$\EE\exp\{\zeta|Z|^\beta/\beta\}=(1-\zeta)^{-1/\beta}$.
For fixed $\zeta$, every term of $B_{\beta,r,\zeta}(u)$ is nonincreasing
when $u\geq(\beta-1)/(2-\beta)$.  At the displayed $r=0.2$ and
$\beta\in\{1.25,1.5,1.75,1.9\}$, this threshold is at most $9$.
Choosing $\zeta$ to minimize $B_{\beta,r,\zeta}(32)$ gives an upper
bound below the numerically evaluated endpoint $K'(0)\leq\rho(s)$ in
every case; the smallest margin is $2.7296\times10^{-4}$.
These parameter-specific tail checks justify the cutoff used in the plots.
The endpoints $\beta=1,2$ use the exact Laplace and Gaussian formulas.
Exponential moments use log-sum-exp arithmetic where needed.

Figure~\ref{fig:gap-metric} uses the captioned four-point pair and, for its
Gaussian panel, $\rho=1/2$ on 100 log-spaced values
$10^{-2}\leq t\leq10^{1.2}$.  Figure~\ref{fig:common-mechanism-gaps}(a) fixes
$t=1$ and uses 61 log-spaced shifts $10^{-3}\leq s\leq10^{-1}$, with
$m=1$ and unit-scale Laplace noise; slopes are fitted over $s\leq0.01$.
Figure~\ref{fig:common-mechanism-gaps}(b)--(c) uses 1001 values of $q\in[0,1]$ with
$(n,\rho_{\rm sum},\|\mathbf v\|/C,t)=(20,5,1,1)$ and plots the positive-$q$
right limit at $q=0$.

Figure~\ref{fig:zcdp-budget-allocation} uses 241 values
of $\omega\in[0,1]$ with $(n,\rho_{\rm sum},t)=(8,1/2,1)$; repeating the
Gaussian-mixture calculation with 384 rather than 192 Hermite nodes leaves
every interior centered second difference of $\Phi_1$ negative on a
501-point grid in $[0,1/2]$, providing numerical evidence for concavity.

Figure~\ref{fig:tail-weight} fixes
$(t,s)=(1,0.2)$ and uses
$\nu\in\{1,2,3,4,5,6.013052641,8,10,15,20,30,50\}$ and
$\beta\in\{1,1.25,1.5,1.75,1.9,2\}$.  Finally,
Figure~\ref{fig:mechanisms} uses 361 values on $0\leq t\leq3$, with
$(r,\sigma)=(1,0.60)$ for discrete Gaussian noise and
$r\in\{1,2,4\}$ at $ar=0.5$ for discrete Laplace noise.  In
Figures~\ref{fig:common-mechanism-gaps}(a) and~\ref{fig:tail-weight}, solid
and dashed curves denote the complete gaps and the stated local-expansion
terms, respectively; the dotted line in Figure~\ref{fig:tail-weight} marks
the maximizer $\nu\approx6.013$.

\subsection{Results at an Alternative Reporting Order}
\label{app:alternative-order-robustness}

The fixed-order displays in Section~\ref{sec:practical-implications} use
$t=1$ (R\'enyi order $\alpha=2$).  Figures~\ref{fig:order-robustness-mechanisms}
and~\ref{fig:order-robustness-composition} repeat every such comparison at
$t=1.5$ ($\alpha=2.5$), with all mechanism parameters and numerical grids
unchanged.  Figure~\ref{fig:mechanisms} already spans $0\leq t\leq3$, so the
lattice examples require no duplicate fixed-order panel.

The fitted log--log slopes of the gap against $s$ are $2.000$ and $3.997$ for the
Gaussian-mixture and Laplace families, respectively, agreeing with the
quadratic and quartic local expansions. The Gaussian gap remains exactly
zero.  The Student-$t_\nu$ curve remains
nonmonotone, while the generalized Gaussian curve decreases to its exact
zero at $\beta=2$.

For the noninteger Poisson order, we evaluate the expectation in
\eqref{eq:subsampled-one-step-gap} by a 192-node Gauss--Hermite rule and
repeat it with 384 nodes; the largest CGF discrepancy is
$7.11\times10^{-15}$. The exact composition agrees with $n\Ga(1.5)$ to
$1.78\times10^{-15}$.  The allocation calculations reuse the exact Laplace
transform and the two Gaussian-mixture quadrature resolutions. On the
displayed $\rho$ grid, the minimum Laplace curvature is $1.15$ and the
maximum Gaussian-mixture curvature is $-1.41$. Thus the computed curves
retain the two allocation directions; the Gaussian-mixture concavity
check remains numerical.
\begin{figure}[!htp]
\centering
\includegraphics[width=.99\linewidth]{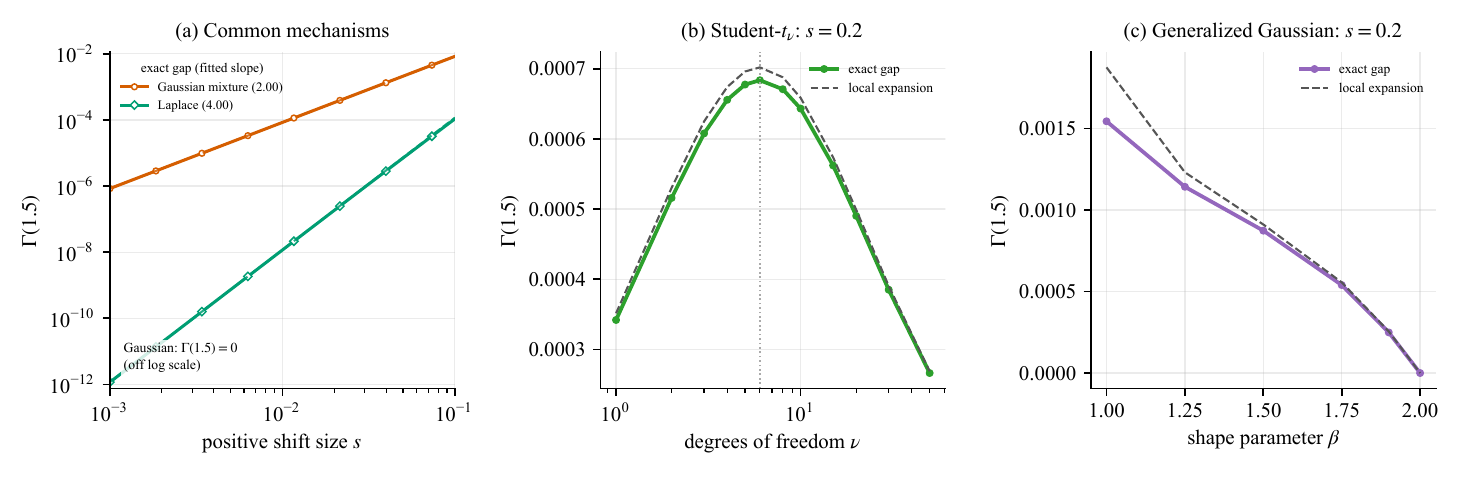}
\caption{The comparisons in Figures~\ref{fig:common-mechanism-gaps}(a)
and~\ref{fig:tail-weight} at $t=1.5$. The zero, quadratic, and quartic
local expansions and the two tail-parameter patterns seen at $t=1$
remain visible.
\textup{(a)} Gaussian, $m=1$ Gaussian-mixture, and unit-Laplace gaps over $10^{-3}\leq s\leq10^{-1}$; solid curves
are exact numerical gaps and dashed curves are the local leading terms, with
fitted slopes shown in parentheses.  \textup{(b)} Student-$t_\nu$ and
\textup{(c)} generalized Gaussian gaps at $s=0.2$, with dashed quartic
local-expansion approximations.}
\label{fig:order-robustness-mechanisms}
\end{figure}

\begin{figure}[!htp]
\centering
\includegraphics[width=.72\linewidth]{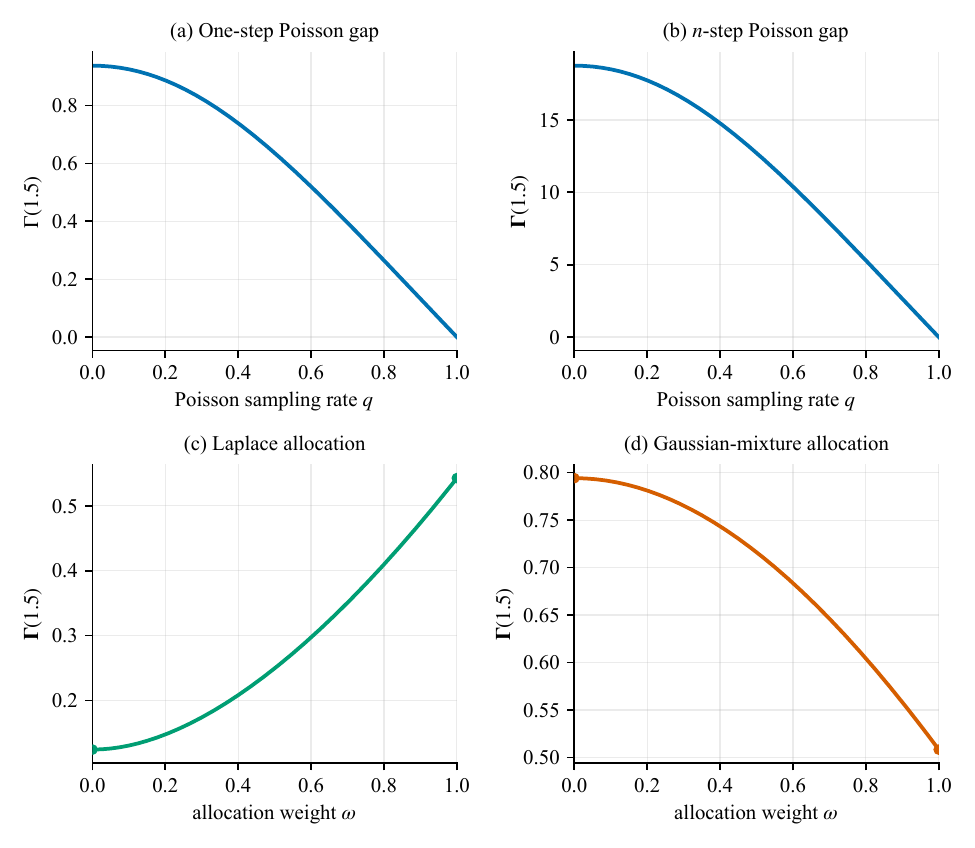}
\caption{The Poisson-subsampling and allocation comparisons in
Figures~\ref{fig:common-mechanism-gaps}(b)--(c)
and~\ref{fig:zcdp-budget-allocation} at $t=1.5$.  The Poisson gaps still decrease from their
positive-$q$ right limits to zero, while concentrating the allocation still
increases the Laplace gap and decreases the computed Gaussian-mixture gap.
\textup{(a)} One-step $\Ga(1.5)$ and \textup{(b)} $n$-step
$\bGa(1.5)=n\Ga(1.5)$ for $n=20$ and
$\rho_{\mathrm{sum}}=5$.  \textup{(c)} Unit-Laplace and \textup{(d)} $m=1$
Gaussian-mixture composition gaps along \eqref{eq:zcdp-allocation-path} for
$n=8$ and $\rho_{\mathrm{sum}}=1/2$.  The Poisson values at $q=0$ are
right-hand limits of the positive-$q$ family.}
\label{fig:order-robustness-composition}
\end{figure}

\FloatBarrier

\subsection{ACS Workloads and Additional Results}
\label{app:acs-details}

\subsubsection{Experimental Setup and Metadata}

The ACS experiment in Section~\ref{subsec:acs-composition} uses published
ACS 5-year aggregates \citep{Census2020ACSST5Y2020.S1501}. The source file
contains $3{,}221$ county and county-equivalent rows. Excluding the District
of Columbia and Puerto Rico leaves $3{,}142$ rows in the 50 states, of which
$2{,}902$ have numeric values in all five earnings columns. We rank states
by raw county count, breaking ties alphabetically by state name. The top
20 states contain $2{,}098$ rows; $1{,}961$ are complete in all five columns,
so $137$ ($6.53\%$) are excluded without imputation. 
Table~\ref{tab:acs-state-workloads} lists the selected states and their
workload sizes: $m_i$ is the number of complete counties in state $i$, and
$d_i=5m_i$ is the number of county-by-education queries.

Earnings are kept in dollars, and every component uses the ordered-pair
shift $\Delta=\$1{,}000$. At a candidate dollar scale $b$, the
standardized shift is $s=\Delta/b=1000/b$. The centered base
noises are $N(0,1)$, the equal-weight Gaussian mixture
$N(-1,1)/2+N(1,1)/2$, and unit-scale Laplace, with variances $1$, $2$, and
$2$, respectively. These choices determine the one-query CGF and optimal
zCDP parameter used in Section~\ref{subsec:acs-composition}. For each state
and family, the RDP and zCDP releases use $20{,}000$ paired draws with shared
base noise, so their MSE difference compares the two calibrated scales on
the same draws.

\begin{table}[!htbp]
\caption{The 20 ACS workloads in Section~\ref{subsec:acs-composition}.
Here $m_i$ is the number of
counties or county-equivalents with all five earnings entries numeric, and
$d_i=5m_i$ is the number of scalar queries in the composed workload. States
are ranked by total county count before incomplete rows are excluded.}
\label{tab:acs-state-workloads}
\centering
\setlength{\tabcolsep}{5.0pt}
\begin{tabular}{r lrr@{\qquad}r lrr}
\toprule
Rank & State & $m_i$ & $d_i$ & Rank & State & $m_i$ & $d_i$ \\
\midrule
1 & TX & 212 & 1060 & 11 & NE & 71 & 355 \\
2 & GA & 146 & 730  & 12 & IN & 91 & 455 \\
3 & VA & 128 & 640  & 13 & OH & 87 & 435 \\
4 & KY & 113 & 565  & 14 & MN & 86 & 430 \\
5 & MO & 114 & 570  & 15 & MI & 81 & 405 \\
6 & KS & 88  & 440  & 16 & MS & 78 & 390 \\
7 & IL & 96  & 480  & 17 & OK & 75 & 375 \\
8 & NC & 96  & 480  & 18 & AR & 73 & 365 \\
9 & IA & 97  & 485  & 19 & WI & 71 & 355 \\
10 & TN & 92 & 460  & 20 & AL & 66 & 330 \\
\bottomrule
\end{tabular}

\end{table}

For each state and family, our experiment records (available from the authors upon request) give the two calibrated
shifts and dollar scales, active orders, achieved privacy, RDP and zCDP
MSE, analytic $\Delta\mathrm{MSE}$, and Monte Carlo summaries. The earnings
values cancel from the paired squared-error difference; they define and
label the workload, while the calibrated noise determines its MSE.
\FloatBarrier

\subsubsection{Results at an Alternative Privacy Target}
\label{app:acs-alternative-privacy}

We repeat the ACS experiment at $(\eps,\delta)=(5,10^{-5})$, changing only
the privacy target.  The workloads, queries, mechanisms, and $20{,}000$
paired draws are unchanged, so Figures~\ref{fig:acs-accounting-utility-eps5}
and~\ref{fig:acs-state-violins-eps5} repeat the comparisons in
Figures~\ref{fig:acs-accounting-utility} and~\ref{fig:acs-state-violins}.
The Gaussian calibrations still coincide, and the Gaussian-mixture MSE
differences still exceed the Laplace differences. The absolute differences
are smaller at the weaker privacy target.
Table~\ref{tab:acs-eps5-summary} gives their ranges across the 20 workloads.

\begin{figure}[!htbp]
\centering
\includegraphics[width=.94\linewidth]{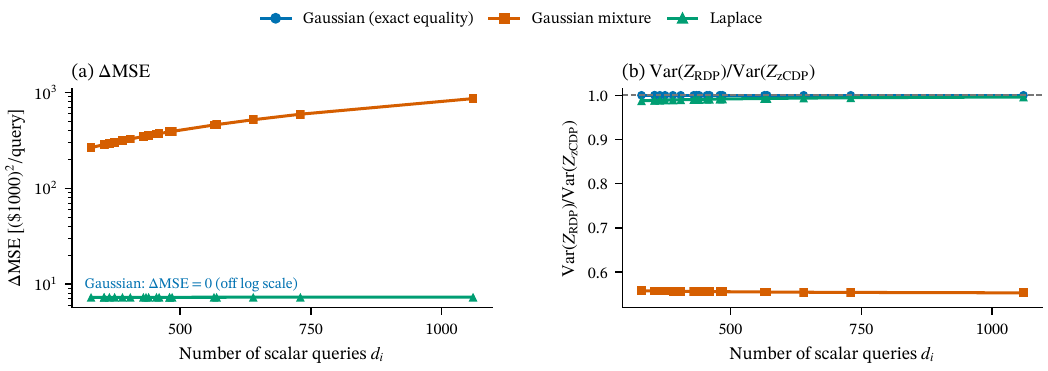}
\caption{Figure~\ref{fig:acs-accounting-utility} repeated at
$(\eps,\delta)=(5,10^{-5})$ on the same 20 ACS workloads. Left:
analytic $\Delta\mathrm{MSE}$ versus $d_i=5m_i$; Gaussian is exactly zero
and is omitted from the logarithmic axis. Right:
$\operatorname{Var}(Z_{\rm RDP})/\operatorname{Var}(Z_{\rm zCDP})$ versus
$d_i$, with Gaussian at one. Gaussian-mixture MSE differences still exceed
Laplace differences.}
\label{fig:acs-accounting-utility-eps5}
\end{figure}

\begin{figure}[!htbp]
\centering
\includegraphics[width=.94\linewidth]{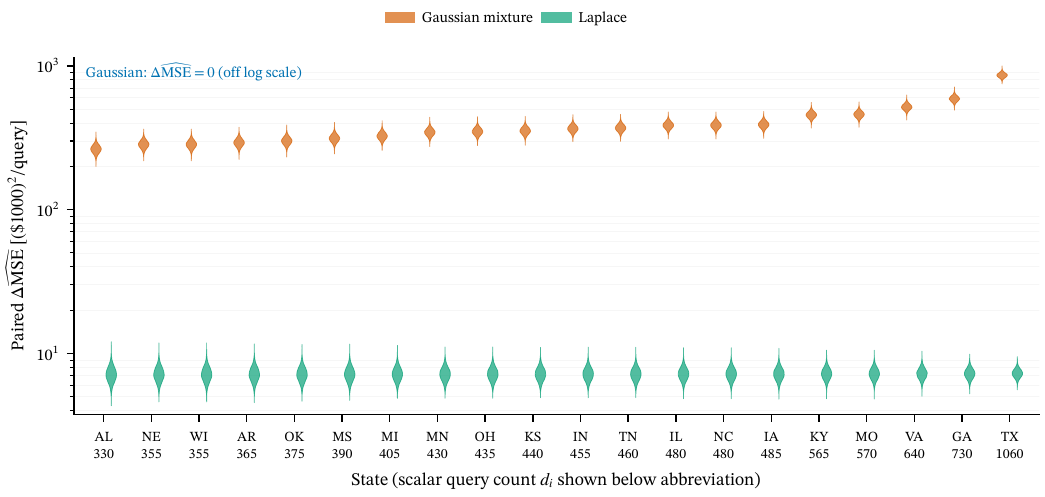}
\caption{Figure~\ref{fig:acs-state-violins} repeated at
$(\eps,\delta)=(5,10^{-5})$. States are ordered by $d_i$; the logarithmic
vertical axis gives $\Delta\widehat{\mathrm{MSE}}$ in $(\$1000)^2$ per
query. Each violin contains $20{,}000$ paired draws and describes
single-release outcomes, not confidence intervals. Gaussian is identically
zero and is omitted. Gaussian-mixture differences still exceed Laplace
differences.}
\label{fig:acs-state-violins-eps5}
\end{figure}

\begin{table}[!htbp]
\caption{ACS MSE differences and variance ratios at
$(\eps,\delta)=(5,10^{-5})$. Ranges are
over the 20 workloads in Table~\ref{tab:acs-state-workloads};
$\Delta\mathrm{MSE}$ is analytic and is measured in $(\$1000)^2$ per query.}
\label{tab:acs-eps5-summary}
\centering
\setlength{\tabcolsep}{6pt}
\begin{tabular}{lcc}
\toprule
Noise family & $\Delta\mathrm{MSE}$ range &
$\operatorname{Var}(Z_{\rm RDP})/\operatorname{Var}(Z_{\rm zCDP})$ range \\
\midrule
Gaussian             & $0$                 & $1$ \\
Gaussian mixture     & $265.0$--$860.6$   & $0.553$--$0.557$ \\
Laplace              & $7.279$--$7.313$   & $0.988$--$0.996$ \\
\bottomrule
\end{tabular}

\end{table}
\FloatBarrier

\subsection{DP-SGD Settings and Additional Results}
\label{app:controlled-dpsgd}
\label{app:reproducibility-details}

\subsubsection{Experimental Setup and Metadata}

The DP-SGD experiment in Section~\ref{subsec:controlled-dpsgd} uses the
standard Fashion-MNIST split of $60{,}000$ training and $10{,}000$ test
images \citep{xiao2017fashion}, with fixed normalization and no
augmentation. The network has two convolution and max-pooling stages,
followed by a 64-unit hidden layer, with no batch normalization. Opacus
clips each example at $C=1$ using flat clipping, adds Gaussian noise once
per training step, and divides by the expected Poisson batch size $qN$, where $N=60{,}000$.
We use SGD with initial learning rate $0.05$, no momentum, weight decay
$10^{-4}$, and a $2000$-step cosine learning-rate schedule shared by all
horizons. The runs use Opacus 1.6.0 and PyTorch 2.11.0 on an
NVIDIA A100-SXM4-80GB.

The sampling-rate sweep uses $q\in\{0.005,0.0075,0.01,0.02,0.05\}$ at
$n=500$; the composition-count sweep uses
$n\in\{250,500,1000,1500,2000\}$ at $q=0.01$. Each setting uses five
seeds. Within a setting and seed, all three methods share the
initialization and Poisson sampling stream, and the RDP and zCDP runs
share the standard-normal stream used to generate their noise. This
pairing holds at both privacy targets, $(\eps,\delta)=(3,10^{-5})$ and
$(5,10^{-5})$.

Test metrics and a fixed deterministic training probe are evaluated every
25 steps; minibatch metrics are recorded at every step. The training
histories appear in Figures~\ref{fig:dpsgd-q-main}(c)--(d)
and~\ref{fig:dpsgd-composition-sweep}(c)--(d).
Table~\ref{tab:dpsgd-full-grid} gives the nine distinct settings at
$(\eps,\delta)=(3,10^{-5})$ and their paired RDP--zCDP summaries. Mean test
accuracies of the RDP and clipping-only runs differ by at most $0.08$
percentage points, while the zCDP results vary with the calibrated noise.

\begin{table}[!htbp]
\caption{Final DP-SGD test metrics at $(\eps,\delta)=(3,10^{-5})$ for
the sweeps in Figures~\ref{fig:dpsgd-q-main}
and~\ref{fig:dpsgd-composition-sweep}. The shared
$(q,n)=(0.01,500)$ cell is listed once, and $nq$ is the expected number of
epochs.  In the upper block, accuracies are percentages and $L_Z$, $L_R$, and
$L_0$ are final test cross-entropies for the zCDP, RDP, and clipping-only
runs.  The lower block reports paired means of $\Delta L_{\rm test}=L_Z-L_R$ and $\Delta{\rm Acc}={\rm Acc}_R-{\rm Acc}_Z$ (positive values favor RDP) and 95\%
confidence intervals; $\Delta{\rm Acc}$ is in percentage points.}
\label{tab:dpsgd-full-grid}
\centering
\setlength{\tabcolsep}{2.0pt}
\begin{tabular}{ccr rrr rrr}
\toprule
$q$ & $n$ & $nq$ & ${\rm Acc}_Z$ & ${\rm Acc}_R$ & ${\rm Acc}_0$
& $L_Z$ & $L_R$ & $L_0$ \\
\midrule
$0.0050$ & $500$  & $2.5$  & $56.408$ & $65.138$ & $65.218$ & $1.673$ & $0.945$ & $0.945$ \\
$0.0075$ & $500$  & $3.75$ & $62.256$ & $65.244$ & $65.204$ & $1.141$ & $0.947$ & $0.948$ \\
$0.0100$ & $500$  & $5$    & $64.188$ & $65.304$ & $65.270$ & $1.020$ & $0.947$ & $0.946$ \\
$0.0200$ & $500$  & $10$   & $66.000$ & $65.232$ & $65.192$ & $0.945$ & $0.944$ & $0.945$ \\
$0.0500$ & $500$  & $25$   & $65.554$ & $65.304$ & $65.342$ & $0.943$ & $0.945$ & $0.945$ \\
\midrule
$0.0100$ & $250$  & $2.5$ & $57.972$ & $61.866$ & $61.928$ & $1.120$ & $1.089$ & $1.089$ \\
$0.0100$ & $1000$ & $10$  & $66.354$ & $69.502$ & $69.512$ & $1.337$ & $0.963$ & $0.962$ \\
$0.0100$ & $1500$ & $15$  & $66.144$ & $70.432$ & $70.460$ & $1.653$ & $0.979$ & $0.978$ \\
$0.0100$ & $2000$ & $20$  & $64.600$ & $70.522$ & $70.564$ & $2.031$ & $0.981$ & $0.980$ \\
\bottomrule
\end{tabular}

\vspace{1.5ex}

\begin{tabular}{cc rr}
\toprule
$q$ & $n$ & $\Delta L_{\rm test}$ & $\Delta{\rm Acc}$ \\
\midrule
$0.0050$ & $500$  & $0.728\pm0.231$  & $8.73\pm6.10$ \\
$0.0075$ & $500$  & $0.193\pm0.123$  & $2.99\pm2.60$ \\
$0.0100$ & $500$  & $0.073\pm0.050$  & $1.12\pm2.92$ \\
$0.0200$ & $500$  & $0.000\pm0.039$  & $-0.77\pm1.44$ \\
$0.0500$ & $500$  & $-0.001\pm0.006$ & $-0.25\pm0.36$ \\
\midrule
$0.0100$ & $250$  & $0.031\pm0.035$  & $3.89\pm1.77$ \\
$0.0100$ & $1000$ & $0.375\pm0.104$  & $3.15\pm1.15$ \\
$0.0100$ & $1500$ & $0.673\pm0.102$  & $4.29\pm1.31$ \\
$0.0100$ & $2000$ & $1.050\pm0.110$  & $5.92\pm1.48$ \\
\bottomrule
\end{tabular}

\end{table}

At $(\eps,\delta)=(3,10^{-5})$, all 115 scheduled trajectories
(25 clipping-only and 90 private) completed and produced all 135 endpoint
results; the $(5,10^{-5})$ target adds 90 private trajectories and reuses
the clipping-only runs. Initialization and Poisson-stream hashes
agree within each group of three methods, and standard-normal stream
hashes agree between the two private runs. At $q=0.01$, shorter horizons
are exact prefixes of the longest clipping-only trajectory.

The maximum calibration errors in $\eps$ are $8.67\times10^{-10}$ at
$\eps=3$ and $1.52\times10^{-9}$ at $\eps=5$. For the main-text target,
refining the order grid changes a noise multiplier by at most
$7.38\times10^{-7}$ in relative terms, and no optimum lies on a search
boundary. At $q=1$, calibration agrees with the Gaussian formula to a
relative error of $9.46\times10^{-10}$. These checks assess the numerical
calibration separately from the variation across training seeds.
\FloatBarrier

\subsubsection{Results at an Alternative Privacy Target}
\label{app:dpsgd-alternative-privacy}

We repeat the controlled DP-SGD experiment at
$(\eps,\delta)=(5,10^{-5})$, changing only the privacy target.  The
architecture, sweep grids, five seeds, and paired random streams remain
unchanged.  Figures~\ref{fig:dpsgd-q-sweep-eps5}
and~\ref{fig:dpsgd-composition-sweep-eps5} repeat the comparisons in
Figures~\ref{fig:dpsgd-q-main} and~\ref{fig:dpsgd-composition-sweep}.
The RDP--zCDP test-cross-entropy differences are smaller at the weaker
privacy target. They still decrease toward zero as $q$ increases and grow
with $n$.

\begin{figure}[!htbp]
\centering
\includegraphics[width=.96\linewidth]{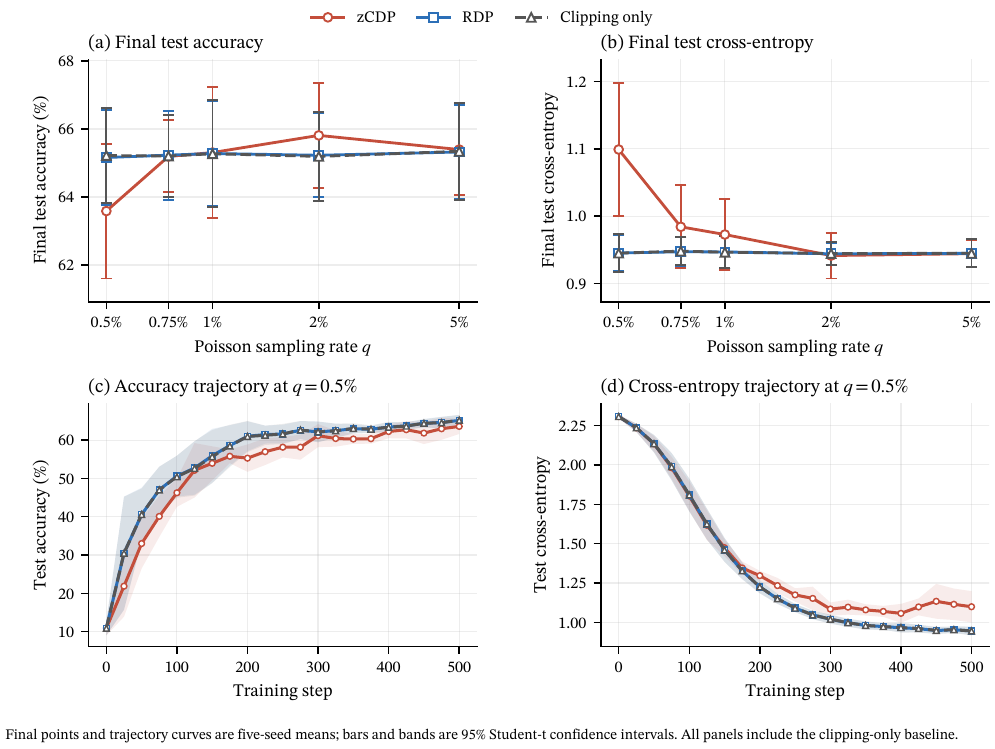}
\caption{The sampling-rate sweep in Figure~\ref{fig:dpsgd-q-main} at
$(\eps,\delta)=(5,10^{-5})$. The low-$q$ zCDP penalty is smaller than at
$(3,10^{-5})$, and the zCDP curve again approaches the RDP and
clipping-only curves as $q$ increases.  All cells use $n=500$; panels
\textup{(a)}--\textup{(b)} report final accuracy and test cross-entropy, and
panels~\textup{(c)}--\textup{(d)} report their trajectories at
$(q,n)=(0.005,500)$.  All panels include the three methods; bars and bands are
five-seed 95\% confidence intervals.}
\label{fig:dpsgd-q-sweep-eps5}
\end{figure}

\begin{figure}[!htbp]
\centering
\includegraphics[width=.96\linewidth]{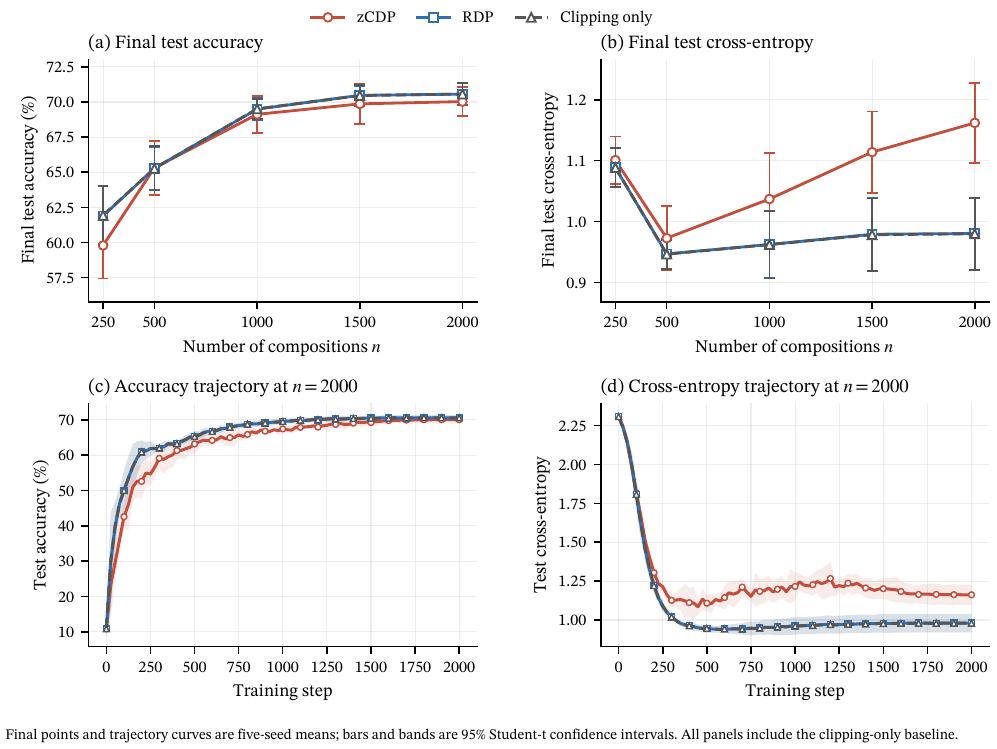}
\caption{The composition-count sweep in
Figure~\ref{fig:dpsgd-composition-sweep} at $(\eps,\delta)=(5,10^{-5})$.
The paired test-cross-entropy difference again grows with $n$, while the
accuracy difference is nonmonotone.  All cells use $q=0.01$; panels~\textup{(a)}--\textup{(b)} report
final accuracy and test cross-entropy, and panels~\textup{(c)}--\textup{(d)}
report the longest-cell trajectories at $(q,n)=(0.01,2000)$.  All panels
include the three methods.  The trajectories use the noise calibrated for
$n=2000$ throughout, so their value at training step 500 is not the
separately calibrated $n=500$ endpoint in panels~\textup{(a)}--\textup{(b)}.
Bars and bands are five-seed 95\% confidence intervals.}
\label{fig:dpsgd-composition-sweep-eps5}
\end{figure}

Table~\ref{tab:dpsgd-full-grid-eps5} gives all endpoint values.  The paired
test-loss difference decreases from $0.154$ at $q=0.005$ to approximately
zero at $q=0.05$ and increases from $0.012$ at $n=250$ to $0.182$ at
$n=2000$.  The seedwise log-$q$ and log-$n$ trends have $p=0.0107$ and $p=0.0152$,
respectively. The accuracy trends are not significant at this target
($p=0.106$ for $q$ and $p=0.361$ for $n$). Thus the test-loss pattern
persists at the weaker privacy target, without a corresponding monotone
accuracy pattern.

\begin{table}[!htbp]
\caption{Final DP-SGD test metrics at
$(\eps,\delta)=(5,10^{-5})$.  Notation and units match
Table~\ref{tab:dpsgd-full-grid}; the lower block reports paired means of $\Delta L_{\rm test}=L_Z-L_R$ and $\Delta{\rm Acc}={\rm Acc}_R-{\rm Acc}_Z$ (positive values favor RDP) and 95\%
confidence intervals.}
\label{tab:dpsgd-full-grid-eps5}
\centering
\setlength{\tabcolsep}{2.0pt}
\begin{tabular}{ccr rrr rrr}
\toprule
$q$ & $n$ & $nq$ & ${\rm Acc}_Z$ & ${\rm Acc}_R$ & ${\rm Acc}_0$
& $L_Z$ & $L_R$ & $L_0$ \\
\midrule
$0.0050$ & $500$  & $2.5$  & $63.586$ & $65.160$ & $65.218$ & $1.099$ & $0.945$ & $0.945$ \\
$0.0075$ & $500$  & $3.75$ & $65.192$ & $65.230$ & $65.204$ & $0.984$ & $0.947$ & $0.948$ \\
$0.0100$ & $500$  & $5$    & $65.306$ & $65.278$ & $65.270$ & $0.973$ & $0.947$ & $0.946$ \\
$0.0200$ & $500$  & $10$   & $65.810$ & $65.226$ & $65.192$ & $0.942$ & $0.944$ & $0.945$ \\
$0.0500$ & $500$  & $25$   & $65.394$ & $65.322$ & $65.342$ & $0.944$ & $0.945$ & $0.945$ \\
\midrule
$0.0100$ & $250$  & $2.5$ & $59.784$ & $61.888$ & $61.928$ & $1.101$ & $1.089$ & $1.089$ \\
$0.0100$ & $1000$ & $10$  & $69.120$ & $69.490$ & $69.512$ & $1.037$ & $0.962$ & $0.962$ \\
$0.0100$ & $1500$ & $15$  & $69.866$ & $70.474$ & $70.460$ & $1.114$ & $0.979$ & $0.978$ \\
$0.0100$ & $2000$ & $20$  & $70.028$ & $70.560$ & $70.564$ & $1.162$ & $0.980$ & $0.980$ \\
\bottomrule
\end{tabular}

\vspace{1.5ex}

\begin{tabular}{cc rr}
\toprule
$q$ & $n$ & $\Delta L_{\rm test}$ & $\Delta{\rm Acc}$ \\
\midrule
$0.0050$ & $500$  & $0.154\pm0.111$  & $1.57\pm2.10$ \\
$0.0075$ & $500$  & $0.037\pm0.055$  & $0.04\pm1.17$ \\
$0.0100$ & $500$  & $0.026\pm0.042$  & $-0.03\pm2.44$ \\
$0.0200$ & $500$  & $-0.003\pm0.027$ & $-0.58\pm1.03$ \\
$0.0500$ & $500$  & $0.000\pm0.005$  & $-0.07\pm0.20$ \\
\midrule
$0.0100$ & $250$  & $0.012\pm0.022$  & $2.10\pm1.12$ \\
$0.0100$ & $1000$ & $0.075\pm0.063$  & $0.37\pm0.69$ \\
$0.0100$ & $1500$ & $0.135\pm0.084$  & $0.61\pm1.67$ \\
$0.0100$ & $2000$ & $0.182\pm0.080$  & $0.53\pm1.38$ \\
\bottomrule
\end{tabular}

\end{table}
\FloatBarrier

\end{document}